\documentclass[]{article}
\usepackage{jheppub}
\usepackage{amsmath}
\usepackage{amssymb}
\usepackage{mathtools}
\usepackage{color}
\definecolor{ao}{rgb}{0.0, 0.5, 0.0}

\usepackage[caption = false]{subfig}
\usepackage{graphicx}
\usepackage{longtable}
\usepackage{tikz}
\usepackage{comment}
\usepackage{capt-of}
\usepackage{amsthm}

\usepackage{bm}

\newtheorem{remark}{Remark}
\newtheorem{theorem}{Theorem}
\newtheorem{proposition}{Proposition}
\newtheorem{corollary}{Corollary}
\newtheorem{definition}{Definition}
\newtheorem{lemma}{Lemma}

\newcommand{\pder}[2]{\ensuremath{\frac{\partial #1}{\partial #2}}}

\long\def \beq#1\eeq {\begin{equation} #1 \end{equation}}
\long\def \beaq#1\eeaq {\begin{equation}\begin{aligned} #1 \end{aligned}\end{equation}}
\long\def \bes#1\ees {\begin{equation}\begin{split} #1 \end{split} \end{equation}}
\long\def \bea#1\eea {\begin{eqnarray} #1 \end{eqnarray}}
\long\def \bse[#1]#2\ese {\begin{subequations}\label{#1}\begin{align} #2 \end{align}\end{subequations}}

\newcommand{\sums}{\sum_{ \boldsymbol \sigma}}
\newcommand{\sumi}{\sum_{i=1}^N}
\newcommand{\suma}{\sum_{a=1}^2}
\newcommand{\sumanew}{\sum_{a=1}^3}
\newcommand{\sumij}{\sum_{i,j=1}^N}
\newcommand{\si}{\sigma_i}
\newcommand{\sj}{\sigma_j}
\newcommand{\dt}{\frac{\partial}{\partial t}}
\newcommand{\dx}{\frac{\partial}{\partial x}}
\newcommand{\dw}{\frac{\partial}{\partial w}}
\newcommand{\qb}{\bar{q}}
\newcommand{\mb}{\bar{m}}

\graphicspath{{Plots/}} 
\title{Replica symmetry breaking in neural networks: \\ a few steps toward rigorous results.}
\author[a,b]{Elena Agliari,}
\author[c]{Linda Albanese,}
\author[b,c,d]{Adriano Barra,}
\author[e]{Gabriele Ottaviani}

\affiliation[a]{Dipartimento di Matematica {\em Guido Castelnuovo}, Sapienza Universit\`a  di Roma, Roma, Italy}
\affiliation[b]{Istituto Nazionale d'Alta Matematica {\em Francesco Severi},  Roma, Italy}
\affiliation[c]{Dipartimento di Matematica e Fisica {\em Ennio De Giorgi},  Universit\`a  del Salento,  Lecce, Italy}
\affiliation[d]{Istituto Nazionale di Fisica Nucleare, Campus Ecotekne, Lecce,Italy}
\affiliation[e]{Dipartimento di Fisica, Sapienza Universit\`a  di Roma, Roma, Italy}

\abstract{In this paper we adapt the {\em broken replica interpolation} technique (developed by Francesco Guerra to deal with the Sherrington-Kirkpatrick model, namely a pairwise mean-field spin-glass whose couplings are i.i.d. standard Gaussian variables) in order to work also with the Hopfield model (i.e., a pairwise mean-field neural-network  whose couplings are drawn according to Hebb's learning rule): this is accomplished by grafting Guerra's telescopic averages on the transport equation technique, recently developed by some of the Authors.
\newline
As an overture, we apply the technique to solve the Sherrington-Kirkpatrick model with i.i.d. Gaussian couplings centered at $J_0$ and with finite variance $J$; the mean $J_0$ plays the role of a {\em signal} to be detected in a noisy environment tuned by $J$, hence making this model a natural test-case to be investigated before addressing the Hopfield model.
\newline
For both the models, an explicit expression of their quenched free energy in terms of their natural order parameters is obtained at the $K$-th step ($K$ arbitrary, but finite) of replica-symmetry-breaking. In particular, for the Hopfield model, by assuming that the overlaps respect Parisi's decomposition (in particular following the \emph{ziqqurat ansatz}) and that the Mattis magnetization is self-averaging, we recover previous results obtained via replica-trick by Amit, Crisanti and Gutfreund (1RSB) and by Steffan and K\"{u}hn (2RSB).
}

\begin{document}

\maketitle

\section{Introduction}

Since the $80$'s, statistical mechanics of spin glasses has been playing a pivotal role in the neural network investigations, both in the learning stage (where these systems are properly trained to accomplish specific tasks) \cite{angel-learning,sompo-learning} and in the operational stage (where these systems perform pattern recognition, classification, etc.) \cite{Amit,Coolen}.
However, beyond countless successes (see e.g. \cite{DL1}), there is still a long way to go before claiming that we do have a 
{\em theory} for Artificial Intelligence, not only as a whole, but -- more specifically for the aim of the present paper -- even restricting to neural networks meant as statistical-mechanics systems displaying  information processing skills as emergent, collective properties. 
\newline
Among the main hurdles to overcome we mention the usage of semi-heuristic techniques (e.g., the well-known replica trick \cite{Amit,Coolen}) and the problem of the stability of the replica symmetric (RS) solution \cite{Crisanti,Dotsenko2,Kuhn}.
As for the former, rigorous alternatives grounded on probabilistic techniques  \cite{ABT,Alemannation1,jean1,Barra-JSP2010,Tirozzi,DmitryBook,Pastur} and on PDE approaches driven from mathematical physics \cite{Agliari-Barattolo,Alemannation1,jean2,Albert1,Martino2,Murrat1,MurratPanchenko} have been developed since the seminal papers by Bovier \cite{Bovier1,Bovier2,Bovier3} and Talagrand \cite{Tala1,Tala2}. As for the latter, many efforts are still in order, also at the conceptual level, as we plan to report soon.
\newline
In this paper we aim to develop adequate mathematical techniques to possibly address the above-mentioned issues, and, to this goal, we adapt the broken replica interpolation introduced by Francesco Guerra for the Sherrington-Kirkpatrick model \cite{Guerra} (i.e., a pairwise mean-field spin-glass usually playing as the ``harmonic oscillator'' for complex systems \cite{Coolen}) into a PDE-framework recently developed by some of the Authors \cite{AABF-NN2020} \footnote{It is worth noticing that this PDE-framework stems from the Hamilton-Jacobi approach, another mathematical approach formulated by Francesco Guerra \cite{BGDiBiasio,GuerraSum})}.
\newline
Before applying this technique to neural networks, we will address a relatively simpler model, that is the Sherrington-Kirkpatrick model with a ferromagnetic contribution, which still requires a complete set of order parameters for its investigation, namely two-replica overlaps and magnetization. Then, we move to the Hopfield model, i.e. the ``harmonic oscillator'' for associative neural-networks, able to perform spontaneous pattern recognition.
For these models, we drop the common simplifying assumption of self-average of the overlap (i.e., the {\em replica symmetric} scenario) and we allow the model to undergo {\em K steps of replica symmetry breaking} (K-RSB),  yet keeping the magnetization self-averaging. In this way, for the Hopfield model, at the first step of RSB, we obtain an explicit expression for the quenched free-energy that perfectly recovers the 1-RSB free-energy expression obtained via replica-trick by Crisanti, Amit and Sompolinsky nearly 35 years ago \cite{Crisanti} and, at the second step, we recover the 2-RSB expression achieved by Steffan and K\"uhn \cite{Kuhn} a few years later.

We stress that in our analysis we rely upon the duality between Hopfield networks and (restricted) Boltzmann machines (originally discussed in \cite{BarraEquivalenceRBMeAHN} and then enlarged in \cite{Agliari-Dantoni,Barra-RBMsPriors1,Barra-RBMsPriors2,Mezard,Monasson}), in such a way that our results hold also for Boltzmann machines, namely the basic architecture for machine learning \cite{DL1}. More precisely, the Boltzmann machine counterpart of the standard Hopfield model is made of a digital layer equipped with Boolean neurons (i.e., Ising spins) and an analog layer made of Gaussian real-valued neurons. As a result, this model can also be seen as a linear combination of two spin-glasses \cite{bipartiti}, the former is a standard hard Sherrington-Kirkpatrick model with Boolean spins (well known to be full-RSB \cite{Guerra,TalaParisi}), the latter is a soft Sherrington-Kirkpatrick model with Gaussian spins (known to be RS \cite{Gauss-1,Gauss-2}). From this perspective a natural question we answer is: when the Boolean neurons undergo RSB does this phenomenon propagate also to the soft neurons despite they usually behave in an RS way? The answer is positive, that is, as the binary component breaks replica symmetry, also the overlap between two replicas of the analog layer aquires a broken replica step. Remarkably, the overlaps for both components (i.e., the overlap related to the Ising spins and the one related to the Gaussian spins) break replica symmetry ``simultaneously'', namely at the same value of $m$ in the Parisi scheme, in agreement with the {\em ziggurat ansatz}  that has been rigorously developed to generalize the Parisi scheme to multi-species Boolean spin glasses \cite{ZiqquratBarra,ZiqquratPanchenko}. 

Admittedly, there are plenty of open questions and discordances regarding the role of RSB within neural networks and its impact on the critical capacity, however, in the current work, we discuss solely the mathematical aspects, supplying the need of a robust and rigorous framework where neural network models can be addressed also allowing for an RSB scenario, but we will not discuss the underlying physics, apart a short remark in the conclusions. We do believe that, for neural networks, RSB should not be seen as a {\em perturbation} of the Amit-Gutfreund-Sompolinksy (AGS) RS painting and, possibly, even the starting assumptions should be revised (in order for AGS theory to be recovered as a proper limit of a broader theory); we will report on our findings in future papers but here, as stated above, we simply prove a novel and rigorous method to recover and extend the existing results achieved under given ansatz and via the replica trick.

The present manuscript is structured as follows:
\newline
The next Sec.~\ref{sec:SK} is dedicated to the Sherrington-Kirkpatrick model where couplings have positive mean and it is split into four subsections: Sec.~\ref{ssec:SKRS} is dedicated to the RS derivation of the transport PDE  that we use as the mathematical backbone, Sec.~\ref{ssec:SKRSB} to perform  with this PDE approach a first step of RSB, Sec.~\ref{ssec:SK2RSB} to accomplish the second step, and Sec.~\ref{ssec:SKKRSB} to give the general expression for arbitrary, but finite, K steps of RSB.
Then, in Sec.~\ref{HopfieldSection} we move to the Hopfield model and, again, the section is split into four subsections mirroring those of the previous section: Sec.~\ref{ssec:HRS} summarizes the RS scenario achieved via the transport PDE (already presented in \cite{AABF-NN2020}), Sec.~\ref{ssec:HRSB1} enlarges the scheme to the first step of RSB (hence recovering the expression for the quenched free energy already found by Crisanti, Amit and Gutfreund \cite{Crisanti}), Sec.~\ref{ssec:HRSB2} enlarges the scheme to the second step (hence recovering the expression for the quenched free energy already found by the Steffan and K\"{u}hn \cite{Kuhn}), and in Sec.~\ref{ssec:HOP_KRSB} a general prescription for the K-RSB scenario is formulated (again for arbitrary, but finite, values of $K$). The final Sec.~\ref{conclusions} containes some comments on the problems arising when using the original self-averaging ansatz on the Mattis magnetization within an RSB scheme and a general outlook. Lengthy calculations are reported in Appendices \ref{app1}-\ref{1body2rsb} for the sake of completeness.

\section{Prelude: the Sherrington-Kirkpatrick model with a signal} \label{sec:SK}
We consider a Sherrington-Kirkpatrick model where pairwise couplings among spins display a non-null mean $J_0$ and a variance $J^2$; this model was already treated in \cite{Coolen} as a pedagogical introduction to the Hopfield model.
Here, our aim is to get an expression for the quenched free-energy of this model in a rigorous way, via generalized Guerra's interpolating technique; first, in subsec.~\ref{ssec:SKRS}, we will focus on the RS scenario to get acquainted with the method and then, in subsecs.~\ref{ssec:SKRSB}-\ref{ssec:SKKRSB}, we will address the first, the second and the $K$-th step of RSB, respectively. 

\begin{definition}
Let $\boldsymbol \sigma \in \{ -1, +1\}^N$ be a configuration of $N$ spins, the Hamiltonian of the Sherrington-Kirkpatrick model with a signal is defined as
\begin{equation}\label{DefSKsignal}
H_N(\boldsymbol \sigma | \boldsymbol J) \coloneqq- \frac{1}{2} \sum_{\substack{i,j=1 \\ i\neq j}}^{N,N}J_{ij}\sigma_i \sigma_j,
\end{equation}
where the pairwise quenched couplings $\boldsymbol J = \{ J_{ij}\}_{i,j=1,..,N} \in \mathbb R^{N\times N}$ are given by
\begin{equation}
\label{eq:coupling}
J_{ij}\coloneqq\frac{J_0}{N}+\frac{J \sqrt{2} z_{ij}}{\sqrt{N}},
\end{equation}
with $J_0 \in \mathbb R^+$, $z_{ij}$ i.i.d. standard random variables drawn from $P(z_{ij}) = \mathcal{N}[0,1]$ for $i<j=1,...N$ and $z_{ij}=z_{ji}$.
\end{definition}
\begin{definition}
The partition function related to the Hamiltonian (\ref{DefSKsignal}) is given by
\begin{equation} \label{PF1}
Z_N(\beta , \boldsymbol J) \coloneqq \sums e^{-\beta H_N(\boldsymbol \sigma | \boldsymbol J)},
\end{equation}
where $\beta \in \mathbb{R}^+$ is the inverse temperature in proper units such that for $\beta \to 0$ the probability distribution for the spin configuration is uniformly spread while for $\beta \to \infty$ it is sharply peaked at the minima of the energy function (\ref{DefSKsignal}).
\end{definition}
Once defined the Hamiltonian (\ref{DefSKsignal}) and the partition function (\ref{PF1}), we can introduce the \emph{Boltzmann average} denoted with $\omega_{\boldsymbol J}(.)$, which, for the generic observable $O(\boldsymbol \sigma )$, reads as
\begin{eqnarray}
\omega_{\boldsymbol J}( O ( \boldsymbol \sigma)) \coloneqq \frac{\sums O(\boldsymbol \sigma) e^{-\beta H_N(\boldsymbol \sigma| \boldsymbol J)}}{Z_N(\beta , \boldsymbol J)}.
\end{eqnarray}
This can be further averaged over the realization of the $J_{ij}$'s (also referred to as \emph{quenched average}), to get
\begin{equation}
\langle O (\boldsymbol \sigma) \rangle \coloneqq \mathbb{E} [ \omega_{\boldsymbol J}(O(\boldsymbol \sigma )) ],
\end{equation}
where the operator $\mathbb{E}$ shall be used in the following to denote, more generally, expectation on quenched quantities.
\newline
Further, we introduce the product state $\Omega_{s, \boldsymbol J} = \omega_{\boldsymbol J}^{(1)} \times \omega_{\boldsymbol J}^{(2)} \times ... \times \omega_{\boldsymbol J}^{(s)}$ over $s$ replicas of the system, characterized by the same realization $\boldsymbol{J}$ of disorder. In the following, we shall use the product state over two replicas only, hence we shall neglect the index $s$ without ambiguity; also, to lighten the notation, we shall omit the subscript $\boldsymbol J$ in $\omega_{\boldsymbol J}$ and in $\Omega_{\boldsymbol J}$. Thus, for an arbitrary observable $O(\boldsymbol \sigma^{(1)}, \boldsymbol \sigma^{(2)})$
	\beq
	\langle O (\boldsymbol \sigma^{(1)}, \boldsymbol \sigma^{(2)}) \rangle \coloneqq \mathbb{E} \Omega (O(\boldsymbol \sigma^{(1)}, \boldsymbol \sigma^{(2)})) = \mathbb{E}  \frac{\sum_{\boldsymbol \sigma} O(\boldsymbol \sigma^{(1)},\boldsymbol \sigma^{(2)}) e^{-\beta [H_N(\boldsymbol \sigma^{(1)}| \boldsymbol J ) + H_N(\boldsymbol \sigma^{(2)}| \boldsymbol J)] }}{Z_N^2(\beta, \boldsymbol J)},
	\eeq
	where $\boldsymbol{\sigma}^{(1,2)}$ is the configuration pertaining to the replica labelled as $1,2$.

\begin{remark}
Notice the different normalization in the definition (\ref{eq:coupling}): the ``signal'' $J_0$ is normalized by $N$, while the quenched ``noise'' $J$ is normalized by $\sqrt{N}$. This  ensures the linear extensivity of all the thermodynamic observables related to the model, e.g. the energy must scale as $\langle H_N(\boldsymbol \sigma| \boldsymbol J) \rangle \sim \mathcal O(N^1)$.
\end{remark}
\begin{definition}
The intensive quenched pressure of the Sherrington-Kirkpatrick model with a signal (\ref{DefSKsignal}) reads as
\begin{eqnarray}\label{PressureDef}
A_N (\beta,J_0,J) \coloneqq \frac{1}{N}\mathbb{E}\log Z_N(\beta, \boldsymbol J),
\end{eqnarray}
and its thermodynamic limit reads as
\begin{equation}
\label{PressureDefLTD}
A(\beta,J_0,J) \coloneqq \lim_{N \to \infty}A_N (\beta,J_0,J).
\end{equation}
\end{definition}
We recall that the pressure $A_N (\beta,J_0,J)$ corresponds, a constant $-\beta$ apart, to the free-energy of the model and that we omit the subscript ``$N$'' when the quantity is evaluated at infinite size.\\
In order to solve the model we want to find out an explicit expression for the quenched pressure (\ref{PressureDefLTD}) in terms of the natural order parameters of the theory, namely the magnetization $m$ and the two-replica overlap $q_{12}$, defined in the following
\begin{definition}
The order parameters used to describe the macroscopic behavior of the model are the standard ones \cite{MPV,Coolen}, namely the magnetization $m$ and the two-replica overlap $q_{12}$, introduced as
\begin{eqnarray}
m &\coloneqq& \frac{1}{N}\sum_{i=1}^{N} \sigma_i, \\
q_{12} &\coloneqq& \frac{1}{N}\sum_{i=1}^{N} \sigma_i^{(1)} \sigma_i^{(2)}.
\end{eqnarray}
\end{definition}
\begin{remark}
We comment on the appellation ``Sherrington-Kirkpatrick model with a signal'': having introduced a positive mean for the couplings, an ordinary magnetization $m$ is also required among the order parameters; having in mind an associative neural network (vide infra), if we introduce a single pattern $\boldsymbol \xi$ (i.e., a vector of $N$ binary entries $\xi_i = \pm 1,\ \ i \in (1,...,N)$), by a Mattis gauge $\sigma_i \to \xi_i \sigma_i$ the standard magnetization turns into the Mattis magnetization, that is the order parameter used to quantify the retrieval of the considered pattern in neural network's theory, yet the Hamiltonian remains invariant under this transformation.
\end{remark}

\subsection{Replica Symmetric Interpolation: RS solution} \label{ssec:SKRS}
In order to get familiar with Guerra's interpolation scheme, it is useful to first address the model (\ref{DefSKsignal}) by assuming the self-averaging of both $m$ and $q_{12}$, namely the so-called replica symmetric scenario.
\begin{definition} \label{def:SK_RS}
Under the replica-symmetry assumption, the order parameters, in the thermodynamic limit, self-average and their distributions get delta-peaked at their equilibrium value (denoted with a bar), independently of the replicas considered, namely
\begin{align}
&\lim_{N \rightarrow + \infty} \langle (m - \bar{m})^2 \rangle =0 \Rightarrow  \lim_{N \rightarrow + \infty}  \langle m \rangle = \mb \label{limform},\\
&\lim_{N \rightarrow + \infty} \langle (q_{12} - \bar{q})^2 \rangle =0 \Rightarrow  \lim_{N \rightarrow + \infty}  \langle q_{12} \rangle = \bar{q} \label{limforq}.
\end{align}
\end{definition}
The technique exploited to solve the model is based on a suitable interpolating partition function $\mathcal Z_N$, whence an interpolating quenched pressure $\mathcal A_N$, which we can solve for and which recovers the quenched pressure $A_N$ of the original model for a suitable choice of the interpolating parameters; hereafter, when dealing with interpolating quantities we shall omit the dependence on $\boldsymbol J, J, J_0, \beta$ to lighten the notation.

\begin{definition}
Given the interpolating parameters $\boldsymbol r := (x, w) \in \mathbb R^2$ and $t \in \mathbb{R}^+$, the interpolating partition function is defined as
\begin{eqnarray}
\mathcal Z_N(t, \boldsymbol r) \coloneqq \sums \exp \left[ \beta \left ( \sqrt{t}\frac{J \sqrt{2}}{2\sqrt{N}} \sumij \si \sj z_{ij} + \sqrt{x}\sumi z_i \si + t\frac{J_0}{2}Nm^2(\boldsymbol{\sigma}) +w J_0 N m(\boldsymbol{\sigma}) \right )\right],
\end{eqnarray}
where $z_i \sim \mathcal N[0,1]$, for $i=1,...,N$.
\end{definition}
\begin{definition}
The interpolating pressure, at finite volume $N$, is introduced as
\begin{align}\label{Ainterpolata}
\mathcal A_N(t,  \boldsymbol r) \coloneqq \frac{1}{N}\mathbb{E}\left[\log \mathcal Z_N(t, \boldsymbol r)\right],
\end{align}
and, in the thermodynamic limit,
\begin{align}\label{AinterpolataN}
\mathcal A(t,  \boldsymbol r) \coloneqq \lim_{N \to \infty} \mathcal A_N(t,  \boldsymbol r).
\end{align}
whose esistence is guaranteed by the Guerra-Toninelli theorem  \cite{GuerraTon}.
By setting $t=1, x=0,w=0$, the interpolating pressure recovers the standard pressure (\ref{PressureDef}), that is, $A_N(\beta, J_0, J) = \mathcal A_N (t=1, \boldsymbol r = 0)$.
\end{definition}
\begin{remark}
The interpolating structure implies an interpolating measure whose related Boltzmann factor reads as
\begin{eqnarray}
\mathcal B (\boldsymbol \sigma; t, \boldsymbol r) &\coloneqq& \exp \left[ \beta \mathcal H (\boldsymbol \sigma; t, \boldsymbol r) \right],\\
 \mathcal H (\boldsymbol \sigma; t, \boldsymbol r) &\coloneqq&  \sqrt{t}\frac{J \sqrt{2}}{2\sqrt{N}} \sumij \si \sj z_{ij} + \sqrt{x}\sumi z_i \si + t\frac{J_0}{2}Nm^2(\boldsymbol{\sigma}) +w J_0 N m(\boldsymbol{\sigma}).
\end{eqnarray}
Clearly, $\mathcal Z_N(t, \boldsymbol r) = \sum_{\boldsymbol \sigma} \mathcal B (\boldsymbol \sigma; t, \boldsymbol r)$.\\
A generalized average follows from this generalized measure as
\beq
	\omega_{t, \boldsymbol r} (O (\boldsymbol \sigma )) \coloneqq  \sum_{\boldsymbol \sigma} O (\boldsymbol \sigma ) \mathcal B (\boldsymbol \sigma; t, \boldsymbol r)
	\eeq
	and
\beq
\langle O (\boldsymbol \sigma ) \rangle_{t, \boldsymbol r}  \coloneqq \mathbb E [ \omega_{t, \boldsymbol r} (O (\boldsymbol \sigma )) ],
\eeq
where $ \mathbb E$ denotes the average over $\boldsymbol J$ and $\{ z_i \}_{i=1,...,N}$.
Of course, when $t=1$ and $\boldsymbol r = 0$, the standard Boltzmann measure and related average are recovered.\\
Hereafter, in order to lighten the notation, we will drop the subscripts $t, \boldsymbol r$.
\end{remark}
The strategy is now to interpret the interpolating parameters $(t, \bm r)$ as {\em space-time} fictitious variables and to show that the interpolating pressure $\mathcal A_N(t, \bm r)$ obeys a standard transport equation in this space-time framework; then, by solving such a PDE and evaluating its solution for $\bm r = (0,0), t=1$ we will have the solution of the original problem as well.
We proceed by computing the first order derivatives with respect to each parameter resulting in the next
\begin{lemma}
The partial derivatives of the interpolating quenched pressure (\ref{Ainterpolata}) read as
\begin{align}
\label{eq:prima}
\dt \mathcal A_N &= \frac{\beta^2}{4}J^2 (1-\langle q_{12}^2 \rangle) + \frac{\beta J_0}{2}\langle m^2 \rangle, \\
\label{eq:seconda}
\dx \mathcal A_N &= \frac{\beta^2}{2}(1-\langle q_{12} \rangle), \\
\label{eq:terza}
\dw \mathcal A_N &= \beta J_0 \langle m \rangle.
\end{align}
\end{lemma}
\begin{proof}
We prove only (\ref{eq:prima}), namely the result related to the time derivative, which is the most tricky among the three, for the others the computation is analogous.
\begin{align}
\dt \mathcal A_N &= \frac{1}{N} \mathbb{E} \left( \frac{1}{ \mathcal Z_N}\dt \mathcal Z_N \right )
=\frac{1}{N} \mathbb{E} \left [ \frac{1}{\mathcal Z_N}\sums \left(\frac{\beta J \sqrt{2}}{2\sqrt{tN}}\sumij \si \sj z_{ij} +\beta J_0 N m^2 \right) \mathcal B (\boldsymbol \sigma; t, \boldsymbol r) \right ].
\end{align}
Now we use the fact that, for a standard Gaussian variable $z$, i.e. $z \sim \mathcal N(0,1)$, and for a generic function $f(z)$ which goes to zero fast enough, by Wick's theorem
\begin{equation}
\label{eqn:gaussianrelation2}
\mathbb{E}_z z f(z) = \mathbb{E}_z \partial_z f(z),
\end{equation}
where $\mathbb{E}_z$ represents the average over $z$.
As a consequence of (\ref{eqn:gaussianrelation2}), it is possible to write
\begin{align}
\dt A_N=&\frac{\beta J \sqrt{2}}{2N\sqrt{tN}} \mathbb{E} \left \{ \sumij \partial_{z_{ij}} \left[ \frac{1}{\mathcal Z_N}\sums  \si \sj ~\mathcal B (\boldsymbol \sigma; t, \boldsymbol r) \right]\right \} + \frac{\beta J_0}{2} \langle m^2 \rangle = \notag \\
=&\frac{\beta^2 J^2}{4}\mathbb{E} \left \{ \sumij \left[ 1-\frac{1}{\mathcal Z_N^2}\left(\sums  \si \sj ~ \mathcal B (\boldsymbol \sigma; t, \boldsymbol r) \right)^2 \right] \right\} + \beta\frac{J_0}{2}\langle m^2 \rangle = \notag \\
=& \frac{\beta^2 J^2}{4}(1-\langle q_{12} \rangle) - \beta\frac{J_0}{2} \langle m^2 \rangle.
\end{align}
\end{proof}
Our target now is to find a PDE for $\mathcal A_N(t, \boldsymbol r)$ in the form
\begin{align}
\label{eq:trans0}
\frac{d \mathcal A_N}{dt}&=\dt \mathcal A_N+\dot x \dx \mathcal A_N +\dot w \dw \mathcal A_N = S(t, \boldsymbol r) + V_N(t, \boldsymbol r),
\end{align}
where we denote with $\dot{x}$ and $\dot{w}$ the time derivative of, respectively, $x$ and $w$, and we introduced a ``potential'' $V_N(t, \boldsymbol r)$ and a ``source'' $S(t,  \boldsymbol r)$, whose explicit expressions will be deepened later. Remarkably, eq.~(\ref{eq:trans0}) displays the structure of a transport equation.
%
%
%
\begin{proposition}
The streaming of the interpolating quenched pressure obeys, at finite volume $N$, a standard transport equation, that reads as
\begin{align}
\label{eq:trans}
\frac{d \mathcal A_N}{dt}&=\dt \mathcal A_N+\dot x \dx \mathcal A_N +\dot w \dw \mathcal A_N = S(t, \bm r) + V_N(t, \bm r),
\end{align}
where
\begin{align}
S(t,  \boldsymbol r) &\coloneqq \frac{\beta^2 J^2}{4}(1-\qb)^2 - \frac{\beta}{2}J_0 \mb^2,\\ \label{potenzialeRS-SK}
V_N(t, \boldsymbol r) &\coloneqq \frac{\beta^2}{4}J^2 \langle (q_{12}-\qb)^2 \rangle + \frac{\beta}{2}J_0 \langle (m-\mb)^2\rangle.
\end{align}
\end{proposition}
\begin{proof}
Recalling the equilibrium values introduced in (\ref{limform}) and (\ref{limforq}), we can write
\begin{align}
\langle (m-\mb)^2 \rangle &= \langle m^2 \rangle + \mb^2 -2\mb\langle m \rangle, \\
\langle (q_{12}-\qb)^2 \rangle &= \langle q_{12}^2 \rangle + \qb^2 -2\qb\langle q_{12} \rangle.
\end{align}
These relations are used while handling the expression (\ref{eq:prima}) of the $t$-derivative of the interpolating quenched pressure to get
\begin{align}
\dt \mathcal A_N &= \frac{\beta^2}{4}J^2 - \frac{\beta^2}{4}J^2 \langle q_{12}^2 \rangle + \frac{\beta}{2}J_0 \langle m^2 \rangle = \notag \\
&=  \frac{\beta^2}{4}J^2 - \frac{\beta^2}{4}J^2 [\langle (q_{12}-\qb)^2 \rangle - \qb^2 + 2\qb\langle q_{12} \rangle] +\frac{\beta}{2}J_0 [\langle (m-\mb)^2 \rangle - \mb^2 + 2\mb\langle m \rangle]= \notag \\
&= \qb J^2 \dx \mathcal{A}_N + \mb \dw \mathcal A_N +\frac{\beta^2 J^2}{4}(1-\qb)^2 - \frac{\beta}{2}J_0 \mb^2 - \frac{\beta^2}{4}J^2 \langle (q_{12}-\qb)^2 \rangle + \frac{\beta}{2}J_0 \langle (m-\mb)^2\rangle = \notag \\
\label{eq:SV}
&=  \qb J^2 \dx \mathcal{A}_N + \mb \dw \mathcal A_N  + S(t,  \boldsymbol r) + V_N(t,  \boldsymbol r).
\end{align}
Since we have the freedom to set $\dot{x}$ and $\dot w$, by choosing $\dot{x}=- \qb J^2$ and $\dot w = - \mb$, we finally obtain (\ref{eq:trans}).
\end{proof}
\begin{remark}
In the thermodynamic limit, under the RS assumption (\ref{def:SK_RS}), we have $\langle (m-\mb)^2 \rangle= 0$ and $\langle (q_{12}-\qb)^2 \rangle=0$, in such a way that the potential in (\ref{eq:trans}) is vanishing, that is
\begin{equation} \label{eq:V0}
\lim_{N \to \infty} V_N(t,  \boldsymbol r)=0.
\end{equation}
\end{remark}
Exploiting the last remark we can prove the following
\begin{proposition}
The transport equation associated to the interpolating pressure function $\mathcal A_N(t, \boldsymbol r)$ in the thermodynamic limit and under the RS assumption is
\begin{align}
\label{eq:propRS}
\dt \mathcal A_{\textrm{RS}} - \qb J^2\dx \mathcal A_{\textrm{RS}} - \mb \dw \mathcal A_{\textrm{RS}} = \frac{\beta^2}{4}J^2 (1-\qb)^2 -\frac{\beta J_0}{2}\mb^2,
\end{align}
\label{propRS}
whose solution is given by
\begin{align}
\label{solutionRS}
\mathcal A_{\textrm{RS}}(t,  \boldsymbol r)=\mathbb{E} \left[\log 2\cosh \left( \beta z  \sqrt{x + \qb J^2}  + \beta  J_0 (w + \mb t) \right) \right]+ \frac{\beta^2}{4}J^2(1-\qb)^2 t - \beta\frac{J_0}{2}\mb^2 t.
\end{align}
\end{proposition}
\begin{proof}
The PDE in (\ref{eq:propRS}) can be obtained straightforwardly from (\ref{eq:SV}) by using (\ref{eq:V0}).
The resulting equation can be solved through the method of characteristics as
\begin{align}
\mathcal A_{\textrm{RS}}(t, \boldsymbol r)= \mathcal  A_{\textrm{RS}} (0, \bm r - \dot {\bm r} t,  t)+ S( t, \boldsymbol r)t
\label{traspparametric}
\end{align}
where $\dot{\bm r} = (\dot x, \dot w)$ and the characteristics are
\begin{align}
\dot x &= - \qb J^2, \notag \\
\dot w &=  - \mb.
\end{align}
Along the characteristics, the fictitious motion in the $(t, \boldsymbol r)$ time-space is linear and returns
\begin{align}
x= x_0 -\qb J^2 t \notag \\
w=w_0 - \mb t,
\end{align}
where $\boldsymbol {r_0} = (x_0, w_0) = (x(t=0), w(t=0))$.
%
The Cauchy condition at $t=0$ is given by a direct computation at finite $N$ as
\begin{align}
&\mathcal A_{RS}(0, \boldsymbol {r_0})=\mathcal  A_{\textrm{RS}} (0, \bm r - \dot {\bm r} ~ t)=\frac{1}{N} \mathbb{E} \log \left[ \sums \exp \left (\beta \sqrt{x_0} \sumi z_i \si + \beta w_0 J_0\sumi \sigma_i \right) \right] = \notag \\
&=\frac{1}{N} \mathbb{E} \prod_{i=1}^N \sums \left[ \exp \left(\beta \sqrt{x_0} z_i + \beta w_0 J_0 \right)\si \right] = \mathbb{E} \log 2\cosh \left( \beta \sqrt{x_0} z + \beta w_0 J_0 \right).
\label{chauchy}
\end{align}
Now, merging (\ref{traspparametric})-(\ref{chauchy}) we get (\ref{solutionRS}).
\end{proof}
\begin{corollary}
The replica symmetric  approximation of the quenched pressure for the Sherrington-Kirkpatrick model with a signal is obtained by posing $t=1$ and $\bm r = \bm 0$ in (\ref{traspparametric}), which gives
\begin{align}\label{SolRS1}
A_{RS}(\beta,J_0,J)= \mathbb{E} \left[\log 2\cosh (\beta(J\sqrt{\qb}z+\mb J_0)\right]+ \frac{\beta^2}{4}J^2 (1-\qb)^2-\beta \frac{J_0}{2}\mb^2.
\end{align}
\end{corollary}
\begin{corollary}
At equilibrium, the order parameters of the model (\ref{DefSKsignal}) fulfill a set of self-consistency equations
\begin{align}
\label{eq:sq11}
\mb &= \mathbb{E}\left \{  \tanh \left [ \beta(J\sqrt{\qb}z+\mb J_0) \right ] \right \}, \\
\label{eq:sq12}
\qb &= \mathbb{E}  \left \{ \tanh^2 \left [ \beta(J\sqrt{\qb}z+\mb J_0) \right ] \right \}.
\end{align}
\end{corollary}
\begin{proof}
Equations (\ref{eq:sq11})-(\ref{eq:sq12}) can be obtained by comparing (\ref{eq:seconda})-(\ref{eq:terza}) with the derivatives of $A_{RS}(\beta,J_0,J)$ calculated from (\ref{SolRS1}) as
\begin{align}
\frac{\partial}{\partial \mb} A_{RS} &= -\beta J_0 \mb + \mathbb{E}\left  \{  \tanh \left [ \beta \left (\frac{J}{\sqrt{2}}\sqrt{\qb}z+\mb J_0 \right) \right ] \right \} \beta J_0 = 0 \\
\frac{\partial}{\partial \qb} A_{RS} &= -\beta\frac{J}{2\sqrt{2\qb}} \mathbb{E} \left \{ \tanh \left[ \beta \left(\frac{J}{\sqrt{2}}\sqrt{\qb}z+\mb J_0 \right) \right] z \right \} - \frac{\beta^2}{4}(1-\qb)J^2 = 0.
\end{align}
This solution recovers the one previously obtained for the same model (\ref{DefSKsignal}) by replica trick \cite{Coolen}.
\end{proof}

\subsection{Broken Replica Interpolation: 1-RSB solution} \label{ssec:SKRSB}

In this subsection we turn to the RSB scenario, following the seminal paper by Francesco Guerra dealing with the standard Sherrington-Kirkpatrick model \cite{Guerra} (or its implementation within the mechanical analogy \cite{BGDiBiasio}). In particular, we no longer assume self-averaging for the two-replica overlaps $q_{ab}$, rather -- at the first broken replica step -- these can concentrate on two values referred to as $\bar{q}_1,\ \bar{q}_2$; as for the magnetization density function $P(m)$, en route for the RSB in the Hopfield model (for which the ansatz adopted in previous works prescribe that the Mattis magnetization still self-averages), we retain $\lim_{N \to \infty}P(m)=\delta(m \pm \bar{m})$, as in the previous section. Then, Definition \ref{def:SK_RS} is updated by
\begin{definition} \label{def:SK_RSB}
In the first step of replica-symmetry breaking, the distribution of the two-replica overlap, in the thermodynamic limit, displays two delta-peaks at the equilibrium values (denoted with $\bar{q}_1,\ \bar{q}_2$) and the concentration on the two values is ruled by $\theta \in [0,1]$ \footnote{Note that this is usually called $m$ in Parisi theory, but here, to avoid ambiguity for the symbol $m$, already used to the denote the magnetization, we refer to the Parisi parameter as $\theta$.}, namely
\begin{equation}
\lim_{N \rightarrow + \infty} P_N(q) = \theta \delta (q - \bar{q}_1) + (1-\theta) \delta (q - \bar{q}_2), \label{limforq2}
\end{equation}
while the magnetization still self-averages at $ \bar{m}$ as in (\ref{limform}).
\end{definition}
Further, we need to introduce a more tricky interpolating structure as well as a more complex quenched average, as reported in the following
\begin{definition}
Given the interpolating parameters $\boldsymbol r = (x^{(1)}, x^{(2)}, w), t$ and the i.i.d. auxiliary fields $\{ h_i^{(1)}, h_i^{(2)}\}_{i=1,...,N}$ with $h_i^{(1,2)} \sim \mathcal N [0,1]$ for $i=1, ...,N$, we can write the 1-RSB interpolating partition function $\mathcal Z_N(t, \boldsymbol r)$ recursively, starting by
\begin{align}
\mathcal Z_2 (t, \bm r)&=  \sums \exp \left [ \beta \left (\frac{\sqrt{t}}{2}\frac{J \sqrt{2}}{\sqrt{N}} \sumij \si \sj z_{ij} + \suma \sqrt{x^{(a)}} \sumi h^{(a)}_i \si + t\frac{J_0}{2}Nm^2(\boldsymbol{\sigma}) +w J_0 N m(\boldsymbol{\sigma}) \right) \right],
\end{align}
and then averaging out the fields one per time, namely by defining
\begin{align}
\mathcal Z_1(t, \bm r) \coloneqq & \mathbb E_2 \bigg [ \mathcal Z_2(t, \bm r)^\theta \bigg ]^{1/\theta}, \\
\mathcal Z_0(t, \bm r) \coloneqq &\exp \mathbb E_1 \bigg[ \log \mathcal Z_1(t, \bm r) \bigg ],\\
\mathcal Z_N(t, \boldsymbol r) \coloneqq & \mathcal Z_0(t, \bm r),
\end{align}
where with $\mathbb E_2$ and $\mathbb E_1$ we denote the average over the variables $h_i^{(2)}$'s and $h_i^{(1)}$'s, respectively, and with $\mathbb E_0$ we shall denote the average over the variables $z_{ij}$'s.
\end{definition}
\begin{definition}
The 1RSB interpolating pressure, at finite volume $N$, is introduced as
\begin{equation}\label{AdiSK1RSB}
\mathcal A_N (t, \bm r) \coloneqq \frac{1}{N} \mathbb E_0 \left [ \log \mathcal Z_N(t, \bm r) \right],
\end{equation}
and, in the thermodynamic limit,
\begin{equation}
\mathcal A (t, \bm r) \coloneqq \lim_{N \to \infty} \mathcal A_N (t, \bm r).
\end{equation}
By setting $t=1, \boldsymbol r= \bm 0$, the interpolating pressure recovers the standard pressure (\ref{PressureDef}), that is, $A_N(\beta, J_0, J) = \mathcal A_N (t =1, \bm r =\bm 0)$.
\end{definition}
\begin{remark}
In order to lighten the notation, hereafter we use the following
\begin{align}
\label{eq:unouno}
\langle m \rangle \coloneqq & \mathbb{E}_0 \mathbb{E}_1 \mathbb{E}_2 \left[\mathcal W_2\frac{1}{N}\sum_{i=1}^N \omega(\sigma_i) \right]\\
\langle m^2 \rangle \coloneqq & \mathbb{E}_0 \mathbb{E}_1 \mathbb{E}_2 \left[\mathcal W_2\frac{1}{N^2}\sum_{i,j=1}^{N,N} \omega(\si \sj) \right] \\
 \langle q_{12} \rangle_1 \coloneqq &  \mathbb{E}_0 \mathbb{E}_1 \left[\frac{1}{N}\sum_{i=1}^N \left( \mathbb E_2 \big[\mathcal W_2\omega(\sigma_i)\big] \right)^2 \right] \\
  \label{eq:duedue}
\langle q_{12} \rangle_2 \coloneqq & \mathbb{E}_0 \mathbb{E}_1 \mathbb{E}_2 \left [\mathcal W_2\frac{1}{N}\sum_{i=1}^N \omega^2(\sigma_i) \right]
\end{align}
where we define the weight
\begin{equation}
\mathcal W_2 \coloneqq \frac{\mathcal Z_2^\theta}{ \mathbb E_2 \left [\mathcal Z_2^\theta \right ]}.
\end{equation}
\end{remark}
In analogy to subsec.~\ref{ssec:SKRS}, we aim to build a differential equation for the interpolating quenched pressure for which we preliminary need to evaluate the partial derivatives as stated in the following
\begin{lemma} \label{lemma:2}
The partial derivatives of the interpolating quenched pressure read as
\begin{align}
\label{eq:2app}
\dt \mathcal A_N &= \frac{\beta^2}{4} J^2 (1-(1-\theta)\langle q_{12}^2 \rangle_2 - \theta \langle q_{12}^2 \rangle_1) + \frac{\beta J_0}{2}\langle m^2 \rangle \\
\frac{\partial}{\partial x^{(1)}} \mathcal A_N &= \frac{\beta^2}{2} (1-(1-\theta)\langle q_{12} \rangle_2 - \theta \langle q_{12} \rangle_1) \label{eq:x1app}\\
\frac{\partial}{\partial x^{(2)}} \mathcal A_N &= \frac{\beta^2}{2}(1-(1-\theta)\langle q_{12}\rangle_2) \label{eq:x2app}\\
\dw \mathcal A_N &= \beta J_0 \langle m \rangle \label{eq:wapp}.
\end{align}
\end{lemma}
Since the proof of this Lemma is pretty lengthy but does not require any tricky passage, we leave it for the Appendix \ref{app1}.
\begin{proposition}
\label{prop:3}
The streaming of the 1-RSB interpolating quenched pressure obeys, at finite volume $N$, a standard transport equation, that reads as
\begin{align}
\label{eq:stream1rsb}
\frac{d \mathcal A_N}{dt}&=\frac{\partial}{\partial t} \mathcal A_N+\dot x^{(1)} \frac{\partial}{\partial x^{(1)}}  \mathcal A_N +\dot x^{(2)} \frac{\partial}{\partial x^{(2)}}  \mathcal A_N +\dot w \frac{\partial}{\partial w}   \mathcal A_N = S(t, \boldsymbol r) + V_N(t, \boldsymbol r),
\end{align}
where
\begin{align}
\label{S}
S(t, \boldsymbol r) &\coloneqq \frac{\beta^2}{4}J^2 \left[ 1 + (1-\theta)\qb_2^2 - 2 \qb_2 +  \theta \qb_1^2 \right] - \frac{\beta}{2}J_0 \mb^2,\\
\label{V}
V_N(t, \boldsymbol r) &\coloneqq\frac{\beta^2}{4} J^2  \left[ (1-\theta) \langle (q_{12}-\qb_2)^2\rangle_2+  \theta\langle (q_{12}-\qb_1)^2\rangle_1 \right]+\frac{\beta J_0}{2}\langle (m^2 - \mb)^2 \rangle.
\end{align}
\end{proposition}
\begin{proof}
Recalling Definition \ref{def:SK_RSB}, 
we can write
\begin{align}
\langle (q_{12}-\qb_2)^2\rangle_1 &:= \langle q_{12}^2 \rangle_1 + \qb_1^2 - 2\qb_2\langle q_{12}\rangle_1 \label{eq:q1} \\
\langle (q_{12}-\qb_2)^2\rangle_2 &:= \langle q_{12}^2 \rangle_2 + \qb_2^2 - 2\qb_2\langle q_{12}\rangle_2 \label{eq:q2}\\
\langle (m^2 - \mb)^2 \rangle &:= \langle m^2 \rangle + \mb^2 -2\mb\langle m \rangle. \label{eq:m}
\end{align}
Starting from (\ref{eq:2app}) and using (\ref{eq:q1}), (\ref{eq:q2}) and (\ref{eq:m})
\begin{align}
\dt \mathcal A_N &= \frac{\beta^2}{4} J^2 \left [1-(1-\theta)(\langle (q_{12}-\qb_2)^2\rangle_2 - \qb_2^2+2\qb_2\langle q_{12} \rangle_2) - \theta (\langle (q_{12}-\qb_1)^2\rangle_1 - \qb_1^2+2\qb_1\langle q_{12} \rangle_1) \right]+ \notag \\
&+ \frac{\beta J_0}{2} \left[ \langle (m - \mb)^2 \rangle - \mb^2 +2\mb\langle m \rangle \right]= \notag \\
&=\frac{\beta^2}{4} J^2-\frac{\beta^2}{4} J^2 (1-\theta)\langle (q_{12}-\qb_2)^2\rangle_2 + \frac{\beta^2}{4} J^2 (1-\theta) \qb_2^2-\frac{\beta^2}{2} J^2 (1-\theta)\qb_2\langle q_{12} \rangle_2 - \notag \\
\notag
&+ \frac{\beta^2}{4} J^2 \theta\langle (q_{12}-\qb_1)^2\rangle_1 - \frac{\beta^2}{4} J^2 \theta \qb_1^2+\frac{\beta^2}{2} J^2 \theta\qb_1\langle q_{12} \rangle_1+\frac{\beta J_0}{2}\langle (m - \mb)^2 \rangle -\frac{\beta J_0}{2} \mb^2 +2\frac{\beta J_0}{2}\mb\langle m \rangle.
\end{align}
Now, we include (\ref{eq:x1app})-(\ref{eq:wapp}) to get 
\begin{align}
\dt \mathcal{A}_N &= \frac{\beta^2}{4} J^2-\frac{\beta^2}{4} J^2 (1-\theta)\langle (q_{12}-\qb_2)^2\rangle_2+\frac{\beta^2}{4} J^2 (1-\theta) \qb_2^2+ J^2\qb_2 \left(\frac{\partial}{\partial x^{(2)}} \mathcal{A}_N - \frac{\beta^2}{2} \right) -\notag \\
&+ \frac{\beta^2}{4} J^2 \theta\langle (q_{12}-\qb_1)^2\rangle_1 -\frac{\beta^2}{4} J^2 \theta \qb_1^2 - J^2\qb_1 \left(\frac{\partial}{\partial x^{(1)}}\mathcal{A}_N - \frac{\beta^2}{2} + \frac{\beta^2}{2}(1-\theta)\langle q_{12} \rangle_2 \right)+ \notag \\
\notag
&+\frac{\beta J_0}{2}\langle (m - \mb)^2 \rangle - \frac{\beta J_0}{2} \mb^2 +\mb \dw \mathcal{A}_N.
\end{align}
Rearranging the equation along with (\ref{S}) and (\ref{V}), we obtain
\begin{align}
\dt \mathcal{A}_N&= %
J^2\qb_1 \frac{\partial}{\partial x^{(1)}}\mathcal{A}_N +\mb \dw \mathcal{A}_N + J^2(\qb_2-\qb_1) \frac{\partial}{\partial x^{(2)}}\mathcal{A}_N + V_N(t, \boldsymbol r) + S(t, \boldsymbol r),
\label{eq:dtAN}
\end{align}
and we reach the thesis by posing
\begin{align}
\dot x^{(1)} &= -J^2 \qb_1 \\
\dot x^{(2)}&=- J^2( \qb_2-\qb_1 ) \\
\dot w &= - \mb.
\end{align}
\end{proof}
\begin{remark} \label{r:latter}
In the thermodynamic limit, in the 1-RSB scenario under investigation, we have
\begin{align}
\lim_{N \rightarrow + \infty} \langle m \rangle &= \mb \label{limform1RSB}\\
\lim_{N \rightarrow + \infty} \langle q_{12} \rangle_1 &= \qb_1 \label{limforq1RSB} \\
\lim_{N \rightarrow + \infty} \langle q_{12} \rangle_2 &= \qb_2 \label{limforq21RSB}
\end{align}
in such a way that the potential in (\ref{eq:stream1rsb}) is vanishing, that is
\begin{equation} \label{eq:V0RSB}
\lim_{N \to \infty} V_N(t, \boldsymbol  r)=0.
\end{equation}
\end{remark} 
Note that setting the potential equal to zero is equivalent to assuming the existence of two temporal scales for thermalization, a slow one and a fast one, and self-averaging within each time scale, in such a way that if two replicas behave the same on both the timescales the average for their overlap is given by $\langle . \rangle_2$, while if they match only on the fast one but not on the slow one then the average for their overlap is given by $\langle . \rangle_1$; we refer to Section \ref{comment} for a deeper discussion on the physics behind this choice.

Exploiting Remark \ref{r:latter} we can prove the following
\begin{proposition}
The transport equation associated to the interpolating pressure of the Sherrington-Kirkpatrick model with a signal, in the thermodynamic limit and in the 1RSB scenario, reads as
\begin{align}
\dt \mathcal A_{\textrm{1RSB}} - &\qb_1 J^2 \frac{\partial}{\partial x^{(1)}} \mathcal A_{\textrm{1RSB}} - (\qb_2-\qb_1)J^2 \frac{\partial}{\partial x^{(2)}} \mathcal A_{\textrm{1RSB}} - \mb \dw \mathcal A_{\textrm{1RSB}} = \notag \\
&\frac{\beta^2}{4}J^2-\frac{\beta J_0}{2}\mb^2 +\frac{\beta^2}{4} J^2 (\theta\qb_1^2 + (1-\theta)\qb_2^2) - \frac{\beta^2}{2}J^2 \qb_2,
\label{1rsbtranseqprima}
\end{align}
whose solution is given by
\begin{align}
\mathcal A_{\textrm{1RSB}}(t, \boldsymbol r) =& \log 2 + \mathbb{E}_1 \left \{ \frac{1}{\theta} \log \left [ \mathbb{E}_2 \left (\cosh^\theta \left(\beta \suma \sqrt{x_0^{(a)}}h^{(a)} + \beta w_0 J_0 \right)\right) \right] \right \} + \notag\\
&+ \frac{\beta^2}{4}J^2 \left [ 1 + (\theta\qb_1^2 + (1-\theta)\qb_2^2)t - 2 \qb_2 t  \right] -\frac{\beta J_0}{2}\mb^2 t.
\label{1rsbtranseq}
\end{align}
\end{proposition}
\begin{proof}
The PDE (\ref{1rsbtranseqprima}) in the thermodynamic limit can be obtained from (\ref{eq:dtAN}) using (\ref{eq:V0RSB}).
This PDE can be solved via the method of the characteristics: the solution can be written in the form
\begin{align}
\mathcal A (t, \boldsymbol r)= \mathcal A_N(0, \bm r  - \dot {\bm r} t) + S(t, \boldsymbol r)t
\label{eq:A1RSB}
\end{align}
%
where $\dot {\bm r} = (\dot x^{(1)}, \dot x^{(2)}, \dot w)$ and the characteristics are
\begin{align}
x^{(1)} &=x_0^{(1)} - J^2 \qb_1 t, \label{charx1}\\
x^{(2)} &= x_0^{(2)} - J^2 (\qb_2-\qb_1) t, \label{charx2} \\
w &= w_0 - \mb t. \label{charw}
\end{align}
The Cauchy condition, corresponding to $t=0$ and $\boldsymbol {r_0} = \bm r (t=0)$, can be calculated directly, as it is a one-body calculation, and returns
\begin{align}
\mathcal A_N(0, \boldsymbol {r_0}) &=\frac{1}{N} \mathbb{E}_0  \mathbb{E}_1 \left\{\frac{1}{\theta} \log \left[\mathbb{E}_2 \left(\sums \exp \left(\beta\sumi \suma \sqrt{x_0^{(a)}}h_i^{(a)} \si +w_0 J_0\si \right)\right)^\theta\right]\right\}= \notag \\
&= \frac{1}{N} \mathbb{E}_1 \left\{\frac{1}{\theta} \log\left[\mathbb{E}_2 \left(\prod_i \sums \exp\left(\beta \suma \sqrt{x_0^{(a)}}h^{(a)}_i \si +w_0 J_0 \si\right)\right)^\theta\right]\right\}
\end{align}
where in the second passage we factorized the exponential functions. Now, we apply the definition of hyperbolic cosine and, since we have $N$ copies of the same average, it is possible to simplify the factor $1/N$ as
\begin{align}
\mathcal A_N(0, \boldsymbol {r_0}) &=\mathbb{E}_1\left\{\frac{1}{\theta} \log\left[\mathbb{E}_2 \left(2^\theta \cosh^\theta\left(\beta\left(\suma \sqrt{x_0^{(a)}}h^{(a)} + w_0J_0\right)\right)\right)\right]\right\}= \notag \\
&=\log 2 + \mathbb{E}_1 \left \{ \frac{1}{\theta} \log \left[ \mathbb{E}_2 \left( \cosh^\theta \left( \beta \left( \suma \sqrt{x_0^{(a)}}h^{(a)} + w_0 J_0 \right ) \right) \right) \right] \right \}. \label{eq:A01RSB}
\end{align}
To sum up, placing (\ref{eq:A01RSB}) in (\ref{eq:A1RSB}), we reach (\ref{1rsbtranseq}).
\end{proof}
We have all the ingredients to state the first main theorem, namely
\begin{theorem}
The 1-RSB quenched pressure for the Sherrington-Kirkpatrick model with a signal, in the thermodynamic limit, reads as
\begin{align}
A_{\textrm{1RSB}}(\beta,J_0,J)&= \log 2 + \mathbb{E}_1 \left\{\frac{1}{\theta} \log \left[\mathbb{E}_2(\cosh^\theta (\beta J\sqrt{\qb_1}h^{(1)} + \beta J \sqrt{\qb_2-\qb_1}h^{(2)} + \beta \mb J_0))\right]\right\} \notag \\
\label{eq:fina}
& + \frac{\beta^2}{4}J^2-\frac{\beta J_0}{2}\mb^2  +\frac{\beta^2}{4} J^2 (\theta\qb_1^2 + (1-\theta)\qb_2^2) - \frac{\beta^2}{2}J^2 \qb_2.
\end{align}
\end{theorem}
\begin{proof}
It is sufficient to pose $t=1$ and $x^{(1)} = x^{(2)} = w=0$ in (\ref{1rsbtranseq}). In fact, for this choice of interpolating parameters we recover the original model.
\end{proof}
\begin{corollary} \label{cor:SC_SK}
The self-consistent equations for the order parameters of the model (\ref{DefSKsignal}) read as
\begin{align}
\label{seq:a}
\qb_1 = \mathbb{E}_1 &\left \{ \frac{\mathbb{E}_2 \left[ \cosh^\theta \left(g(\boldsymbol h, \mb) \right)\tanh \left(g(\boldsymbol h, \mb) \right) \right]}{\mathbb{E}_2 \left [ \cosh^\theta \left(g(\boldsymbol h, \mb) \right) \right ]}  \right \}^2 \\
\label{seq:b}
\qb_2 = \mathbb{E}_1 &\left \{ \frac{\mathbb{E}_2 \left[ \cosh^\theta \left(g(\boldsymbol h, \mb) \right)\tanh^2 \left(g(\boldsymbol h, \mb) \right) \right ]}{\mathbb{E}_2 \left[ \cosh^\theta \left(g(\boldsymbol h, \mb) \right) \right] }\right \} \\
\label{seq:c}
\mb = \mathbb{E}_1 &\left \{ \frac{\mathbb{E}_2 \left[ \cosh^\theta \left(g(\boldsymbol h, \mb) \right)\tanh \left(g(\boldsymbol h, \mb) \right) \right]}{\mathbb{E}_2 \left[ \cosh^\theta \left(g(\boldsymbol h, \mb) \right) \right] }\right \},
\end{align}
where
$\boldsymbol h = (h^{(1)}, h^{(2)})$, and $g(\boldsymbol h, \mb)=\beta \frac{J}{2} \sqrt{\qb_1}h^{(1)} + \beta \frac{J}{2} \sqrt{\qb_2-\qb_1}h^{(2)} +\beta \mb J_0$.
\end{corollary}
\begin{proof}
Here we just sketch the proof, while full details are provided in Appendix \ref{app:SC_SK}.\\
First, let us resume the derivatives (\ref{eq:x1app})-(\ref{eq:wapp}) and set them in the 1RSB framework 
\begin{align}
\frac{\partial}{\partial x^{(1)}} A_{\textrm{1RSB}} &=\frac{\beta^2}{2} - \frac{\beta^2}{2}(1-\theta)\qb_2 - \theta \qb_1 \label{selfx1} \\
\frac{\partial}{\partial x^{(2)}} A_{\textrm{1RSB}} &= \frac{\beta^2}{2} - \frac{\beta^2}{2}(1-\theta)\qb_2 \label{selfx2} \\
\frac{\partial}{\partial w} A_{\textrm{1RSB}} &= \frac{\beta J_0}{2}\mb \label{selfw}.
\end{align}
This set of equations is interpreted as a system of three equations and three unknowns $(\mb, \qb_1, \qb_2)$.
Next, we evaluate the derivatives of $A_{\textrm{1RSB}}$ w.r.t. $x^{(1,2)}$ and $w$ starting from (\ref{eq:fina}), we plug the resulting expressions into (\ref{selfx1})-( \ref{selfw}) and, finally, with some algebra, we get (\ref{seq:a})-(\ref{seq:c}).
\end{proof}

\subsection{Broken Replica Interpolation: 2-RSB solution} \label{ssec:SK2RSB}
In the second step of RSB, the two-replica overlaps $q_{ab}$ can concentrate on three values, referred to as $\bar{q}_1,\ \bar{q}_2, \ \bar{q}_3$, while we still assume $m$ to be self-averaging. Otherwise stated, Definition \ref{def:SK_RSB} is updated by
\begin{definition} \label{def:SK_2RSB}
In the second step of replica-symmetry breaking, the distribution of the two-replica overlap, in the thermodynamic limit, displays three delta-peaks at the equilibrium values (denoted with $\bar{q}_1,\ \bar{q}_2, \bar{q}_3$) and the concentration on the three values is ruled by $\theta_1 \in [0,1], \theta_2 \in [0,1]$, namely
\begin{equation}
\lim_{N \rightarrow + \infty} P_N(q) = \theta_1\delta (q - \bar{q}_1) + \theta_2 \delta (q - \bar{q}_2) +(1-\theta_2) \delta (q - \bar{q}_3), \label{limforq3}
\end{equation}
while the magnetization still self-averages at $ \bar{m}$ as in (\ref{limform}).
\end{definition}

Further, we need to introduce a new interpolating structure as well as a more complex quenched average, as reported in the following
\begin{definition}
Given the interpolating parameters $\boldsymbol r = (x^{(1)}, x^{(2)}, x^{(3)}, w)$, $t$ and the i.i.d. auxiliary fields $\{ h_i^{(1)}, h_i^{(2)}, h_i^{(3)}\}_{i=1,...,N}$, with $h_i^{(1,2,3)} \sim \mathcal N[0,1]$, for $i = 1,....,N$, we can write the 2-RSB interpolating partition function $\mathcal Z_N(t, \boldsymbol r)$ recursively, starting by
\begin{align}
\mathcal Z_3 (t, \bm r)&=  \sums \exp \left [ \beta \left (\sqrt{t}\frac{J}{2\sqrt{N}} \sumij \si \sj z_{ij} + \sumanew \sqrt{x^{(a)}} \sumi h^{(a)}_i \si + t\frac{J_0}{2}Nm^2(\boldsymbol{\sigma}) +w J_0 N m(\boldsymbol{\sigma}) \right) \right],
\end{align}
and then averaging out the fields one per time, namely by defining
\begin{align}
\mathcal{Z}_2(t, \bm r)=& \mathbb{E}_3\left[ \mathcal{Z}_3(t, \bm r)^{\theta_2}\right]^{1/{\theta_2}}, \\
\mathcal Z_1(t, \bm r)=&\mathbb{E}_2 \left[ \mathcal Z_2(t, \bm r)^{\theta_1} \right]^{1/{\theta_1}}, \\
\mathcal Z_0(t, \bm r)=&\exp \left(\mathbb{E}_1\left[\log \mathcal{Z}_1(t, \bm r) \right]\right),\\
\mathcal Z_N(t, \boldsymbol r) \coloneqq & \mathcal Z_0(t, \bm r),
\end{align}
where with $\mathbb{E}_a$ we denote the average over the variables $h_i^{(a)}$'s, for $a=1, 2, 3$ and with $\mathbb E_0$ we shall denote the average over the variables $z_{ij}$'s, further, we adopt the vectorial notation $\boldsymbol x = (x^{(1)}, x^{(2)}, x^{(3)})$. 
\end{definition}

\begin{definition}
The 2RSB interpolating pressure, at finite volume $N$, is introduced as
\begin{equation}\label{AdiSK2RSB}
\mathcal A_N (t, \bm r) \coloneqq \frac{1}{N} \mathbb{E}_0\left[ \log \mathcal Z_N(t, \bm r) \right],
\end{equation}
and, in the thermodynamic limit,
\begin{equation}
\mathcal A (t, \bm r) \coloneqq \lim_{N \to \infty} \mathcal A_N (t, \bm r).
\end{equation}
By setting $t=1$, $\boldsymbol x = \boldsymbol 0$, and $w=0$, the interpolating pressure recovers the standard pressure (\ref{PressureDef}), that is, $A_N(\beta, J_0, J) = \mathcal A_N (t =1, \bm r =\bm 0)$.
\end{definition}

\begin{remark}
In order to lighten the notation, hereafter we use the following
\begin{align}
\langle m \rangle=&\mathbb{E}_0  \mathbb{E}_1 \mathbb{E}_2  \mathbb{E}_3\left \{ \mathcal{W}_2 \mathcal{W}_3 \frac{1}{N}\sum_{i=1}^N \omega(\sigma_i) \right \}\\
\langle m^2 \rangle=&\mathbb{E}_0 \mathbb{E}_1 \mathbb{E}_2 \mathbb{E}_3 \left \{ \mathcal{W}_2 \mathcal{W}_3 \frac{1}{N^2}\sum_{i,j=1}^{N,N} \omega(\si \sj) \right \} \\
\langle q_{12} \rangle_1=&\mathbb{E}_0  \mathbb{E}_1 \left\{ \frac{1}{N}\sum_{i=1}^N \left[ \mathbb{E}_2\left ( \mathcal{W}_2\mathbb{E}_3 \left(  \mathcal{W}_3 \omega(\sigma_i) \right) \right) \right]^2 \right\}  \\
\langle q_{12} \rangle_2=&\mathbb{E}_0  \mathbb{E}_1 \mathbb{E}_2 \left\{ \mathcal{W}_2 \frac{1}{N}\sum_{i=1}^N \left[\mathbb{E}_3\left( \mathcal{W}_3 \omega(\sigma_i) \right)\right]^2\right\} \\
\langle q_{12} \rangle_3=&\mathbb{E}_0  \mathbb{E}_1 \mathbb{E}_2 \left\{\mathcal{W}_2\mathbb{E}_3\left[  \mathcal{W}_3 \frac{1}{N}\sum_{i=1}^N \omega^2(\sigma_i) \right]\right\}
\end{align}
where we define the weights
\begin{align}
\mathcal{W}_2 &\coloneqq\frac{\mathcal Z_2^{\theta_1}}{\mathbb{E}_2\left[\mathcal Z_2^{\theta_1}\right]}, \\
\mathcal{W}_3 &\coloneqq\frac{\mathcal Z_3^{\theta_2}}{\mathbb{E}_3\left[\mathcal Z_3^{\theta_2}\right]}.
\end{align}
\end{remark}

In analogy to subsecs.~\ref{ssec:SKRS} and \ref{ssec:SKRSB}, we aim to build a differential equation for the interpolating quenched pressure for which we preliminary need to evaluate the partial derivatives as given by
\begin{lemma}
\label{lemma2RSB}
The partial derivatives of the interpolating quenched pressure read as
\begin{align}
\label{eq:2app2RSB}
\dt \mathcal A_N &= \frac{\beta^2}{4} J^2 (1-(1-\theta_2)\langle q_{12}^2 \rangle_3 - (\theta_2-\theta_1) \langle q_{12}^2 \rangle_2 - \theta_1 \langle q_{12}^2 \rangle_1) + \frac{\beta J_0}{2}\langle m^2 \rangle \\
\frac{\partial}{\partial x^{(1)}} \mathcal A_N &= \frac{\beta^2}{2} (1-(1-\theta_2)\langle q_{12} \rangle_3 - (\theta_2-\theta_1) \langle q_{12} \rangle_2 - \theta_1 \langle q_{12}\rangle_1) \label{eq:x1app2RSB}\\
\frac{\partial}{\partial x^{(2)}} \mathcal A_N &= \frac{\beta^2}{2} (1-(1-\theta_2)\langle q_{12} \rangle_3 - (\theta_2-\theta_1) \langle q_{12} \rangle_2) \label{eq:x2app2RSB}\\
\frac{\partial}{\partial x^{(3)}} \mathcal A_N &= \frac{\beta^2}{2} (1-(1-\theta_2)\langle q_{12} \rangle_3) \label{eq:x3app2RSB}\\
\dw \mathcal A_N &= \beta J_0 \langle m \rangle \label{eq:wapp2RSB}.
\end{align}
\end{lemma}
We omit the proof since it is similar to that provided for Lemma \ref{lemma:2}.

Indeed, similarly to Proposition \ref{prop:3}, we can write
\begin{proposition}
The streaming of the 2-RSB interpolating quenched pressure obeys, at finite volume $N$, a standard transport equation, that reads as
\begin{align}
\label{eq:stream2rsb}
\frac{d \mathcal A_N}{dt}&=\frac{\partial}{\partial t} \mathcal A_N+\dot x^{(1)} \frac{\partial}{\partial x^{(1)}}  \mathcal A_N +\dot x^{(2)} \frac{\partial}{\partial x^{(2)}}  \mathcal A_N +\dot x^{(3)} \frac{\partial}{\partial x^{(3)}}  \mathcal A_N + \dot w \frac{\partial}{\partial w}   \mathcal A_N = S(t, \boldsymbol r) + V_N(t, \boldsymbol r),
\end{align}
where
\begin{align}
S(t, \boldsymbol r) \coloneqq &\frac{\beta^2}{4}J^2 \left[ (1 -\qb_3)^2 - \theta_1(\qb_2^2-\qb_1^2) - \theta_2(\qb_3^2 - \qb_2^2) \right] - \frac{\beta}{2}J_0 \mb^2,\\
V_N(t, \boldsymbol r) \coloneqq &\frac{\beta J_0}{2}\langle (m^2 - \mb)^2 \rangle-\frac{\beta^2}{8} J^2  \left[(1-\theta_2)\langle (q_{12} - \qb_3)^2 \rangle_3+(\theta_2-\theta_1) \langle (q_{12}-\qb_2)^2\rangle_2+ \right.\notag \\
&\left. + \theta_1\langle (q_{12}-\qb_1)^2\rangle_1 \right].
\end{align}
\end{proposition}

\begin{remark}
In the thermodynamic limit, in the 2-RSB scenario considered, we have
\begin{align}
\lim_{N \rightarrow + \infty} \langle m \rangle &= \mb \label{limform2RSB}\\
\lim_{N \rightarrow + \infty} \langle q_{12} \rangle_1 &= \qb_1 \label{limforq2RSB} \\
\lim_{N \rightarrow + \infty} \langle q_{12} \rangle_2 &= \qb_2 \label{limforq22RSB} \\
\lim_{N \rightarrow + \infty} \langle q_{12} \rangle_3 &= \qb_3 \label{limforq32RSB}
\end{align}
in such a way that the potential in (\ref{eq:stream2rsb}) is vanishing, that is
\begin{equation} \label{eq:V02RSB}
\lim_{N \to \infty} V_N(t, \boldsymbol  r)=0.
\end{equation}
\end{remark}
By naturally extending the picture obtained for the 1RSB case, setting the potential equal to zero is equivalent to requiring three temporal scales (slow, intermediate and fast) for thermalization and self-averaging within each time scale; we refer to Section \ref{comment} for a deeper discussion on the physics behind this choice.

Exploiting the last remark we can prove the following
\begin{proposition}
The transport equation associated to the interpolating pressure function defined in (\ref{AdiSK2RSB}), in the thermodynamic limit and in the 2RSB scenario, reads as
\begin{align}
\dt \mathcal A_{\textrm{2RSB}} - &\qb_1 J^2 \frac{\partial}{\partial x^{(1)}} \mathcal A_{\textrm{2RSB}} - (\qb_2-\qb_1) J^2 \frac{\partial}{\partial x^{(2)}} \mathcal A_{\textrm{2RSB}} -  (\qb_3-\qb_2) J^2 \frac{\partial}{\partial x^{(3)}} \mathcal A_{\textrm{2RSB}}-\notag \\
&- \mb \dw \mathcal A_{\textrm{2RSB}} = -\frac{\beta J_0}{2}\mb^2 +\frac{\beta^2}{4} J^2 \left[(1-\qb_3)^2 - \theta_1(\qb_2^2-\qb_1^2) -(\theta_2-\theta_1)(\qb_3^2-\qb_2^2)\right]
\label{2rsbtranseqprima}
\end{align}
whose solution is given by
\begin{align}
\mathcal A_{\textrm{2RSB}}(t, \boldsymbol r) =& \log 2 + \mathbb{E}_1 \left \{ \frac{1}{\theta_1} \log \left\{ \mathbb{E}_2 \left[\mathbb{E}_3 \left (\cosh^{\theta_2} \left(\beta \sumanew \sqrt{x_0^{(a)}}h^{(a)} + \beta w_0 J_0 \right)\right) \right]^{\frac{\theta_1}{\theta_2}}\right\} \right \} + \notag\\
&+ \frac{\beta^2}{8}J^2 t \left [ (1-\qb_3)^2 -\theta_1 (\qb_2^2 - \qb_1^2) - \theta_2(\qb_3^2-\qb_2^2) \right] -\frac{\beta J_0}{2}\mb^2 t.
\label{2rsbtranseq}
\end{align}
\end{proposition}

We have all the ingredients to update the main theorem to the 2RSB scenario, namely
\begin{theorem}
The 2-RSB quenched pressure for the Sherrington-Kirkpatrick model with a signal, in the thermodynamic limit, reads as
\begin{align}
\label{eq:fina2RSB}
&A_{\textrm{2RSB}}(\beta,J_0,J)= \log 2 + \mathbb{E}_1 \Big \{\frac{1}{\theta_1} \log \Big[ \mathbb{E}_2 \Big[ \mathbb{E}_3 \cosh^{\theta_2} \Big (\beta J \sqrt{\qb_1}h^{(1)} + \beta J \sqrt{(\qb_2-\qb_1)}h^{(2)} +   \\ \notag 
&  \beta J \sqrt{(\qb_3-\qb_2)}h^{(3)}+ \beta J_0 \mb \Big)  \Big]^{\frac{\theta_1}{\theta_2}} \Big] \Big\} + \frac{\beta^2}{2}J^2 \left[ (1-\qb_3)^2 - \theta_1(\qb_2^2-\qb_1^2) - \theta_2(\qb_3^2 - \qb_2^2) \right ]-\frac{\beta J_0}{2}\mb^2.
\end{align}
\end{theorem}

\begin{proof}
It is sufficient to pose $t=1$ and $x^{(1)} = x^{(2)} = x^{(3)}= w=0$ in (\ref{2rsbtranseq}). In fact, for this choice of interpolating parameters we recover the original model.
\end{proof}

\begin{corollary} \label{cor:SC_SK_2RSB}
The self-consistent equations for the order parameters of the model (\ref{DefSKsignal}) read as
\begin{align}
\label{seq:a2RSB}
\qb_1 =& \mathbb{E}_1 \left\{ \frac{\mathbb{E}_2 \left[ \left[ \mathbb{E}_3 \cosh^{\theta_2}(g(\boldsymbol h,\mb))\right]^{\frac{\theta_1}{\theta_2}}\displaystyle{\frac{\mathbb{E}_3 \left( \cosh^{\theta_2}(g(\boldsymbol h,\mb))\tanh(g(\boldsymbol h,\mb)) \right)}{\mathbb{E}_3 \cosh^{\theta_2}(g(\boldsymbol h,\mb))}}\right]}{\mathbb{E}_2 \left[\mathbb{E}_3(\cosh^{\theta_2} (g(\boldsymbol h,\mb) ) \right]^{\frac{\theta_1}{\theta_2}}}\right\}^2\\
\label{seq:b2RSB}
\qb_2 =& \mathbb{E}_1 \left\{ \frac{\mathbb{E}_2 \left[ \left[ \mathbb{E}_3 \cosh^{\theta_2}(g(\boldsymbol h,\mb))\right]^{\frac{\theta_1}{\theta_2}}\left[\displaystyle{\frac{\mathbb{E}_3 \left( \cosh^{\theta_2}(g(\boldsymbol h, \mb))\tanh(g(\boldsymbol h,\mb)) \right)}{\mathbb{E}_3 \cosh^{\theta_2}(g(\boldsymbol h,\mb))}}\right]^2\right]}{\mathbb{E}_2 \left[\mathbb{E}_3(\cosh^{\theta_2} (g(\boldsymbol h,\mb) ) \right]^{\frac{\theta_1}{\theta_2}}}\right\}
\end{align}
\begin{align}
\label{seq:c2RSB}
\qb_3 =& \mathbb{E}_1 \left\{ \frac{\mathbb{E}_2 \left[ \left[ \mathbb{E}_3 \cosh^{\theta_2}(g(\boldsymbol h,\mb)\right]^{\frac{\theta_1}{\theta_2}}\displaystyle{\frac{\mathbb{E}_3 \left( \cosh^{\theta_2}(g(\boldsymbol h,\mb))\tanh^2(g(\boldsymbol h,\mb)) \right)}{\mathbb{E}_3 \cosh^{\theta_2}(g(\boldsymbol h,\mb))}}\right]}{\mathbb{E}_2 \left[\mathbb{E}_3 \cosh^{\theta_2} (g(\boldsymbol h,\mb) ) \right]^{\frac{\theta_1}{\theta_2}}}\right\} \\
\label{seq:m2RSB}
\mb =& \mathbb{E}_1 \left\{ \frac{\mathbb{E}_2 \left[ \left[ \mathbb{E}_3 \cosh^{\theta_2}(g(\boldsymbol h,\mb)\right]^{\frac{\theta_1}{\theta_2}}\displaystyle{\frac{\mathbb{E}_3 \left( \cosh^{\theta_2}(g(\boldsymbol h,\mb)\tanh(g(\boldsymbol h,\mb) \right)}{\mathbb{E}_3 \cosh^{\theta_2}(g(\boldsymbol h,\mb)}}\right]}{\mathbb{E}_2 \left[\mathbb{E}_3\cosh^{\theta_2} (g(\boldsymbol h,\mb) ) \right]^{\frac{\theta_1}{\theta_2}}}\right\}
\end{align}
where $\boldsymbol h = (h^{(1)}, h^{(2)},  h^{(3)})$, and $g(\boldsymbol h, \mb)=\beta J\sqrt{\qb_1}h^{(1)} + \beta J \sqrt{\qb_2-\qb_1}h^{(2)} + \beta J \sqrt{(\qb_3 - \qb_2)}h^{(3)}+\beta \mb J_0$. 
\end{corollary}

\begin{proof}
Here we just sketch the proof, which is similar to that provided for Corollary \ref{cor:SC_SK}.

First, let us resume the derivatives (\ref{eq:x1app2RSB})-(\ref{eq:wapp2RSB}) and set them in the 2RSB framework
\begin{align}
\frac{\partial}{\partial x^{(1)}} A_{\textrm{2RSB}} &=\frac{\beta^2}{2} - \frac{\beta^2}{2}(1-\theta_2)\qb_3 - \frac{\beta^2}{2}(\theta_2-\theta_1)\qb_2 - \theta_1 \qb_1 \label{selfx12RSB} \\
\frac{\partial}{\partial x^{(2)}} A_{\textrm{2RSB}} &= \frac{\beta^2}{2} - \frac{\beta^2}{2}(1-\theta_2)\qb_3 - \frac{\beta^2}{2}(\theta_2-\theta_1)\qb_2\label{selfx22RSB} \\
\frac{\partial}{\partial x^{(3)}} A_{\textrm{2RSB}} &= \frac{\beta^2}{2} - \frac{\beta^2}{2}(1-\theta_2)\qb_3 \label{selfx32RSB} \\
\frac{\partial}{\partial w} A_{\textrm{2RSB}} &= \beta J_0\mb \label{selfw2RSB}.
\end{align}
This set of equations is interpreted as a system of three equations and three unknowns $(\mb, \qb_1, \qb_2, \qb_3)$.
Next, we evaluate the derivatives of $A_{\textrm{2RSB}}$ w.r.t. $x^{(1,2,3)}$ and $w$ starting from (\ref{eq:fina2RSB}), we plug the resulting expressions into (\ref{selfx12RSB})-(\ref{selfw2RSB}) and, finally, with some algebra, we get the self-consistencies.
\end{proof}

\subsection{Broken Replica Interpolation: K-RSB solution} \label{ssec:SKKRSB}
While the first two steps of RSB were treated in details (for illustrative purposes and because, in the Hopfield counterpart, we will recover the already known expressions for its 1RSB and 2RSB quenched pressure), we now just give hints on the structure of the quenched pressure of the Sherrington-Kirkpatrick model with a signal for arbitrary, but finite, $K$ steps of RSB.
\begin{definition}
In the K-th step of replica-symmetry breaking, the distribution of the two-replica overlap, in the thermodynamic limit displays $K+1$ delta-peaks at the equilibrium values (denoted by $\qb_1, ... \qb_{K+1}$) and the concentration is ruled by $\theta_i \in [0, 1],\  i=1, ..., K$, namely
\begin{align}
\lim_{N\rightarrow +\infty} P_N(q)= \sum_{a=0}^K (\theta_{a+1}-\theta_{a}) \delta(q-\qb_{a+1})
\end{align}
with $\theta_{0}=0$ and $\theta_{K+1}=1$.
\end{definition}

\begin{definition}
\label{def:SK_KRSB}
Given the interpolating parameters $\boldsymbol r =(x^{(1)},...,x^{(K+1)}, w)$, $t$ and the i.i.d. auxiliary fields $\{h_i^{(1)}, h_i^{(2)}, ..., h_i^{(K+1)}\}_{i=1,...,N}$, with $h_i^{(1, 2, ..., K+1)} \sim \mathcal N [0,1]$, for $i=1, ...,N$, we can write the K-RSB interpolating partition function $\mathcal Z_N(t, \boldsymbol r)$ recursively, starting by
\begin{align}
\mathcal Z_{K+1} (t, \bm r)&=  \sums \exp \left [ \beta \left (\sqrt{t}\frac{J \sqrt{2}}{\sqrt{N}} \sumij \si \sj z_{ij} + \sum_{a=1}^{K+1} \sqrt{x^{(a)}} h^{(a)}_i \si + t\frac{J_0}{2}Nm^2(\boldsymbol{\sigma}) +w J_0 N m(\boldsymbol{\sigma}) \right) \right],
\end{align}
and then averaging out the fields one per time by appplying $\mathbb{E}_a$, which denotes the average over the variables $h_i^{(a)}$'s, for $a =1, ... K+1$, while $\mathbb{E}_0$ denotes the average over the variables $z_{ij}$'s.
\end{definition}

\begin{proposition}
The streaming of the $K-RSB$ interpolating quenched pressure obeys, at finite volume $N$ and finite $K$, a standard transport equation, namely
\begin{align}
\frac{d}{dt}\mathcal{A}_N = \dt \mathcal{A}_N + \sum_{b=1}^{K+1} \dot{x^{(b)}} \frac{\partial}{\partial x^{(b)}} \mathcal{A}_N + \dot{w} \dw \mathcal{A}_N = S(t, \bm r) + V_N(t, \bm r)
\label{eq:ANfinite_KRSB}
\end{align}
where
\begin{align}
S(t, \bm r)&=-\frac{\beta J_0}{2}\mb^2-\frac{\beta^2}{4}J^2 \left( \sum_{a=0}^K (\theta_{a+1} - \theta_a) \qb_{a+1}^2\right) + \frac{\beta^2}{2}J^2(1-\qb_{K+1}) \\
V_N (t, \bm r) &= - \frac{\beta^2}{4}J^2\sum_{a=0}^K (\theta_{a+1} - \theta_a) \langle (q_{12} - \qb_{a+1})^2 \rangle_{a+1} + \frac{\beta J_0}{2} \langle (m - \mb)^2 \rangle.
\end{align}
\end{proposition}

\begin{proof}
The proof is analogous to the previous ones provided for the $1RSB$ and $2RSB$ pictures. We start from
\begin{align}
\dt \mathcal{A}_N =\frac{\beta J_0}{2} \langle m^2 \rangle + \frac{\beta^2}{4}J^2 \left(1- \sum_{a=0}^K (\theta_{a+1} - \theta_a)\langle q_{12}^2 \rangle_{a+1}\right),
\end{align}
then we use
\begin{align}
\langle (q_{12} - \qb_a)^2 \rangle_a &= \langle q_{12}^2 \rangle_a + \qb_a^2 -2\qb_a\langle q_{12} \rangle_a, \ \forall a=1, ..., K+1 \label{quad_qa}\\
\langle (m - \mb)^2 \rangle &= \langle m^2 \rangle + \mb-2\mb\langle m \rangle \\
\frac{\beta^2}{2}&=\frac{\beta^2}{2}\left(1-\qb_{K+1} + \sum_{b=0}^{K}(\qb_{b+1}-\qb_{b})\right)
\end{align}
where $\qb_0=0$.
In this way, posing
\begin{align}
\dot{x}^{(b)} &= - (\qb_{b} - \qb_{b-1}), \ b=1, ..., K+1 \\
\dot{w} &= - \mb
\end{align}
we obtain the derivatives w.r.t. each $x^{(b)}$ and w.r.t. $w$.
\end{proof}

\begin{proposition}
The transport equation associated to the $K-RSB$ interpolating pressure, in the thermodynamic limit, reads as
\begin{eqnarray}
\notag
\frac{d}{dt}\mathcal{A}_N &=& \dt \mathcal{A}_N + \sum_{b=0}^K \dot{x}^{(b+1)} \frac{\partial}{\partial x^{(b+1)}} \mathcal{A}_N + \dot{w} \dw \mathcal{A}_N \\
&=& \frac{\beta J_0}{2}\mb^2-\frac{\beta^2}{4} \left( \sum_{a=0}^K (\theta_{a+1} - \theta_a) \qb_{a+1}^2\right) + \frac{\beta^2}{2}(1-\qb_{K+1})
\label{eq:AN_KRSB}
\end{eqnarray}
whose solution is given by
\begin{align}
\mathcal{A}_N = \mathcal{A}_N(0, \bm r - \boldsymbol{ \dot{r}}t) +  t\left [\frac{\beta J_0}{2}\mb^2-\frac{\beta^2}{4} \left( \sum_{a=0}^K (\theta_{a+1} - \theta_a) \qb_{a+1}^2\right) + \frac{\beta^2}{2}(1-\qb_{K+1}) \right].
\end{align}
where
\begin{align}
\mathcal{A}_N(0, \bm r - \boldsymbol{ \dot{r}}t) = \frac{1}{\theta_1}\int D h^{(1)} \log \mathcal{N}_1
\end{align}
with
\begin{align}
\mathcal{N}_a = \begin{cases}
\displaystyle{\int D h^{(a+1)} \left[ \mathcal{N}_{a+1}\right]^{\theta_a/\theta_{a+1}} }\ &\textnormal{for} \ a =1, ..., K \\
2 \cosh \left( \beta (w +  \sum_{a=0}^{K} \sqrt{x_0^{(a)}} h^{(a)}) \right) \ &\textnormal{for} \ a = K +1 \\
\end{cases}
\end{align}
and $D h^{(a)}$ represents the Gaussian measure, namely $dh^{(a)} (\sqrt{2\pi})^{-1} \exp(-{h^{(a)}}^2/2)$.
\end{proposition}
We can finally state the last, and more general, theorem for the first model under investigation.
\begin{theorem}
The K-RSB quenched pressure for the Sherrington-Kirkpatrick model with a signal, in the thermodynamic limit, reads as
\begin{align}
\mathcal{A}_{KRSB} = \frac{1}{\theta_1} \int D J^{(1)} \log \mathcal{N}_1 + \left[\frac{\beta J_0}{2}\mb^2-\frac{\beta^2}{4} \left( \sum_{a=0}^K (\theta_{a+1} - \theta_a) \qb_a^2\right) + \frac{\beta^2}{2}(1-\qb_K) \right]
\label{eq:Afin_KRSB}
\end{align}
with
\begin{align}
\mathcal{N}_a = \begin{cases}
\displaystyle{\int D h^{(a+1)} \left[\mathcal{N}_{a+1}\right]^{\theta_a/\theta_{a+1}} }\ &\textnormal{for} \ a =1, ... K \\
2 \cosh \left( \beta (\mb +  \sum_{a=0}^{K} \sqrt{\qb_{a}-\qb_{a-1}}  h^{(a)}) \right) \ &\textnormal{for} \ a = K+1 \\
\end{cases}
\label{eq:def_N}
\end{align}
\end{theorem}
\begin{proof}
If we put $t=1$ and $ \bm r = \bm 0$ we obtain the K-RSB quenched pressure.
\end{proof}

\begin{remark}
The construction of the transport equation is guaranteed by the original Guerra's scheme \cite{Guerra}: considering $J_{ij} \sim \mathcal{N}(0,1)$, if we pose $\sqrt{x^{(a)}}= \sqrt{(1-t)(\qb_a-\qb_{a-1})}$ in Guerra's partition function
\begin{align}
\sums \exp \left( \beta \sqrt{\frac{t}{N}} \sumij J_{ij} \si\sj + \beta h \sumi \si + \beta \sqrt{1-t} \sum_{a=0}^K \sqrt{\qb_a-\qb_{a-1}} \sum_i J_i^{(a)} \si \right)
\end{align}
we can use Theorem $4$ of the paper \cite{Guerra}, namely
\begin{align}
\label{eq:Guerra}
\frac{d}{dt} \mathcal{A}_N =-\frac{\beta^2}{4} \left( 1- \sum_{a=0}^K (m_{a+1} - m_a) \qb_a^2\right) - \frac{\beta^2}{4}\sum_{a=0}^K (m_{a+1} - m_a) \langle (q_{12} - \qb_a)^2 \rangle_a
\end{align}
to compute, being aware of (\ref{quad_qa}),  the transport equation also for the present (trivial) generalization.
\end{remark}
\begin{corollary}
\label{cor:selfKRSB_SK}
The self-consistence equations for the order parameters are thus 
\begin{align}
\qb_1 &=  \mathbb{E}_1 \left \{ \frac{1}{\mathcal{N}_1}\mathbb{E}_2 \left[ \mathcal{N}_2^{\frac{\theta_1}{\theta_2}-1} \cdots \mathbb{E}_{K+1} \left( \cosh^{\theta_K} \left( g(\boldsymbol h, \bar m) \right) \tanh \left( g(\boldsymbol h, \bar m)  \right) \right)  \right] \right \}^2 \\
\qb_2 &=  \mathbb{E}_1 \left \{ \frac{1}{\mathcal{N}_1}\left[\mathbb{E}_2 \left (  \mathcal{N}_2^{\frac{\theta_1}{\theta_2}-1} \cdots \mathbb{E}_{K+1} \left( \cosh^{\theta_K} \left( g(\boldsymbol h, \bar m) \right) \left( \tanh g(\boldsymbol h, \bar m)  \right) \right) \right ) \right]^2\right \} \\
\notag
... \\
\qb_{K+1} &=  \mathbb{E}_1 \left \{ \frac{1}{\mathcal{N}_1}\mathbb{E}_2 \left[ \mathcal{N}_2^{\frac{\theta_1}{\theta_2}-1} \cdots \mathbb{E}_{K+1} \left( \cosh^{\theta_K} \left (  g(\boldsymbol h, \bar m ) \right)  \tanh^2  \left( g(\boldsymbol h, \bar m)  \right) \right)  \right]\right \} \\
\bar{m} &= \mathbb{E}_1 \left \{ \frac{1}{\mathcal{N}_1}\mathbb{E}_2 \left[ \mathcal{N}_2^{\frac{\theta_1}{\theta_2}-1} \cdots \mathbb{E}_{K+1} \left(  \cosh^{\theta_K} \left ( g(\boldsymbol h, \bar m) \right) \tanh  \left ( g(\boldsymbol h, \bar m)  \right ) \right)  \right]\right \}
\end{align}
where $\boldsymbol h = (h^{(1)}, \cdots, h^{(K+1)})$, $g(\boldsymbol h, \bar m) = \beta \mb + \beta\sum_{a=1}^{K+1} \sqrt{\qb_a - \qb_{a-1}}h^{(a)}$, and $\mathcal{N}_1, ..., \mathcal{N}_{K+1}$ are defined in (\ref{eq:def_N}). 
\end{corollary}
\begin{proof}
The proof to achieve the expression of the self consistency equation for $\mb$ is identical to the previous cases, hence we omit it.
Conversely, for the other self-consistence equations, we extremize the pressure, the latter depending on $\{ x_0^{(a)} \}_{a=1,...,K+1}$, on $w_0$ and on the derivatives w.r.t. these parameters.
\newline
The derivatives w.r.t. $x_0^{(a)}$ consist of $K+2-a$ pieces, namely
\begin{align}
\partial_{x^{(a)}} \mathcal{A}_N = \sum_{j=1}^{K+2-a} A_j,
\end{align}
\begin{align}
\notag
\begin{cases}
A_1 = -\frac{\beta^2}{2} \theta_1 \left \{ \mathbb{E}_1 \left[\frac{1}{\mathcal{N}_1}\mathbb{E}_2 \left[ \mathcal{N}_2^{\frac{\theta_1}{\theta_2}-1} \cdots \mathbb{E}_{K+1}\left( \cosh^{\theta_K} ( g(\boldsymbol h, \bar m) ) \tanh ( g(\boldsymbol h, \bar m) ) \right) \right]\right]\right \}^2  \\
A_2 = -\frac{\beta^2}{2} (\theta_1-\theta_2)\mathbb{E}_1 \left \{ \frac{1}{\mathcal{N}_1}\left[\mathbb{E}_2 \left[ \mathcal{N}_2^{\frac{\theta_1}{\theta_2}-1} \cdots \mathbb{E}_{K+1}\left( \cosh^{\theta_K} ( g(\boldsymbol h, \bar m) ) \tanh ( g(\boldsymbol h, \bar m) ) \right) \right]\right]^2\right \} \\
... \\
A_{K+1} = \frac{\beta^2}{2} - \frac{\beta^2}{2}(1-\theta_K)\mathbb{E}_1 \left \{ \frac{1}{\mathcal{N}_1}\mathbb{E}_2 \left[ \mathcal{N}_2^{\frac{\theta_1}{\theta_2}-1} \cdots \mathbb{E}_{K+1}\left( \cosh^{\theta_K} (g(\boldsymbol h, \bar m)) \tanh^2 (g(\boldsymbol h, \bar m))  \right) \right]\right\}
\end{cases}
\end{align}
therefore,
\begin{align}
&\partial_{x^{(j+1)}} \mathcal{A}_N - \partial_{x^{(j)}} \mathcal{A}_N = - A_j = \frac{\beta^2}{2}(\theta_j-\theta_{j-1})\qb_j \ \ \ \ \ \ j=1,.., K\\
&\partial_{x^{(K+1)}} \mathcal A_N  = A_{K+1} = \frac{\beta^2}{2} - \frac{\beta^2}{2} (1-\theta_K) \qb_{K+1}
\end{align}
accounting for all $\qb$'s self-consistencies.
\end{proof}


\section{Main theme: the Hopfield neural network}\label{HopfieldSection}
The Hopfield model is the paradigmatic model for associative neural networks performing pattern recognition \cite{Amit,Coolen}. Hereafter, we solve for its quenched pressure by exploiting the approach based on Guerra's broken interpolating technique and the transport equation, already implemented for the Sherrington-Kirkpatrick model with a signal: in subsec.~\ref{ssec:HRS}, we will address the RS scenario and recover the AGS picture; in subsec.~\ref{ssec:HRSB1}, we will solve the model at the first level of RSB, rigorously proving the expression provided by Amit, Crisanti and Gutfreund \cite{Crisanti}, in the subsec.~\ref{ssec:HRSB2} we will solve the model at the second level of RSB, rigorously proving also the expression provided by Steffan and K\"{u}hn \cite{Kuhn}. The generalization to arbitrary, but finite, K steps is then presented, more succinctly, in subsec.~\ref{ssec:HOP_KRSB}, which closes the section.
\begin{definition} Set $\alpha \in \mathbb{R}^+$  and let $\boldsymbol \sigma \in \{- 1, +1\}^N$ be a configuration of $N$ binary neurons. Given $P=\alpha N$ random patterns $\{\boldsymbol \xi^{\mu}\}_{\mu=1,...,P}$, each made of $N$ i.i.d. digital entries drawn from probability $P(\xi_i^{\mu}=+1)=P(\xi_i^{\mu}=-1)=1/2$, for $i=1,...,N$, the Hamiltonian of the Hopfield model is defined as
	\beq
	\label{eq:hop_hbare}
	H_N(\boldsymbol \sigma| \boldsymbol \xi) \coloneqq -\frac{1}{2N}\sum_{\mu=1}^{P}\sum_{i,j=1}^{N,N}\xi_i^{\mu}\xi_j^{\mu}\sigma_i\sigma_j.
	\eeq
	\label{hop_hbare}
\end{definition}
\begin{definition}The partition function related to the Hamiltonian (\ref{eq:hop_hbare}) is given by
	\beq
	\label{eq:hop_BareZ}
	Z_N(\beta,\boldsymbol \xi) \coloneqq \sum_{ \boldsymbol \sigma } \exp \left[ -\beta H_N(\boldsymbol \sigma | \boldsymbol \xi) \right]=\sum_{ \boldsymbol \sigma } \exp \left( \frac{\beta}{2N}\sum_{\mu=1}^{P}\sum_{i,j=1}^{N,N}\xi_i^{\mu}\xi_j^{\mu}\sigma_i\sigma_j\right),
	\eeq
	where $\beta\in \mathbb R^+$ is the inverse temperature in proper units such that for $\beta \to 0$ the probability distribution for the neural configuration is uniformly spread while for $\beta \to \infty$ it is sharply peaked at the minima of the energy function (\ref{eq:hop_hbare}).
	\label{hop_BareZ}
\end{definition}
Analogously to the Sherrington-Kirkpatrick model, we introduce the \emph{Boltzmann average} induced by the partition function (\ref{eq:hop_BareZ}), denoted with $\omega_{\boldsymbol \xi}$ and, for an arbitrary observable $O(\boldsymbol \sigma)$, defined as
	\begin{equation}
	\omega_{\boldsymbol \xi} (O (\boldsymbol \sigma)) : = \frac{\sum_{\boldsymbol \sigma} O(\boldsymbol \sigma) e^{- \beta H_N(\boldsymbol \sigma| \boldsymbol \xi)}}{Z_N(\beta, \boldsymbol \xi)}.
	\end{equation}
	This can be further averaged over the realization of the $\xi_i^{\mu}$'s (also referred to as \emph{quenched average}) to get
	\beq
	\langle O(\boldsymbol \sigma) \rangle \coloneqq \mathbb{E} \omega_{\boldsymbol \xi} (O(\boldsymbol \sigma)).
	\eeq
	Further, we introduce the product state $\Omega_{s, _{\boldsymbol \xi}} = \omega_{\boldsymbol \xi}^{(1)} \times \omega_{\boldsymbol \xi}^{(2)} \times ... \times \omega_{\boldsymbol \xi}^{(s)}$ over $s$ replicas of the system, characterized by the same realization $\boldsymbol \xi$ of disorder. In the following, we shall use the product state over two replicas only, hence we shall neglect the index $s$ without ambiguity; also, to lighten the notation, we shall omit the subscript $\boldsymbol \xi$ in $\omega_{\boldsymbol \xi}$ and in $\Omega_{\boldsymbol \xi}$. Thus, for an arbitrary observable $O(\boldsymbol \sigma^{(1)}, \boldsymbol \sigma^{(2)})$
	\beq
	\langle O ( \boldsymbol \sigma^{(1)}, \boldsymbol \sigma^{(2)} ) \rangle \coloneqq \mathbb{E} \Omega (O (\boldsymbol \sigma^{(1)}, \boldsymbol \sigma^{(2)}) )= \mathbb{E}  \frac{\sum_{\boldsymbol \sigma} O(\boldsymbol \sigma^{(1)},\boldsymbol \sigma^{(2)}) e^{-\beta [H_N(\boldsymbol \sigma^{(1)}| \boldsymbol \xi) + H_N(\boldsymbol \sigma^{(2)}| \boldsymbol \xi)] }}{Z_N^2(\beta, \boldsymbol \xi)},
	\eeq
	where $\boldsymbol{\sigma}^{(1,2)}$ is the configuration pertaining to the replica labelled as $1,2$.
\begin{definition}The intensive quenched pressure of the Hopfield model (\ref{eq:hop_hbare}) is defined as
\begin{equation}
	A_N(\alpha, \beta) \coloneqq \frac{1}{N} \mathbb{E} \log Z_N(\beta, \boldsymbol \xi),
	\label{hop_BareA}
\end{equation}
and its thermodynamic limit, assuming its existence, is referred to as
\begin{equation}
A(\alpha, \beta) \coloneqq \lim_{N \to \infty} A_N(\alpha, \beta).
\end{equation}
\end{definition}
\begin{remark}
	In the following we shall exploit the {\em universality property} of the quenched noise in spin-glasses \cite{Univ1,Genovese}, namely, whatever the nature of the pattern entries (that is, digital -- e.g., Boolean -- or analog -- e.g., Gaussian), provided that their distribution is centered, symmetrical and with finite variance, they ultimately give analogous contribution to the structure of the quenched noise in the pressure (\ref{hop_BareA}) in the infinite volume limit ($N \to \infty$).
	Note that such a property is not guaranteed in the low storage regime, that is when $\lim_{N \to \infty} P/N=0$ \cite{Agliari-Barattolo,Albert1}.
\end{remark}

Focusing on pure state retrieval, we will assume without loss of generality that the candidate pattern to be retrieved $\boldsymbol \xi$ is a Boolean vector of $N$ entries, while $\boldsymbol \xi^{\mu}$, $\mu=1,...,P-1$ are real vectors whose $N$ entries are i.i.d. standard Gaussians. Accordingly, the average $\mathbb{E}$ acts as a Boolean average over $\boldsymbol \xi$ and as a Gaussian average over $\boldsymbol \xi^1 \cdots \boldsymbol \xi^{P-1}$.

\begin{definition} The order parameters used to describe the macroscopic behavior of the model are the standard ones \cite{Amit,Coolen}, namely, the Mattis magnetization 
	\begin{equation}
	m(\boldsymbol \sigma) :=m(\boldsymbol \sigma| \boldsymbol \xi) \coloneqq \frac{1}{N}\sum_{i=1}^{N} \xi_i \sigma_i
	\end{equation}
	to quantify the retrieval capabilities of the network, and the two-replica overlap in the $\boldsymbol \sigma$'s variables
\begin{equation}
\label{q}
q_{12}(\boldsymbol \sigma) \coloneqq \frac{1}{N}\sum_{i=1}^N \sigma_i^{(1)}\sigma_i^{(2)}
\end{equation}
to quantify the level of slow noise the network must cope with when performing pattern recognition.
Further, as an additional set of variables $\{\tau_{\mu}\}_{\mu=1,...,P-1}$ shall be introduced (vide infra), we accordingly define the related two-replica overlap
\begin{equation}
\label{p}
p_{11}(\boldsymbol \tau) \coloneqq  \frac{1}{P} \sum_{\mu=1}^P \tau^{(1)}_\mu \tau^{(1)}_\mu, ~~~ p_{12}(\boldsymbol \tau) \coloneqq  \frac{1}{P} \sum_{\mu=1}^P \tau^{(1)}_\mu \tau^{(2)}_\mu
\end{equation}
which as well captures the level of the noise due to pattern interference.
\label{hop_orderparameters}
\end{definition}
Note that, for the sake of simplicity, without loss of generality in the infinite volume limit (where $P\to \infty,\ N \to \infty$ such that $\alpha:=P/N \in \mathbb{R}^+$), we will approximate $P-1 \sim P$.

\subsection{Replica Symmetric Interpolation: RS solution} \label{ssec:HRS}
In this subsection we focus on the solution of the Hopfield model via interpolating techniques under the RS assumption. This route was already paved in \cite{AABF-NN2020}, yet it is reported (without proofs) hereafter for completeness and to allow the reader to get familiar with the technique before moving to the more challenging RSB scenario. Again, the strategy is to introduce an interpolating pressure $\mathcal A_N$ living in a fictitious space-time framework and recovering the intensive quenched pressure $A_N$ of the original model in a specific point of this space, and to show that it fulfills a transport equation in such a way that the solution of the statistical mechanical problem is recast in the solution of a partial differential equation.
\begin{definition}
	Under the replica-symmetry assumption, the order parameters, in the thermodynamic limit, self-average and their distributions get delta-peaked at their equilibrium value (denoted with a bar), independently of the replica considered, namely
	\begin{eqnarray}
	\lim_{N\to \infty} \langle (m - \bar m)^2 \rangle = 0 &\Rightarrow& \lim_{N\to \infty}  \langle m \rangle = \bar m\\
	\lim_{N\to \infty} \langle (q_{12} - \bar q)^2 \rangle = 0 &\Rightarrow& \lim_{N \to \infty}  \langle q_{12} \rangle = \bar q\\
	\lim_{N\to \infty} \langle (p_{12} - \bar pq)^2 \rangle = 0 &\Rightarrow& \lim_{N \to \infty}  \langle p_{12} \rangle = \bar p.
	\end{eqnarray}
	For the generic order parameter $X$ this can be rewritten as $\langle (\Delta X)^2 \rangle \overset{N\to\infty}{\longrightarrow}0$,
	where
	$$
	\Delta X \coloneqq  X - \bar{X},
	$$
	and, clearly, the RS approximation also implies that, in the thermodynamic limit, $\langle \Delta X \Delta Y \rangle = 0$ for any generic pair of order parameters $X,Y$.
\end{definition}
\begin{definition} Given the interpolating parameters $\bm r = (x, y, z, w), t$ to be set a posteriori and $N+P$ auxiliary quenched i.i.d. random variables $J_i \sim \mathcal{N}[0,1], i \in (1,...,N)$ and  $\tilde{J}_{\mu} \sim \mathcal{N}[0,1], \mu \in (1,...,P)$, the interpolating partition function for the Hopfield model (\ref{eq:hop_hbare}) is defined as
\footnotesize
\begin{eqnarray}
\mathcal{Z}_N(t, \bm r) &:=&  \sum_{\boldsymbol \sigma} \int \mathcal{D} \bm \tau \exp \left[ \beta \left( \frac{\sqrt{t}}{\sqrt{N}} \sum_{i,\mu =1}^{N,P}\xi_i^{\mu}\sigma_i\tau_\mu +\frac{t N}{2}m^2(\boldsymbol \sigma) + \sqrt{x} \sum_{i=1}^N J_i \sigma_i +  \sqrt{y} \sum_{\mu=1}^P \tilde J_\mu \tau_\mu + z \sum_{\mu=1}^P \frac{\tau_\mu^2}{2}+w N m(\boldsymbol \sigma)\right) \right ],
\label{hop_Z}
\end{eqnarray}
\normalsize
where, for any $\mu=1,...,P$, $\tau_{\mu} \sim \mathcal N [0, 1/\beta]$ and $\mathcal{D} \bm \tau \coloneqq \prod_{\mu=1}^P \frac{e^{- \beta \tau_{\mu}^2/2}}{\sqrt{2\pi \beta} }$ 
is the related measure.
\end{definition}
\begin{definition} The interpolating pressure for the classical Hopfield model (\ref{eq:hop_hbare}), at finite $N$, is introduced as
\begin{eqnarray}
\mathcal{A}_N(t, \bm r) &\coloneqq& \frac{1}{N} \mathbb{E} \left[  \log \mathcal{Z}_N(t, \bm r)  \right],
\label{hop_GuerraAction}
\end{eqnarray}
where the expectation $\mathbb E$ is now meant over $\boldsymbol \xi$, $\boldsymbol J$, and $\boldsymbol{\tilde{J}}$ and, in the thermodynamic limit,
\begin{equation}
\mathcal{A}(t, \bm r) \coloneqq \lim_{N \to \infty} \mathcal{A}_N(t, \bm r).
\label{hop_GuerraAction_TDL}
\end{equation}
By setting $t=1$ and $\bm r = \bm 0$ the interpolating pressure recovers the original one (\ref{hop_BareA}), that is $A_N (\alpha, \beta) = \mathcal{A}_N(t=1, \bm r = \bm 0)$.
\end{definition}
\begin{remark}
	The interpolating structure implies an interpolating measure, whose related Boltzmann factor reads as
	\begin{equation}
	\mathcal B (\boldsymbol \sigma, \boldsymbol \tau; t,\boldsymbol r )\coloneqq  \exp \left[ \beta \mathcal H (\boldsymbol \sigma, \boldsymbol \tau; t,\boldsymbol r ) \right],
	\end{equation}
with
\begin{equation}
 \mathcal H (\boldsymbol \sigma, \boldsymbol \tau; t,\boldsymbol r )  := \frac{\sqrt{t}}{\sqrt{N}} \sum_{i,\mu =1}^{N,P}\xi_i^{\mu}\sigma_i\tau_\mu +\frac{t N}{2}m^2(\boldsymbol \sigma) + \sqrt{x} \sum_{i=1}^N J_i \sigma_i + \sqrt{y} \sum_{\mu=1}^P \tilde J_\mu \tau_\mu + z \sum_{\mu=1}^P \frac{\tau_\mu^2}{2}+w N m(\boldsymbol \sigma).
\end{equation}
Clearly, $\mathcal Z_N(t, \boldsymbol r) = \int \mathcal{D} \bm \tau \sum_{\boldsymbol \sigma} \mathcal B (\boldsymbol \sigma, \bm \tau; t, \boldsymbol r)$.\\
A generalized average follows from this generalized measure as
\beq
	\omega_{t, \boldsymbol r} (O (\boldsymbol \sigma, \boldsymbol \tau )) \coloneqq  \int \mathcal{D} \bm \tau  \sum_{\boldsymbol \sigma} O (\boldsymbol \sigma , \boldsymbol \tau) \mathcal B (\boldsymbol \sigma, \boldsymbol \tau; t, \boldsymbol r)
	\eeq
	and
\beq
\langle O (\boldsymbol \sigma, \bm \tau ) \rangle_{t, \boldsymbol r}  \coloneqq \mathbb E [ \omega_{t, \boldsymbol r} (O (\boldsymbol \sigma, \bm \tau )) ].
\eeq
Of course, when $t=1$ and $\boldsymbol r = \bm 0$, the standard Boltzmann measure and related average is recovered.
Hereafter, in order to lighten the notation, we will drop the subindices $t, \boldsymbol r$.
\end{remark}
\begin{lemma} The partial derivatives of the interpolating pressure (\ref{hop_GuerraAction}) w.r.t. $t,x,y,z,w$ give the following expectation values:
	\bea
	\label{hop_expvalsa}
	\frac{\partial \mathcal{A}_N}{\partial t} &=& \frac{\alpha}{2} \big[ \langle p_{11} \rangle- \langle p_{12}q_{12} \rangle \big] + \frac{1}{2} \langle{m^2} \rangle \\
	\label{hop_expvalsmiddle}
	\frac{\partial \mathcal{A}_N}{\partial x}  &=& \frac{1}{2}  \big[1- \langle{q_{12}} \rangle \big],\\
	\frac{\partial \mathcal{A}_N}{\partial y}  &=& \frac{\alpha}{2} \big[\langle{p_{11}}\rangle-\langle{p_{12}}\rangle\big],\\
	\frac{\partial \mathcal{A}_N}{\partial z}  &=& \frac{\alpha}{2}  \langle{p_{11}}\rangle,\\
	\frac{\partial \mathcal{A}_N}{\partial w}  &=& \langle{m}\rangle.
	\label{hop_expvalsb}
	\eea
\end{lemma}
%
%
\begin{proposition} \label{prop:5}
The interpolating pressure \eqref{hop_GuerraAction} at finite size obeys the following differential equation:
	\begin{equation}
	\frac{d \mathcal{A}_N}{dt} =  \pder{\mathcal{A}_N}{t} + \dot{x} \pder{\mathcal{A}_N}{x} + \dot y \pder{\mathcal{A}_N}{y} + \dot z \pder{\mathcal{A}_N}{z} +\dot w \pder{\mathcal{A}_N}{w}= S(t, \bm r)+V_N(t, \bm r),
	\label{hop_GuerraAction_DE}
	\end{equation}
	where we set $\dot x = -\alpha \bar p$, $\dot y = - \bar q$, $\dot z = - (1-\bar q)$, $\dot w = -\bar m$ and
	\begin{eqnarray}
	S(t, \bm r) &\coloneqq& -\frac{1}{2} \bar{m}^2 -\frac{\alpha}{2} \bar{p} (1 -\bar q) \\ \label{potenziale-RS-Hopfield}
	V_N(t, \bm r) &\coloneqq& \frac{1}{2}\langle{(\Delta m)^2}\rangle-\frac{1}{2}\langle{\Delta p_{12}\Delta q_{12}}\rangle.
	\end{eqnarray}
	\end{proposition}
	\begin{proposition}
The transport equation associated to the interpolating pressure $\mathcal{A}_N(t, \bm r)$ in the thermodynamic limit and under the RS assumption is
	\begin{equation}
	\pder{\mathcal{A}_{\textrm{RS}}}{t}-\alpha \bar p \pder{\mathcal{A}_{\textrm{RS}}}{x}-\bar q \pder{\mathcal{A}_{\textrm{RS}}}{y}- (1-\bar q)\pder{\mathcal{A}_{\textrm{RS}}}{z} -\bar m \pder{\mathcal{A}_{\textrm{RS}}}{w}=-\frac{\alpha}{2}\bar p (1-\bar q ) - \frac{1}{2}\bar m^2,
	\label{hop_GuerraAction_RSDE}
	\end{equation}
	whose solution is given by
	\begin{eqnarray}
	\notag
	\mathcal{A}_{\textrm{RS}}(t,\boldsymbol r)&=&\log 2 + \mathbb{E}  \log\cosh\left( \bar{m}t + w + J\sqrt{\alpha \bar{p} t + x}\right) + \frac{\alpha}{2}\frac{y+ \bar{q} t}{[1-z- (1-\bar{q}t)]} \\
	&-&\frac{\alpha}{2}\log\left[1-z-(1-\bar{q}t)\right]+
	-\frac{1}{2}[\alpha	\bar p(t)(1-\bar{q}t) +\bar{m}^2 t].
	\label{hop_mechanicalsolution}
	\end{eqnarray}
\end{proposition}	

 \begin{corollary}
 The RS approximation of the quenched pressure for the Hopfield model is obtained by posing $t=1$ and $\bm r = \bm 0$ in (\ref{hop_mechanicalsolution}), which returns
 	\begin{eqnarray}
	\notag
	\small
		A_{\rm RS}(\alpha, \beta)&=&\log 2 + \mathbb{E} \log\cosh\Big[\beta \bar m+J\sqrt{\alpha \beta \bar p}\Big]-\frac{\beta}{2}[\alpha \bar p\big(1-\bar q\big) +\bar m^2 ]\\
		&+&\frac{\alpha}{2}\frac{\beta \bar q}{1-\beta [1-\bar q]} -\frac{\alpha}{2}\log\Big(1-\beta [1-\bar q]\Big).
\normalsize
	\label{hop_agssolution}
	\end{eqnarray}
 \end{corollary}

\begin{corollary}
	The self-consistency equations obtained from the quenched pressure \eqref{hop_agssolution} are
\begin{eqnarray}
\bar p&=& \frac{\beta \bar q}{\left[1-\beta(1-\bar q)\right]^2}\\
\bar m&=&\mathbb{E} \tanh\left(\beta \bar m+J\sqrt{\alpha\beta \bar p}\right)\\
\bar q&=&\mathbb{E} \tanh^2\left(\beta \bar m+J\sqrt{\alpha\beta \bar p}\right).
\label{hop_agsSCE}
\end{eqnarray}
\end{corollary}

\subsection{Broken Replica Interpolation: 1-RSB solution}\label{ssec:HRSB1}

In this subsection we turn to the solution of the Hopfield model via the generalized  broken-replica interpolating technique, restricting the description at the first step of RSB: the aim is to recover rigorously the expression provided by Cristanti, Amit and Gutfreund in the 80's via replica trick \cite{Crisanti}.
\newline
We anticipate that in the following 1RSB setting the probability distributions of the two overlaps $q$ and $p$ (see eq.s \ref{limforq2} and \ref{limforp2} respectively) display an analogous structure and, in particular, they display the same $\theta$. This choice emerges naturally in the current calculations and, also, it was somehow expected since the Hopfield model can be looked at as a bipartite spin-glass \cite{BarraEquivalenceRBMeAHN,bipartiti} and, for general multi-partite spin glasses, the ziqqurat ansatz \cite{ZiqquratBarra,ZiqquratPanchenko} is known to enlarge the Parisi scheme (the latter being recovered when collapsing the system to a single party).
\begin{definition} \label{def:HM_RSB}
In the first step of replica-symmetry breaking, the distribution of the two-replica overlap $q$, in the thermodynamic limit, displays two delta-peaks at the equilibrium values, referred to as $\bar{q}_1,\ \bar{q}_2$, and the concentration on the two values is ruled by $\theta \in [0,1]$, namely
\begin{equation}
\lim_{N \rightarrow + \infty} P'_N(q) = \theta \delta (q - \bar{q}_1) + (1-\theta) \delta (q - \bar{q}_2). \label{limforq2}
\end{equation}
Similarly, for the overlap $p$, denoting with $\bar{p}_1,\ \bar{p}_2$ the equilibrium values, we have
\begin{equation}
\lim_{N \rightarrow + \infty} P''_N(p) = \theta \delta (p - \bar{p}_1) + (1-\theta) \delta (p - \bar{p}_2). \label{limforp2}
\end{equation}
The magnetization still self-averages at $ \bar{m}$ as in (\ref{limform}).
\end{definition}
\begin{remark}
Scope of the present paper is to work out a mathematical method that is able to account for RSB in associative neural network, and -- a short discussion apart in the final section -- we are not going to go into physical implications: as a matter of fact, we preserve the historical ansatz and recover rigorously the expression for the model pressure as the steps of RSB take place, in particular, here and in the following, we always assume the Mattis magnetization to be self-averaging.
\end{remark}

Following the same route pursued in the previous sections, we need an interpolating partition function $\mathcal Z$ and an interpolating quenched pressure $\mathcal A$,  that are defined hereafter.
\begin{definition}
Given the interpolating parameters $\bm r = (x^{(1)}, x^{(2)}, y^{(1)}, y^{(2)}, w), t$ and the i.i.d. auxiliary fields $\{h_i^{(1)}, h_i^{(2)}\}_{i=1,...,N}$, with $h_i^{(1,2)} \sim \mathcal N[0,1]$ for $i=1, ..., N$ and $\{J_{\mu}^{(1)}, J_{\mu}^{(2)}\}_{\mu=1,...,P}$, with $J_{\mu}^{(1,2)} \sim \mathcal N[0,1]$ for $\mu=1,...,P$, we can write the 1-RSB interpolating partition function $\mathcal Z_N(t, \boldsymbol r)$ for the Hopfield model (\ref{eq:hop_hbare}) recursively, starting by
\begin{eqnarray}
\label{eqn:Z2}
\mathcal Z_2(t, \bm r) &=&\sum_{\bm \sigma} \int D \bm \tau \exp \left \{ \beta \left[ \frac{t}{2N}\sum_{i,j=1}^{N,N}\xi_i \xi_j\sigma_i \sigma_j+\sqrt{\frac{t}{N}}\sum_{i, \mu=1}^{N,P} \xi_i^\mu \sigma_i \tau_\mu + w \sum_{i=1}^N \xi_i\sigma_i \right. \right. \nonumber \\ &+&  \left. \left.  \sum_{a=1}^2 \sqrt{x^{(a)}}\sum_{i=1}^N h_i^{(a)} \sigma_i + \sum_{a=1}^2 \sqrt{y^{(a)}}\sum_{\mu=1}^P J_\mu^{(a)}\tau_\mu  +z \sum_{\mu=1}^P \frac{\tau_\mu^2}{2}  \right] \right\}
\end{eqnarray}
where $\bm \xi$ is the pattern we want to retrieve, while the $\xi_i^\mu$'s are i.i.d. standard Gaussians. 
Averaging out the fields recursively, we define
\begin{align}
\label{eqn:Z1}
\mathcal Z_1(t, \bm r) \coloneqq& \mathbb E_2 \left [ \mathcal Z_2(t, \bm r)^\theta \right ]^{1/\theta} \\
\label{eqn:Z0}
\mathcal Z_0(t, \bm r) \coloneqq&  \exp \mathbb E_1 \left[ \log \mathcal Z_1(t, \bm r) \right ] \\
\mathcal Z_N(t, \boldsymbol r) \coloneqq& \mathcal Z_0(t, \bm r) ,
\end{align}
where with $\mathbb E_a$ we mean the average over the variables $h_i^{(a)}$'s and $J_\mu^{(a)}$'s, for $a=1, 2$, and with $\mathbb{E}_0$ we shall denote the average over the variables $\xi_i^{\mu}$'s.
\end{definition}
\begin{definition}
The 1RSB interpolating pressure, at finite volume $N$, is introduced as
\begin{equation}\label{AdiSK1RSB}
\mathcal A_N (t, \bm r) \coloneqq \frac{1}{N}\mathbb E_0 \big[ \log \mathcal Z_0(t, \bm r) \big],
\end{equation}
and, in the thermodynamic limit, assuming its existence
\begin{equation}
\mathcal A (t, \bm r) \coloneqq \lim_{N \to \infty} \mathcal A_N (t, \bm r).
\end{equation}
By setting $t=1$, $\bm r = \bm 0$, the interpolating pressure recovers the standard pressure (\ref{PressureDef}), that is, $A_N(\alpha, \beta) = \mathcal A_N (t =1, \bm r =0)$.
\end{definition}
It is worth showing in details the previous equivalence, in fact,
\begin{align}
\mathcal A_N(\beta, \bm 0)= &\frac{1}{N} \mathbb E_0 \bigg [ \log \sum_{\bm \sigma} \int D \bm \tau \exp \bigg \{\frac{\beta}{2N}\sum_{i,j=1}^{N,N} \xi_i \xi_j\sigma_i \sigma_j +\sqrt{\frac{\beta}{N}}\sum_{i, \mu=1}^{N,P} \xi_i^\mu \sigma_i \tau_\mu\bigg\} \bigg]= \nonumber\\
=&\frac{1}{N} \mathbb E_0 \bigg [ \log \sum_{\bm \sigma}\prod_{\mu=1}^P \bigg( \int \frac{d\tau_\mu}{\sqrt{2\pi}} \exp \bigg \{\frac{\beta}{2N}\sum_{i,j=1}^{N,N} \xi_i \xi_j\sigma_i \sigma_j -\frac{\tau_\mu^2}{2}+\tau_\mu\sqrt{\frac{\beta}{N}}\sum_{i=1}^N \xi_i^\mu \sigma_i \bigg\}\bigg) \bigg] =\nonumber\\
=&\frac{1}{N} \mathbb E_0 \bigg [ \log \sum_{\bm \sigma} \exp \bigg \{\frac{\beta}{2N}\sum_{i,j=1}^{N,N} \xi_i \xi_j\sigma_i \sigma_j +{\frac{\beta}{2N}}\sum_{i, j, \mu=1}^{N,N,P} \xi_i^\mu \xi_j^\mu \sigma_i  \sigma_j \bigg\} \bigg]
\end{align}
which is exactly the Hopfield pressure.
\begin{remark}
Analously to what done before for the Sherrington-Kirkpatrick model with a signal, we introduce the thermal average
\begin{equation}
\omega \big(\mathcal O(\bm \sigma, \bm \tau)\big) = \frac{1}{\mathcal Z_2(t, \bm r)}\sum_{\bm \sigma} \int D \bm \tau \mathcal O(\bm \sigma, \bm \tau) e^{\beta \mathcal H(\bm \sigma,  \bm \tau; t, \bm r)}
\end{equation}
where $\mathcal H(\bm \sigma,  \bm \tau; t, \bm r)$ depends also on all the interpolating and random variables, i.e.
\begin{align}
\mathcal H(\bm \sigma,  \bm \tau; t, \bm r)=&\frac{t}{2N}\sum_{i,j=1}^{N,N}\xi_i \xi_j\sigma_i \sigma_j+\sqrt{\frac{t}{N}}\sum_{i, \mu=1}^{N,P} \xi_i^\mu \sigma_i \tau_\mu +\sum_{a=1}^2 \sqrt{x^{(a)}}\sum_{i=1}^N h_i^{(a)}	\sigma_i \nonumber \\ +& \sum_{a=1}^2 \sqrt{y^{(a)}}\sum_{\mu=1}^P J_\mu^{(a)}\tau_\mu +z\sum_{\mu=1}^P \frac{\tau_\mu^2}{2}+ w \sum_{i=1}^N \xi_i\sigma_i.
\end{align}
\end{remark}
\begin{remark}
In order to lighten the notation, hereafter we use the following 
\begin{align}
\label{eq:unouno_a}
\langle m \rangle=& \mathbb E_0  \mathbb E_1  \mathbb E_2 \left[\mathcal W_2\frac{1}{N}\sum_{i=1}^N \omega(\xi_i \sigma_i) \right] \\
\langle m^2 \rangle=&  \mathbb E_0  \mathbb E_1  \mathbb E_2 \left[ \mathcal W_2\frac{1}{N^2}\sum_{i,j=1}^{N,N} \omega(\xi_i \xi_j\sigma_i \sigma_j) \right] \\
 \langle p_{11} \rangle=& \mathbb E_0  \mathbb E_1  \mathbb E_2 \left[ \mathcal W_2\frac{1}{P}\sum_{\mu=1}^P \omega(\tau_\mu^2) \right]\\
 \langle p_{12} \rangle_1=& \mathbb E_0  \mathbb E_1   \left[ \frac{1}{P}\sum_{\mu=1}^P \left( \mathbb E_2  \left [\mathcal W_2\omega(\tau_\mu)\right] \right)^2 \right] \\
\langle p_{12} \rangle_2=& \mathbb E_0  \mathbb E_1  \mathbb E_2 \left[ \mathcal W_2\frac{1}{P}\sum_{\mu=1}^P \omega(\tau_\mu)^2 \right] \\
\label{eqn:q121_a}
 \langle q_{12} \rangle_1=&\mathbb E_0  \mathbb E_1  \left[\frac{1}{N} \sum_{i=1}^N \left( \mathbb E_2 \left[\mathcal W_2\omega(\sigma_i)\right] \right)^2 \right] \\
\label{eqn:q122_a}
\langle q_{12} \rangle_2=& \mathbb E_0  \mathbb E_1  \mathbb E_2 \left [ \mathcal W_2\frac{1}{N}\sum_{i=1}^N \omega(\sigma_i)^2 \right] \\
\langle p_{12}q_{12} \rangle_1 =& \mathbb E_0  \mathbb E_1  \left [ \frac{1}{P}\sum_{\mu=1}^P \frac{1}{N}\sum_{i=1}^N \left(  \mathbb E_2 \left [ \mathcal W_2 \omega(\tau_\mu \sigma_i ) \right] \right)^2 \right] \\
\label{eq:duedue_a}
\langle p_{12} q_{12}\rangle_2=& \mathbb E_0  \mathbb E_1  \mathbb E_2 \bigg[\mathcal W_2\frac{1}{P}\sum_\mu\frac{1}{N}\sum_{i=1}^N \omega(\tau_\mu \sigma_i)^2 \bigg],
\end{align}
where we define the weight
\begin{equation}
\mathcal W_2=\frac{\mathcal Z_2^\theta}{\mathbb E_2 \left [\mathcal Z_2^\theta\right]}.
\end{equation}
\end{remark}
The next step is building a transport equation for the interpolating quenched pressure, for which we preliminary need to evaluate the related partial derivatives, as discussed in the next
\begin{lemma} \label{lemma:4}
The partial derivative of the interpolating quenched pressure with respect to a generic variable $\rho$ reads as
%
%
%
\begin{equation}
\label{eqn:partialrA}
\frac{\partial }{\partial \rho} \mathcal A_N(t, \bm r)=\frac{1}{N} \mathbb E_0  \mathbb E_1  \mathbb E_2 \left[\mathcal W_2 \omega \big( \partial_\rho \mathcal B(\bm \sigma,  \bm \tau; t, \bm r) \big)\right].
\end{equation}
In particular,
\begin{align}
\label{eqn:partialtA}
\frac{\partial }{\partial t} \mathcal A_N =&\frac{1}{2}\langle m^2 \rangle+ \frac{\alpha}{2}\big ( \langle p_{11} \rangle -(1-\theta)\langle p_{12}q_{12} \rangle_2-\theta\langle p_{12}q_{12} \rangle_1 \big) \\
\label{eqn:partialx1A}
\frac{\partial }{\partial x^{(1)}} \mathcal A_N =& \frac{1}{2}\big ( 1 -(1-\theta)\langle q_{12} \rangle_2-\theta\langle q_{12} \rangle_1 \big) \\
\label{eqn:partialx2A}
\frac{\partial }{\partial x^{(2)}} \mathcal A_N =& \frac{1}{2}\big ( 1 -(1-\theta)\langle q_{12} \rangle_2 \big) \\
\label{eqn:partialy1A}
\frac{\partial }{\partial y^{(1)}} \mathcal A_N =& \frac{\alpha}{2}\big (  \langle p_{11} \rangle -(1-\theta)\langle p_{12} \rangle_2-\theta\langle p_{12} \rangle_1 \big) \\
\label{eqn:partialy2A}
\frac{\partial }{\partial y^{(2)}} \mathcal A_N =& \frac{\alpha}{2}\big (  \langle p_{11} \rangle -(1-\theta)\langle p_{12} \rangle_2
\big)\\
\label{eqn:partialzA}
\frac{\partial }{\partial z} \mathcal A_N =& \frac{\alpha}{2} \langle p_{11} \rangle \\
\label{eqn:partialwA}
\frac{\partial }{\partial w}  \mathcal A_N =& \langle m \rangle
\end{align}
\end{lemma}
\begin{proof}
The proof is pretty lengthy and basically requires just standard calculations, so it is left for the Appendix \ref{app2}. Here we just prove that, in  complete generality
\begin{align}
\label{eqn:partialrA1}
\frac{\partial }{\partial \rho}  \mathcal A_N(t, \bm r)=& \frac{1}{N} \mathbb E_0 \mathbb E_1 \bigg [\partial_\rho \log\mathcal Z_1 \bigg] \nonumber \\
=&\frac{1}{N} \mathbb E_0 \mathbb E_1 \bigg[\frac{1}{\theta}\frac{1}{\mathcal Z_1} \big[ \mathcal Z_2^\theta \big]^{1/\theta-1} \mathbb E_2 \big[\partial_\rho\mathcal Z_2^\theta \big] \bigg] \nonumber \\
=&\frac{1}{N} \mathbb E_0 \mathbb E_1 \mathbb E_2\bigg[\frac{\mathcal Z_2^\theta}{\mathbb E_2 \mathcal Z_2^\theta }\frac{\partial_\rho \mathcal Z_2}{\mathcal Z_2} \bigg]  \nonumber\\
=&\frac{1}{N} \mathbb E_0 \mathbb E_1 \mathbb E_2 \bigg[\mathcal W_2\frac{\partial_\rho \mathcal Z_2}{\mathcal Z_2} \bigg].
\end{align}
\end{proof}

\begin{proposition}
\label{prop:9}
The streaming of the 1-RSB interpolating quenched pressure obeys, at finite volume $N$, a standard transport equation, that reads as
\begin{align}
\label{eqn:transportequation}
\frac{d\mathcal A}{dt}&=\partial_t \mathcal A+\dot x^{(1)}\partial_{x_1} \mathcal A +\dot x^{(2)}\partial_{x_2} \mathcal A +\dot y^{(1)} \partial_{y_1} \mathcal A +\dot y^{(2)} \partial_{y_2} \mathcal A +\dot z \partial_{z} \mathcal A+\dot w \partial_{w} \mathcal A \nonumber \\
&= S(t, \bm r) + V_N(t, \bm r)
\end{align}
where
\begin{align}
\label{eqn:f}
 S(t, \bm r)  \coloneqq & -\frac{\bar m^2}{2}-\frac{\alpha}{2} \bar p_2 (1-\bar q_2)-\frac{\alpha}{2} \theta (\bar p_2 \bar q_2 - \bar p_1 \bar q_1) \\
\label{eqn:V}
V_N(t, \bm r) \coloneqq & \frac{1}{2} \langle (m-\mb)^2 \rangle - \frac{\alpha}{2}(1-\theta) \langle \Delta p_{12}\Delta q_{12} \rangle_2 - \frac{\alpha}{2} \theta \langle \Delta p_{12}\Delta q_{12} \rangle_1
\end{align}
\end{proposition}
\begin{proof}
Similarly to Proposition \ref{prop:3} for the Sherrington-Kirkpatrick model with a signal, we have
\begin{align}
\langle \Delta \bar{p}_{12} \Delta \qb_{12} \rangle_a = \langle \bar{p}_{12} \qb_{12} \rangle_a + \bar{p}_a \qb_a - \qb_a \langle p_{12} \rangle - \bar{p}_a \langle q_{12} \rangle \ \ \forall a=1, 2.
\end{align}
Now, starting to evaluate explicitly $\dt \mathcal{A}_N$ by using (\ref{eqn:partialx1A}), (\ref{eqn:partialy1A}), (\ref{eqn:partialzA}) and (\ref{eqn:partialwA}) we write
\begin{align}
\dt \mathcal{A}_N= &\frac{1}{2} \left[ \langle (m-\mb)^2 \rangle - \mb^2 +2\mb \langle m \rangle \right] + \frac{\alpha}{2} \Big[ \langle p_{11} \rangle - (1-\theta) \Big (\langle \Delta p_{12}\Delta q_{12} \rangle_2 - \bar{p}_2\qb_2 + \bar{p}_2\langle q_{12} \rangle_2 + \notag \\
&+ \bar{q}_2\langle p_{12} \rangle_2 \Big) - \theta \Big(\langle \Delta p_{12}\Delta q_{12} \rangle_1 - \bar{p}_1\qb_1 + \bar{p}_1\langle q_{12} \rangle_1+ \bar{q}_1\langle p_{12} \rangle_1 \Big) \Big]= \notag \\
&=\frac{1}{2} \langle (m-\mb)^2 \rangle - \frac{1}{2} \mb^2 + \mb \dw \mathcal{A}_N + \frac{\partial}{\partial z} \mathcal{A}_N - \frac{\alpha}{2}(1-\theta) \langle \Delta p_{12}\Delta q_{12} \rangle_2 - \frac{\alpha}{2} \theta \langle \Delta p_{12}\Delta q_{12} \rangle_1 + \notag \\
&+ \frac{\alpha}{2}(1-\theta) \bar{p}_2\qb_2 +\frac{\alpha}{2}\bar{p}_1\qb_1 - \frac{\alpha}{2}(1-\theta)\bar{p}_2\langle q_{12} \rangle_2 + \alpha \bar{p}_1 \left( \frac{\partial}{\partial x^{(1)}}\mathcal{A}_N - \frac{1}{2} - \frac{1}{2}(1-\theta)\langle q_{12} \rangle_2\right) - \notag \\
&- \frac{\alpha}{2}(1-\theta)\bar{q}_2\langle p_{12} \rangle_2 + \bar{q}_1 \left( \frac{\partial}{\partial y^{(1)}}\mathcal{A}_N - \frac{\alpha}{2}\langle p_{11} \rangle - \frac{\alpha}{2}(1-\theta)\langle p_{12} \rangle_2\right)
\end{align}
In the same way we use (\ref{eqn:partialx2A}) and (\ref{eqn:partialy2A})
\begin{align}
\dt \mathcal{A}_N =& \frac{1}{2} \langle (m-\mb)^2 \rangle - \frac{1}{2} \mb^2 + \mb \dw \mathcal{A}_N + \frac{\partial}{\partial z} \mathcal{A}_N - \frac{\alpha}{2}(1-\theta) \langle \Delta p_{12}\Delta q_{12} \rangle_2 - \frac{\alpha}{2} \theta \langle \Delta p_{12}\Delta q_{12} \rangle_1 + \notag \\
&+ \frac{\alpha}{2}(1-\theta) \bar{p}_2\qb_2 +\frac{\alpha}{2}\bar{p}_1\qb_1 - \frac{\alpha}{2}(1-\theta)\bar{p}_2\langle q_{12} \rangle_2 + \alpha \bar{p_1} \frac{\partial}{\partial x^{(1)}} \mathcal{A}_N - \frac{\alpha}{2}\bar{p}_1 + \alpha(\bar{p}_2 - \bar{p}_1) \cdot \notag \\
&\cdot \left( \frac{\partial }{\partial x^{(2)}} \mathcal{A}_N - \frac{1}{2} \right) + \qb_1 \frac{\partial}{\partial y^{(1)}} \mathcal{A}_N - \frac{\alpha}{2}\qb_1\langle p_{11} \rangle - (\qb_2-\qb_1) \left( \frac{\partial }{\partial y^{(2)}} \mathcal{A}_N - \frac{\alpha}{2}\langle p_{11} \rangle\right)
\end{align}
After some algebra, by placing
\begin{align}
 \label{eqn:dotx1}
\dot{x}^{(1)}&= - \alpha \bar{p}_1 \\
\dot{x}^{(2)}&= - \alpha (\bar{p}_2-\bar{p}_1) \\
\dot{y}^{(1)}&= - \bar{q}_1 \\
\dot{y}^{(2)}&= - \bar{q}_2-\bar{q}_1 \\
\dot{z}&= -(1-\qb_2)\\
\dot{w} &= - \mb \label{eq:dotw}
\end{align}
we reach the thesis.

\end{proof}

\begin{remark} \label{r:above}
In the thermodynamic limit, in the 1RSB scenario, we have
\begin{align}
\label{eqn:thlimaveragem}
 \lim_{N\rightarrow \infty} \langle (m - \bar m)^2 \rangle =& 0\\
 \lim_{N\rightarrow \infty} \langle (q_{12}-\bar q_i)^2 \rangle_i=& 0; \: \: \: i=1,2\\
\label{eqn:thlimaveragep}
 \lim_{N\rightarrow \infty} \langle( p_{12}-\bar p_i)^2\rangle_i =& 0; \: \: \: i=1,2
\end{align}
The potential (\ref{eqn:V}) can be written as
\begin{equation}
\label{eqn:compactV}
V_N(t, \bm r) = \frac{1}{2}\bigg \{\langle (m-\bar m)^2 \rangle - \alpha (1-\theta)\langle (q_{12}-\bar q_2) (p_{12}-\bar p_2)\rangle_2 -\alpha \theta \langle (q_{12}-\bar q_1) (p_{12}-\bar p_1)\rangle_1 \bigg\}
\end{equation}
in such a way that
\begin{equation} \label{eq:V0_HRSB}
\lim_{N \to \infty} V_N(t, \bm r) = 0.
\end{equation}
\end{remark}
The approximation we achieve by killing the potential is equivalent to requiring the existence of two temporal scales for thermalization, a slow one and a fast one, and the self-averaging within each time scale, in such a way that if two replicas behave the same on both the timescales we get $\langle . \rangle_2$, while if they match on the fast one but they do not on the slow one then we get $\langle . \rangle_1$; we refer to Section \ref{comment} for a deeper discussion on the physics behind this choice.

Exploiting Remark \ref{r:above} we can prove the following
\begin{proposition} \label{propHRSB}
The transport equation associated to the interpolating pressure function $\mathcal A_N(t, \boldsymbol r)$, in the thermodynamic limit and under the 1RSB assumption, can be written as
\begin{equation}
\label{eqn:solutionzeroV}
\mathcal A_{\textrm{1RSB}}(t, \bm r)=\mathcal A_{\textrm{1RSB}}(0, \bm{r}-\bm{\dot r}t)+S(t, \bm r )t
\end{equation}
whose explicit solution is given by
\begin{align}
\mathcal A_{1RSB}=&\log 2+\frac{1}{\theta}\int Dh^{(1)} \log  \int Dh^{(2)}\cosh^\theta \left(  h^{(1)} \sqrt{x_0^{(1)}}+h^{(2)} \sqrt{x_0^{(2)}} + w_0\right) + \frac{\alpha}{2\theta}\log\bigg(1+\frac{\theta y_0^{(2)}}{1-z_0-\theta y_0^{(2)}}\bigg) - \nonumber \\
&-\frac{\alpha}{2}\log(1-z_0)
+\frac{\alpha}{2}\frac{ y_0^{(1)}}{(1-z_0-\theta y_0^{(2)})} - t\frac{\mb^2}{2} -\frac{\alpha}{2} t \bar{p}_2(1-\qb_2) - \frac{\alpha}{2}t\theta(\bar{p}_2\qb_2 - \bar{p}_1\qb_1)
\end{align}
\end{proposition}
\begin{proof}
By putting (\ref{eqn:dotx1})-(\ref{eq:dotw}) into (\ref{eqn:transportequation}) we find
\begin{equation}
\label{eqn:solutionzeroV1}
\mathcal A_0(t, \bm r)=\mathcal A_0(0, \bm r_0) -\frac{\bar m^2}{2}t-\frac{\alpha}{2} \bar p_2 (1-\bar q_2)t-\frac{\alpha}{2} \theta (\bar p_2 \bar q_2 - \bar p_1 \bar q_1)t
\end{equation}
where $r_0$ can be obtained by using the equation of motion
\begin{equation}
\label{eqn:linearmotion}
\bm r = \bm r_0 + \dot{\bm r} t
\end{equation}
where the velocities are defined in (\ref{eqn:dotx1})-(\ref{eq:dotw}). Then, all we have to compute is $\mathcal A_0(0, \bm r_0)$, that can be easily done because at $t=0$ the two body interaction vanishes and the (\ref{eqn:Z2}) can be written as
\begin{align}
\label{eqn:A0fin}
\mathcal A_0(0, \bm r_0)=&\log 2+\frac{1}{\theta}\int Dh^{(1)} \log  \int Dh^{(2)}\cosh^\theta \big [ h^{(1)} \sqrt{x_0^{(1)}}+h^{(2)} \sqrt{x_0^{(2)}} + w_0\big ] + \nonumber \\
&+\frac{\alpha}{2\theta}\log\bigg(1+\frac{\theta y_0^{(2)}}{1-z_0-\theta y_0^{(2)}}\bigg)
-\frac{\alpha}{2}\log(1-z_0)
+\frac{\alpha}{2}\frac{ y_0^{(1)}}{(1-z_0-\theta y_0^{(2)})}
\end{align}
and we refer to Appendix \ref{1body1rsb} for a detailed proof of this result.
\newline
Then, putting together (\ref{eqn:A0fin}), (\ref{eqn:solutionzeroV1}), (\ref{eqn:linearmotion}) and (\ref{eqn:dotx1})-(\ref{eq:dotw}) we finally achieve an explicit expression for the interpolating pressure of the Hopfield model in the 1RSB approximation
\begin{eqnarray}
\label{eqn:Afinal}
&&\mathcal A_0(t, \bm r)=\frac{1}{\theta}\int Dh^{(1)} \log  \int Dh^{(2)}\cosh^\theta \left( h^{(1)} \sqrt{x^{(1)}+\alpha \bar p_1 t}+ h^{(2)}\sqrt{x^{(2)}+\alpha(\bar p_2 - \bar p_1)t} + w+ \bar m t\right )  \nonumber \\
\nonumber
&+&\log 2+\frac{\alpha}{2\theta}\log\left[1+\theta \frac{y^{(2)}+(\bar q_2-\bar q_1)t}{1-z-(1-\bar q_2)t-\theta (y^{(2)}+(\bar q_2-\bar q_1)t)}\right] -\frac{\alpha}{2}\log\left[1-z-(1-\bar q_2)t\right]\\
&+&\frac{\alpha}{2}\frac{ y^{(1)}+\bar q_1 t}{1-z-(1-\bar q_2)t-\theta (y^{(2)}+(\bar q_2-\bar q_1)t)}-\frac{\bar m^2}{2}t-\frac{\alpha}{2} \bar p_2 (1-\bar q_2)t-\frac{\alpha}{2} \theta (\bar p_2 \bar q_2 - \bar p_1 \bar q_1)t.
\end{eqnarray}
\end{proof}
To sum up, we have the following main theorem for the 1RSB scenario
\begin{theorem}
The 1-RSB quenched pressure for Hopfield model, in the thermodynamic limit, reads as
\begin{align}
\label{eqn:hopfieldAfinal}
  A(\alpha, \beta)=&\frac{1}{\theta}\int Dh^{(1)} \log  \int Dh^{(2)}\cosh^\theta  (g( \boldsymbol h, \bar m)) +\log 2 \nonumber \\
  \nonumber
&+\frac{\alpha}{2\theta}\log\left[ 1+\beta\theta \frac{\bar q_2-\bar q_1}{1-(1-\bar q_2)\beta-\theta \beta(\bar q_2-\bar q_1)}\right] -\frac{\alpha}{2}\log\left[1-\beta(1-\bar q_2)\right] \\
&+\frac{\alpha \beta}{2}\frac{\bar q_1 }{1-\beta(1-\bar q_2)-\theta \beta(\bar q_2-\bar q_1)}-\beta\frac{\bar m^2}{2}-\frac{\alpha \beta}{2}\bar p_2 (1-\bar q_2)-\frac{\alpha \beta}{2}\theta (\bar p_2 \bar q_2 - \bar p_1 \bar q_1)
\end{align}
where $\boldsymbol h = (h^{(1)}, h^{(2)})$, and we introduced the function $g( \boldsymbol h, \bar m)=\beta \bar m+  h^{(1)} \sqrt{\alpha \beta \bar p_1}+ h^{(2)}\sqrt{\alpha\beta(\bar p_2 - \bar p_1)}$.
\end{theorem}
\begin{proof}
By taking $\bm r= \boldsymbol 0$ and $t=\beta$ we find the Hopfield pressure in the 1RSB approximation.
\end{proof}

\begin{corollary}
\label{cor:SC_HOP_1RSB}
The self-consistent equations for the order parameters are
\begin{align}
\bar m  =&\int Dh^{(1)}\frac{\int Dh^{(2)} \cosh^\theta (g( \boldsymbol h, \bar m)) \tanh  (g(\boldsymbol h, \bar m)) }{\int Dh^{(2)} \cosh^\theta (g(\boldsymbol h, \bar m))} \\
\bar q_1 =&\int Dh^{(1)}\left(\frac{\int Dh^{(2)}\tanh  g(h^{(1)}, h^{(2)})\cosh^\theta  (g(\boldsymbol h, \bar m))}{\int Dh^{(2)} \cosh^\theta (g(\boldsymbol h, \bar m))}\right)^2 \\
\label{eqn:barq2}
\bar q_2 =& \int Dh^{(1)}\frac{\int Dh^{(2)} \cosh^\theta (g(\boldsymbol h, \bar m)) \tanh^2 (g(\boldsymbol h, \bar m)) }{\int Dh^{(2)} \cosh^\theta  (g(\boldsymbol h, \bar m))} \\
\bar p_1=& \frac{\beta \bar q_1}{\left[1-\beta(1-\bar q_2)-\beta\theta (\bar q_2-\bar q_1)\right]^2}\\
\bar p_2 =& \bar p_1 +\frac{\beta (\bar q_2-\bar q_1)}{\left[ 1-\beta(1-\bar q_2)\right] \left[ 1-\beta(1-\bar q_2)-\beta\theta (\bar q_2-\bar q_1)\right]}.
\end{align}
\end{corollary}
\begin{proof}
We proceed in the same way we computed the proof of Corollary (\ref{cor:SC_SK}), that is by taking derivatives of (\ref{eqn:Afinal}) and putting $t=\beta, \bm r= \boldsymbol 0$ so to find \\

\begin{align}
\partial_{x^{(1)}} \mathcal A_0(\beta, \bm 0)=&\frac{1}{2}\left [  1-\theta \int Dh^{(1)}\left(\frac{\int Dh^{(2)}\tanh (g(\boldsymbol h, \bar m)) \cosh^\theta  (g(\boldsymbol h, \bar m))}{\int Dh^{(2)} \cosh^\theta (g(\boldsymbol h, \bar m))}\right)^2 + \right. \nonumber \\
&\left. -(1-\theta)\int Dh^{(1)}\frac{\int Dh^{(2)}\tanh^2 (g(\boldsymbol h, \bar m)) \cosh^\theta(g(\boldsymbol h, \bar m))}{\int Dh^{(2)} \cosh^\theta (g(\boldsymbol h, \bar m))} \right] \\
\partial_{x^{(2)}} \mathcal A_0(\beta, \bm 0)=&\frac{1}{2}\left [ 1-(1-\theta)\int Dh^{(1)}\frac{\int Dh^{(2)}\tanh^2(g(\boldsymbol h, \bar m)) \cosh^\theta  (g(\boldsymbol h, \bar m))}{\int Dh^{(2)} \cosh^\theta  (g(\boldsymbol h, \bar m))}\right ]\\
\partial_{y^{(1)}} \mathcal A_0(\beta, \bm 0)=&\frac{\alpha}{2}\frac{1}{1-\beta(1-\bar q_2)-\beta\theta (\bar q_2-\bar q_1)}\\
\partial_{y^{(2)}} \mathcal A_0(\beta, \bm 0)=&\frac{\alpha\beta}{2}\frac{\theta \bar q_1}{\big(1-\beta(1-\bar q_2)-\beta\theta (\bar q_2-\bar q_1)\big)^2}+\partial_{y_1}\mathcal A_0(\beta, \bm 0)\\
\partial_{z} \mathcal A_0(\beta, \bm 0)=&-\frac{\alpha}{2\theta}(1-\theta)\frac{1}{1-\beta(1-\bar q_2)}+\frac{1}{\theta}\partial_{y_2} \mathcal A_0(\beta, \bm 0)\\
\partial_w \mathcal A_0(\beta, \bm 0)=&\int Dh^{(1)}\frac{\int Dh^{(2)}\tanh(g(\boldsymbol h, \bar m)) \cosh^\theta(g(\boldsymbol h, \bar m))}{\int Dh^{(2)} \cosh^\theta( g(\boldsymbol h, \bar m))}
\end{align}
Putting the derivatives inside (\ref{eqn:dotx1})-(\ref{eq:dotw}), replacing the averages with their asymptotic values in the thermodynamic limit (\ref{eqn:thlimaveragem})-(\ref{eqn:thlimaveragep}), we finally obtain the above self-consistency equations.
\end{proof}

\subsection{Broken Replica Interpolation: 2-RSB solution}\label{ssec:HRSB2}

Aim of this section is to deepen the structure of the second step of RSB in order to recover rigorously the formula for the quenched free energy, obtained via the replica trick by Stefann and K\"{u}hn \cite{Kuhn}, for the Hopfield model; as far as we know an explicit formula also for the 3-RSB is not available and, in the next subsection, we will provide the general K-RSB solution, without deepening all the calculations as in these first two steps.
\begin{definition} \label{def:Htheta_2RSB}
In the second step of replica-symmetry breaking, the distribution of the two-replica overlap $q$, in the thermodynamic limit, displays three delta-peaks at the equilibrium values, referred to as $\bar{q}_1,\ \bar{q}_2$ and $\bar{q}_3$, and the concentration on the three values is ruled by $\theta_2, \theta_3 \in [0,1]$, namely
	\begin{equation}
	\lim_{N \rightarrow + \infty} P'_N(q) = \theta_1 \delta (q - \bar{q}_1) + (\theta_2 -\theta_1) \delta (q - \bar{q}_2)+(1 -\theta_2) \delta (q - \bar{q}_3). \label{limforq2}
	\end{equation}
	Similarly, for the overlap $p$, denoting with $\bar{p}_1,\ \bar{p}_2$ and $\bar{p}_3$ the related equilibrium values, we have
	\begin{equation}
	\lim_{N \rightarrow + \infty} P''_N(p) = \theta_1 \delta (p - \bar{p}_1) + (\theta_2 -\theta_1) \delta (p - \bar{p}_2)+(1 -\theta_2) \delta (p - \bar{p}_3). \label{limforp2}
	\end{equation}
	The magnetization still self-averages at $ \bar{m}$ as in (\ref{limform}).
\end{definition}

Following the same route pursued in the previous sections, we need an interpolating partition function $\mathcal Z$ and an interpolating quenched pressure $\mathcal A$, defined as follows.
\begin{definition}
	Given the interpolating parameters $\bm r = (x^{(1)}, x^{(2)}, x^{(3)}, y^{(1)}, y^{(2)}, y^{(3)}, w), t$ and the i.i.d. auxiliary fields $\{h_i^{(1)}, h_i^{(2)}, h_i^{(3)} \}_{i=1,...,N}$, with $h_i^{(1,2,3)}\sim \mathcal N[0,1]$ for $i=1,...,N$, and $\{J_{\mu}^{(1)}, J_{\mu}^{(2)}, J_{\mu}^{(3)} \}_{\mu=1,...,P}$, with $J_{\mu}^{(1,2,3)} \sim \mathcal N[0,1]$ for $\mu=1,...,P$, we can write the 2-RSB interpolating partition function $\mathcal Z_N(t, \boldsymbol r)$ for the Hopfield model (\ref{eq:hop_hbare}) recursively, starting by
	\begin{eqnarray}
	\label{eqn:Z3_2}
	\mathcal Z_3(t, \bm r) &=&\sum_{\bm \sigma} \int D \bm \tau \exp \left \{ \beta \left[ \frac{t}{2N}\sum_{i,j=1}^{N,N}\xi_i \xi_j\sigma_i \sigma_j+\sqrt{\frac{t}{N}}\sum_{i, \mu=1}^{N,P} \xi_i^\mu \sigma_i \tau_\mu + w \sum_{i=1}^N \xi_i\sigma_i \right. \right. \nonumber \\ &+&  \left. \left.  \sum_{a=1}^3 \sqrt{x^{(a)}}\sum_{i=1}^N h_i^{(a)} \sigma_i + \sum_{a=1}^3 \sqrt{y^{(a)}}\sum_{\mu=1}^P J_\mu^{(a)}\tau_\mu  +z \sum_{\mu=1}^P \frac{\tau_\mu^2}{2}  \right] \right\}
	\end{eqnarray}
	where $\bm \xi$ is the pattern we want to retrieve, while the $\xi_i^\mu$'s are i.i.d. standard Gaussians.
	Averaging out the fields recursively, we define
	\begin{align}
	\label{eqn:Z2_2}
	\mathcal Z_2(t, \bm r) \coloneqq& \mathbb E_3 \left [ \mathcal Z_3(t, \bm r)^{\theta_2} \right ]^{1/\theta_2} \\
	\label{eqn:Z1_2}
	\mathcal Z_1(t, \bm r) \coloneqq& \mathbb E_2 \left [ \mathcal Z_2(t, \bm r)^{\theta_1} \right ]^{1/\theta_1} \\
	\label{eqn:Z0_2}
	\mathcal Z_0(t, \bm r) \coloneqq&  \exp \mathbb E_1 \left[ \log \mathcal Z_1(t, \bm r) \right ] \\
	\mathcal Z_N(t, \boldsymbol r) \coloneqq& \mathcal Z_0(t, \bm r) ,
	\end{align}
	where with $\mathbb E_a$ we mean the average over the variables $h_i^{(a)}$'s and $J_\mu^{(a)}$'s, for $a=1,2,3$, and with $\mathbb{E}_0$ we shall denote the average over the variables $\xi_i^{\mu}$'s.
\end{definition}
\begin{definition}
	The 2RSB interpolating pressure, at finite volume $N$, is introduced as
	\begin{equation}\label{AdiSK2RSB_HOP}
	\mathcal A_N (t, \bm r) \coloneqq \frac{1}{N}\mathbb E_0 \big[ \log \mathcal Z_0(t, \bm r) \big],
	\end{equation}
	and, in the thermodynamic limit, assuming its existence, is
	\begin{equation}
	\mathcal A (t, \bm r) \coloneqq \lim_{N \to \infty} \mathcal A_N (t, \bm r).
	\end{equation}
	By setting $t=1$, $\bm r = \bm 0$, the interpolating pressure recovers the standard pressure (\ref{PressureDef}), that is, $A_N(\alpha, \beta) = \mathcal A_N (t =1, \bm r =0)$.
\end{definition}
%
\begin{remark}
	In analogy with previous cases, we introduce the statistical average
	\begin{equation}
	\omega \big(\mathcal O(\bm \sigma, \bm \tau)\big) = \frac{1}{\mathcal Z_3(t, \bm r)}\sum_{\bm \sigma} \int D \bm \tau \mathcal O(\bm \sigma, \bm \tau) e^{\beta \mathcal H(\bm \sigma,  \bm \tau; t, \bm r)}
	\end{equation}
	where $\mathcal H(\bm \sigma,  \bm \tau; t, \bm r)$ depends also on all the interpolating and random variables, i.e.
	\begin{align}
	\mathcal H(\bm \sigma,  \bm \tau; t, \bm r)=&\frac{t}{2N}\sum_{i,j=1}^{N,N}\xi_i \xi_j\sigma_i \sigma_j+\sqrt{\frac{t}{N}}\sum_{i, \mu=1}^{N,P} \xi_i^\mu \sigma_i \tau_\mu +\sum_{a=1}^3 \sqrt{x^{(a)}}\sum_{i=1}^N h_i^{(a)}	\sigma_i \nonumber \\ +& \sum_{a=1}^3 \sqrt{y^{(a)}}\sum_{\mu=1}^P J_\mu^{(a)}\tau_\mu +z\sum_{\mu=1}^P \frac{\tau_\mu^2}{2}+ w \sum_{i=1}^N \xi_i\sigma_i.
	\end{align}
\end{remark}
\begin{remark}
	In order to lighten the notation, hereafter we use the following 
	\begin{align}
	\langle m \rangle=& \mathbb E_0  \mathbb E_1  \mathbb E_2 \left\{ \mathcal W_2\mathbb E_3\left[\mathcal W_3 \frac{1}{N}\sum_{i=1}^N \omega(\xi_i \sigma_i) \right] \right\} \\
	\langle m^2 \rangle=&  \mathbb E_0  \mathbb E_1  \mathbb E_2 \left\{ \mathcal W_2\mathbb E_3\left[ \mathcal W_3 \frac{1}{N^2}\sum_{i,j=1}^{N,N} \omega(\xi_i \xi_j\sigma_i \sigma_j) \right] \right\} \\
	\langle p_{11} \rangle=& \mathbb E_0  \mathbb E_1  \mathbb E_2 \left\{ \mathcal W_2\mathbb E_3\left[ \mathcal W_3 \frac{1}{P}\sum_{\mu=1}^P \omega(\tau_\mu^2) \right] \right\}\\
	\langle p_{12} \rangle_1=& \mathbb E_0  \mathbb E_1   \left \{ \frac{1}{P}\sum_{\mu=1}^P \Big [ \mathbb E_2 \Big (\mathcal W_2\mathbb E_3 \left [\mathcal W_3 \omega(\tau_\mu)\right] \Big )\Big ]^2  \right \} \\
	\langle p_{12} \rangle_2=& \mathbb E_0  \mathbb E_1  \mathbb E_2 \left\{\mathcal W_2\Bigg[ \frac{1}{P}\sum_{\mu=1}^P \left(  \mathbb E_3 \left [\mathcal W_3 \omega(\tau_\mu)\right] \right)^2 \Bigg] \right\}\\
	\langle p_{12} \rangle_3=& \mathbb E_0  \mathbb E_1  \mathbb E_2 \left\{ \mathcal W_2\mathbb E_3\left[ \mathcal W_3 \frac{1}{P}\sum_{\mu=1}^P \omega(\tau_\mu)^2 \right] \right\}\\
	\label{eqn:q121}
	\langle q_{12} \rangle_1=&\mathbb E_0  \mathbb E_1  \left \{ \frac{1}{N} \sum_{i=1}^N \Bigg [ \mathbb E_2 \Bigg( \mathcal W_2\mathbb E_3 \left[\mathcal W_3 \omega(\sigma_i)\right] \Bigg) \Bigg]^2  \right \} \\
	\langle q_{12} \rangle_2=&\mathbb E_0  \mathbb E_1  \mathbb E_2 \left\{\mathcal W_2\left[\frac{1}{N} \sum_{i=1}^N \Bigg ( \mathbb E_3 \left[\mathcal W_3 \omega(\sigma_i)\right] \Bigg)^2 \right]\right\} 
			\end{align}
	\begin{align}
	\label{eqn:q122}
	\langle q_{12} \rangle_3=& \mathbb E_0  \mathbb E_1  \mathbb E_2 \left\{\mathcal W_2\mathbb E_3\left [ \mathcal W_3 \frac{1}{N}\sum_{i=1}^N \omega(\sigma_i)^2 \right] \right\}\\
	\langle p_{12}q_{12} \rangle_1 =& \mathbb E_0  \mathbb E_1  \left \{ \frac{1}{P}\sum_{\mu=1}^P \frac{1}{N}\sum_{i=1}^N \Big[  \mathbb E_2 \mathcal W_2\mathbb E_3 \left ( \mathcal W_3  \omega(\tau_\mu \sigma_i ) \right) \Big]^2 \right \} \\
	\langle p_{12}q_{12} \rangle_1 =& \mathbb E_0  \mathbb E_1 \mathbb E_2 \left\{ \mathcal W_2 \left [ \frac{1}{P}\sum_{\mu=1}^P \frac{1}{N}\sum_{i=1}^N \Big (  \mathbb E_3 \left [ \mathcal W_3  \omega(\tau_\mu \sigma_i ) \right] \Big)^2 \right] \right\} \\
	\langle p_{12} q_{12}\rangle_3=& \mathbb E_0  \mathbb E_1  \mathbb E_2 \left\{ \mathcal W_2\mathbb E_3  \bigg[\mathcal W_3 \frac{1}{P}\sum_{\mu=1}^P \frac{1}{N}\sum_{i=1}^N \omega(\tau_\mu \sigma_i)^2 \bigg]\right\},
	\end{align}
	where we define the weights
	\begin{align}
	\mathcal W_2=\frac{\mathcal Z_2^{\theta_1}}{\mathbb E_2 \left [\mathcal Z_2^{\theta_1}\right]} \\
	\mathcal W_3=\frac{\mathcal Z_3^{\theta_2}}{\mathbb E_3 \left [\mathcal Z_3^{\theta_2}\right]}.
	\end{align}
\end{remark}
The next step is building a transport equation for the interpolating quenched pressure, for which we preliminary need to evaluate the related partial derivatives, as given in the next
\begin{lemma} \label{lemma:4}
	The partial derivative of the interpolating quenched pressure with respect to a generic variable $\rho$ reads as
	%
	%
	%
	\begin{equation}
	\label{eqn:partialrA}
	\frac{\partial }{\partial \rho} \mathcal A_N(t, \bm r)=\frac{1}{N} \mathbb E_0  \mathbb E_1  \mathbb E_2 \left \{ \mathcal W_2 \mathbb E_3 \left[\mathcal W_3\omega \big( \partial_\rho \mathcal B(\bm \sigma,  \bm \tau; t, \bm r) \big)\right]\right \}.
	\end{equation}
	In particular,
	\begin{align}
	\label{eqn:partialtA}
	\frac{\partial }{\partial t} \mathcal A_N =&\frac{\beta}{2}\langle m^2 \rangle+ \frac{\alpha\beta}{2}\big ( \langle p_{11} \rangle -(1-\theta_2)\langle p_{12}q_{12} \rangle_3-(\theta_2-\theta_1)\langle p_{12}q_{12} \rangle_2-\theta_1\langle p_{12}q_{12} \rangle_1 \big) \\
	\frac{\partial }{\partial x^{(1)}} \mathcal A_N =& \frac{\beta}{2}\big ( 1 -(1-\theta_2)\langle q_{12} \rangle_3 -(\theta_2-\theta_1)\langle q_{12} \rangle_2-\theta_1\langle q_{12} \rangle_1 \big) \\
	\frac{\partial }{\partial x^{(2)}} \mathcal A_N =& \frac{\beta}{2}\big ( 1 -(1-\theta_2)\langle q_{12} \rangle_3 -(\theta_2-\theta_1)\langle q_{12} \rangle_2 \big) \\
	\frac{\partial }{\partial x^{(3)}} \mathcal A_N =& \frac{\beta}{2}\big ( 1 -(1-\theta_2)\langle q_{12} \rangle_3  \big) \\
	\frac{\partial }{\partial y^{(1)}} \mathcal A_N =& \frac{\alpha\beta}{2}\big (  \langle p_{11} \rangle -(1-\theta_2)\langle p_{12} \rangle_3-(\theta_2-\theta_1)\langle p_{12} \rangle_2-\theta_1\langle p_{12} \rangle_1 \big) \\
	\frac{\partial }{\partial y^{(2)}} \mathcal A_N =& \frac{\alpha\beta}{2}\big (  \langle p_{11} \rangle -(1-\theta_2)\langle p_{12} \rangle_3-(\theta_2-\theta_1)\langle p_{12} \rangle_2\big) \\
	\frac{\partial }{\partial y^{(3)}} \mathcal A_N =& \frac{\alpha\beta}{2}\big (  \langle p_{11} \rangle -(1-\theta_2)\langle p_{12} \rangle_3\big) \\
	\frac{\partial }{\partial z} \mathcal A_N =& \frac{\alpha\beta}{2} \langle p_{11} \rangle \\
	\frac{\partial }{\partial w}  \mathcal A_N =& \beta\langle m \rangle.
	\end{align}
\end{lemma}
\begin{proof}
	The proof is pretty lengthy and basically requires just standard calculations at this point, so we just show how to prove (\ref{eqn:partialrA}).
\begin{align}
	\label{eqn:partialrA1}
	\frac{\partial }{\partial \rho}  \mathcal A_N(t, \bm r)=& \frac{1}{N} \mathbb E_0 \mathbb E_1 \bigg [\partial_\rho \ln\mathcal Z_1 \bigg] \nonumber \\
	=&\frac{1}{N} \mathbb E_0 \mathbb E_1 \left\{\frac{1}{\theta_2}\frac{1}{\mathcal Z_1} \mathbb{E}_2\big[ \mathcal Z_2^{\theta_2} \big]^{1/{\theta_2}-1} \mathbb E_2 \big[\partial_\rho\mathcal Z_2^{\theta_2} \big] \right\} \nonumber \\
	=&\frac{1}{N} \mathbb E_0 \mathbb E_1 \mathbb E_2\left\{ \mathcal{W}_2 \frac{1}{\theta_2} \frac{1}{\mathcal Z_2} \left[\mathbb{E}_3 \mathcal{Z}_3^{\theta_2}\right]^{\frac{1}{\theta_2}-1} \mathbb{E}_3 \left(\theta_2 \mathcal{Z}_3^{\theta_2-1}\partial_\rho Z_3\right)\right\} \nonumber \\
	=& \frac{1}{N} \mathbb E_0 \mathbb E_1 \mathbb E_2\bigg[ \mathcal{W}_2   \mathbb{E}_3 \left(\mathcal{W}_3 \frac{1}{Z_3} \partial_\rho Z_3\right)\bigg].
	\end{align}
\end{proof}
\begin{proposition}\label{prop13}
	The streaming of the 2-RSB interpolating quenched pressure obeys, at finite volume $N$, a standard transport equation, that reads as
	\begin{align}
	\label{eqn:transportequation2}
	\frac{d\mathcal A}{dt}&=\partial_t \mathcal A+\dot x^{(1)}\partial_{x_1} \mathcal A +\dot x^{(2)}\partial_{x_2} \mathcal A +\dot x^{(3)}\partial_{x_3} +\mathcal A\dot y^{(1)} \partial_{y_1} \mathcal A +\dot y^{(2)} \partial_{y_2} \mathcal A +\dot y^{(3)} \partial_{y_3} \mathcal A +\dot z \partial_{z} \mathcal A+\dot w \partial_{w} \mathcal A \nonumber \\
	&= S(t, \bm r) + V_N(t, \bm r)
	\end{align}
	where
	\begin{align}
	\label{eqn:f2}
	S(t, \bm r)  \coloneqq & -\frac{\beta}{2}\bar m^2-\frac{\alpha\beta^2}{2} \bar p_3 (1-\bar q_3)-\frac{\alpha\beta^2}{2} \theta_2 (\bar p_3 \bar q_3 - \bar p_2 \bar q_2)-\frac{\alpha\beta^2}{2} \theta_1 (\bar p_2 \bar q_2 - \bar p_1 \bar q_1) \\
	\label{eqn:V2}
	V_N(t, \bm r) \coloneqq & \frac{\beta}{2}\langle (m - \mb)^2 \rangle -(1-\theta_2)\frac{\alpha\beta^2}{2}\langle \Delta p_{12}  \Delta q_{12}\rangle_3 -(\theta_2-\theta_1)\frac{\alpha\beta^2}{2}\langle \Delta p_{12} \Delta q_{12}\rangle_2 \\
	\nonumber
	&-\theta_1 \frac{\alpha\beta^2}{2}\langle \Delta p_{12} \Delta q_{12}\rangle_1
	\end{align}
\end{proposition}
The proof works along the same lines of that of Proposition \ref{prop:9}, hence we omit it.
\newline
The velocities ruling the motion can be written, as standard so far, in terms of $\bar m,\bar p_1, \bar p_2, \bar q_1, \bar q_2$ as
	\begin{align}
	\label{eqn:dotx12}
	\dot x^{(1)} =& -\alpha \bar p_1 \\
	\dot x^{(2)} =& -\alpha (\bar p_2 - \bar p_1) \\
	\dot x^{(3)} =& -\alpha (\bar p_3 - \bar p_2) \\
	\dot y^{(1)} =& -\bar q_1 \\
	\dot y^{(2)} =& - (\bar q_2 - \bar q_1) \\
	\dot y^{(3)} =& - (\bar q_3 - \bar q_2) \\
	\dot z =& -(1-\bar q_3) \\
	\label{eqn:dotw2}
	\dot w =& -\bar m
	\end{align}
\begin{remark} \label{r:rr}
	In the thermodynamic limit, in the 2RSB scenario considered, we have
	\begin{align}
	\label{eqn:thlimaveragem}
	\lim_{N\rightarrow \infty} \langle (m - \bar m)^2 \rangle =& 0\\
	\lim_{N\rightarrow \infty} \langle (q_{12}-\bar q_i)^2 \rangle_i=& 0; \: \: \: i=1,2,3\\
	\label{eqn:thlimaveragep}
	\lim_{N\rightarrow \infty} \langle( p_{12}-\bar p_i)^2\rangle_i =& 0; \: \: \: i=1,2,3
	\end{align}
	and the potential (\ref{eqn:V2}) can be written as
	\begin{align}
	\label{eqn:compactV}
	V_N(t, \bm r) = \frac{\beta}{2}\bigg \{ & \langle (m-\bar m)^2 \rangle - \alpha\beta (1-\theta_2)\langle (q_{12}-\bar q_2) (p_{12}-\bar p_2)\rangle_3 +\nonumber \\
	&- \alpha\beta (\theta_2-\theta_1)\langle (q_{12}-\bar q_2) (p_{12}-\bar p_2)\rangle_2 -\alpha \beta\theta \langle (q_{12}-\bar q_1) (p_{12}-\bar p_1)\rangle_1 \bigg\}
	\end{align}
such that
	\begin{equation} \label{eq:V0_HRSB}
	\lim_{N \to \infty} V_N(t, \bm r) = 0.
	\end{equation}
\end{remark}
By naturally extending the scenario depicted in the 1RSB approximation, killing the potential is equivalent to requiring three temporal scales for thermalization, a slow one a middle  one and a fast one, and self-averaging within each time scale; we refer to Section \ref{comment} for a deeper discussion on the physics behind this choice.
\newline
Taking advantage of Remark \ref{r:rr} we can prove the following
\begin{proposition} \label{propHRSB}
	The transport equation associated to the interpolating pressure function $\mathcal A_N(t, \boldsymbol r)$, in the thermodynamic limit and under the 2RSB assumption, reads as
	\begin{equation}
	\label{eqn:solutionzeroV2}
	\mathcal A_{\textrm{2RSB}}(t, \bm r)=\mathcal A_{\textrm{2RSB}}(0, \bm{r}-\bm{\dot r}t)+S(t, \dot r)t
	\end{equation}
	whose solution is given by
	\begin{align}
		\label{eqn:Afinal}
	&\mathcal A_{2RSB}(t, \bm r)= \log 2+\frac{\alpha}{2\theta_2}\log\left[1+\beta\theta_2 \frac{y^{(3)}+(\bar q_3-\bar q_2)t}{1-\beta[z+(1-\bar q_2)t+\theta_1 (y^{(2)}+(\bar q_2-\bar q_1)t)]}\right]\nonumber \\
	&\frac{1}{\theta_1}\mathbb E_1 \log  \mathbb E_2\left \{ \mathbb E_3\cosh^{\theta_2} \right [ h^{(1)} \sqrt{x^{(1)}+\alpha \bar p_1 t}+ h^{(2)}\sqrt{x^{(2)}+\alpha(\bar p_2 - \bar p_1)t} + h^{(3)}\sqrt{x^{(3)}+\alpha(\bar p_3 - \bar p_2)t}+ w+mt\left] \right \}^{\frac{\theta_1}{\theta_2}} + \nonumber  \\
	&+\frac{\alpha}{2\theta_1}\log \left[1+\beta\theta_1 \frac{y^{(2)}+(\bar q_2-\bar q_1)t}{1-\beta(z+(1-\bar q_2)t+\theta_1 (y^{(2)}+(\bar q_2-\bar q_1)t)+\theta_2 (y^{(3)}+(\bar q_3-\bar q_2)t))}\right] +\nonumber \\
	&-\frac{\alpha}{2}\log\big(1-\beta(z+(1-\bar q_2)t)\big)
	+\frac{\alpha\beta}{2}\frac{ y^{(1)}+\bar q_1 t}{1-\beta(z+(1-\bar q_2)t+\theta_1 (y^{(2)}+(\bar q_2-\bar q_1)t)+\theta_2 (y^{(3)}+(\bar q_3-\bar q_2)t))}+ \nonumber \\
	&-\frac{\beta}{2}\bar m^2t-\frac{\alpha\beta^2}{2} \bar p_3 (1-\bar q_3)t-\frac{\alpha\beta^2}{2} \theta_2 (\bar p_3 \bar q_3 - \bar p_2 \bar q_2)t-\frac{\alpha\beta^2}{2} \theta_1 (\bar p_2 \bar q_2 - \bar p_1 \bar q_1)t
	\end{align}
\end{proposition}
\begin{proof}
By putting (\ref{eqn:dotx12})-(\ref{eqn:dotw2}) into (\ref{eqn:transportequation2}) we find	
	\begin{equation}
	\label{eqn:solutionzeroV1}
	\mathcal A_N(t, \bm r)=\mathcal A_0(0, \bm r_0) -\frac{\beta}{2}\bar m^2t-\frac{\alpha\beta^2}{2} \bar p_3 (1-\bar q_3)t-\frac{\alpha\beta^2}{2} \theta_2 (\bar p_3 \bar q_3 - \bar p_2 \bar q_2)t-\frac{\alpha\beta^2}{2} \theta_1 (\bar p_2 \bar q_2 - \bar p_1 \bar q_1)t
	\end{equation}
	where $r_0$ can be obtained using the equation of motion
	\begin{equation}
	\label{eqn:linearmotion}
	\bm r = \bm r_0 + \dot{\bm r} t
	\end{equation}
	where the velocities are defined in (\ref{eqn:dotx12})-(\ref{eqn:dotw2}). So all we have to compute is $\mathcal A_0(0, \bm r_0)$, that can be easily done because at $t=0$ the two body interaction vanishes and the (\ref{eqn:Z3_2}) can be written as a factorized one-body calculation (reported in detail in Appendix \ref{1body2rsb}) resulting in
\begin{align}
	\label{eqn:A0fin}
	\mathcal A_0(0, \bm r_0)=&\log 2+\frac{1}{\theta_1}\int Dh^{(1)} \log  \int Dh^{(2)}\bigg[ \int Dh^{(3)}\cosh^{\theta_2} \beta \big\{ h^{(1)} \sqrt{x_0^{(1)}}+h^{(2)} \sqrt{x_0^{(2)}}+h^{(3)} \sqrt{x_0^{(3)}} + w_0\big\}\bigg]^{\frac{\theta_1}{\theta_2}} + \nonumber \\
	&+\frac{\alpha}{2\theta_2}\log\bigg(1+\beta \frac{\theta_2 y_0^{(3)}}{1-\beta(z_0+\theta_2 y_0^{(3)})}\bigg)+\frac{\alpha}{2\theta_1}\log\bigg(1+\beta \frac{\theta_1 y_0^{(2)}}{1-\beta(z_0+\theta_1 y_0^{(2)}+\theta_2 y_0^{(3)})}\bigg) +\nonumber \\
	&-\frac{\alpha}{2}\log(1-\beta z_0)
	+\frac{\alpha\beta}{2}\frac{ y_0^{(1)}}{1-\beta(z_0+\theta_1 y_0^{(2)}+\theta_2 y_0^{(3)})}
	\end{align}
Finally, we put together (\ref{eqn:A0fin}), (\ref{eqn:solutionzeroV1}), (\ref{eqn:linearmotion}) and (\ref{eqn:dotx1})-(\ref{eq:dotw}) to get the complete expression for the interpolating pressure in the 2RSB approximation.
\end{proof}
Summarizing, overall we obtain the following
\begin{theorem}
The 2RSB quenched pressure for Hopfield model, in the thermodynamic limit, reads as
\begin{align}
\notag
&	\mathcal A_{\textrm{2RSB}}(t, \bm r)= \log 2+\frac{1}{\theta_1}\int Dh^{(1)} \log  \int Dh^{(2)}\bigg[ \int Dh^{(3)}\cosh^{\theta_2}(g(\boldsymbol h, \bar m))
 \bigg]^{\frac{\theta_1}{\theta_2}} + \nonumber \\
	&+\frac{\alpha}{2\theta_2}\log\left[1+ \frac{\beta\theta_2 (\qb_3-\qb_2)}{1-\beta[(1-\qb_2)+\theta_2 (\qb_3-\qb_2)]}\right]+\frac{\alpha}{2\theta_1}\log \left[1+ \frac{\beta \theta_1 (\qb_2-\qb_1)}{1-\beta[(1-\qb_2)+\theta_1 (\qb_2-\qb_1)+\theta_2 (\qb_3-\qb_2)]}\right] +\nonumber \\
	&-\frac{\alpha}{2}\log[1-\beta (1-\qb_2)] 	+\frac{\alpha\beta}{2}\frac{ \qb_1}{1-\beta[(1-\qb_2)+\theta_1 (\qb_2-\qb_1)+\theta_2 (\qb_3-\qb_2)]} -\frac{\beta}{2}t \bar m^2-\frac{\alpha\beta^2}{2} t \bar p_3 (1-\bar q_3)- \notag \\
	&-\frac{\alpha\beta^2}{2}t \theta_2 (\bar p_3 \bar q_3 - \bar p_2 \bar q_2)-\frac{\alpha\beta^2}{2}t \theta_1 (\bar p_2 \bar q_2 - \bar p_1 \bar q_1)
	\end{align}
	where $\boldsymbol h = (h^{(1)}, h^{(2)}, h^{(3)})$ and $g(\boldsymbol h, \bar m)= \beta \bar m+  h^{(1)} \sqrt{\alpha \beta \bar p_1}+ h^{(2)}\sqrt{\alpha\beta(\bar p_2 - \bar p_1)} + h^{(3)}\sqrt{\alpha\beta(\bar p_3 - \bar p_2)}$.
\end{theorem}
\begin{proof}	
	By taking $\bm r= \boldsymbol 0$ and $t=\beta$ we find the Hopfield pressure in the 2RSB approximation.	
\end{proof}
\begin{corollary}
The self-consistency equation are
\begin{align}
\label{seq:a2RSB}
\qb_1 =& \mathbb{E}_1 \left\{ \frac{\mathbb{E}_2 \left[ \left[ \mathbb{E}_3 \cosh^{\theta_2}(g(\boldsymbol h,\mb))\right]^{\frac{\theta_1}{\theta_2}}\displaystyle{\frac{\mathbb{E}_3 \left( \cosh^{\theta_2}(g(\boldsymbol h,\mb))\tanh(g(\boldsymbol h,\mb)) \right)}{\mathbb{E}_3 \cosh^{\theta_2}(g(\boldsymbol h,\mb))}}\right]}{\mathbb{E}_2 \left[\mathbb{E}_3 \cosh^{\theta_2} (g(\boldsymbol h,\mb) ) \right]^{\frac{\theta_1}{\theta_2}}}\right\}^2 \\
\label{seq:b2RSB}
\qb_2 =& \mathbb{E}_1 \left\{ \frac{\mathbb{E}_2 \left[ \left[ \mathbb{E}_3 \cosh^{\theta_2}(g(\boldsymbol h,\mb))\right]^{\frac{\theta_1}{\theta_2}}\left[\displaystyle{\frac{\mathbb{E}_3 \left( \cosh^{\theta_2}(g(\boldsymbol h,\mb))\tanh(g(\boldsymbol h,\mb)) \right)}{\mathbb{E}_3 \cosh^{\theta_2}(g(\boldsymbol h,\mb))}}\right]^2\right]}{\mathbb{E}_2 \left[\mathbb{E}_3\cosh^{\theta_2} (g(\boldsymbol h,\mb) ) \right]^{\frac{\theta_1}{\theta_2}}}\right\} \\
\label{seq:c2RSB}
\qb_3 =& \mathbb{E}_1 \left\{ \frac{\mathbb{E}_2 \left[ \left[ \mathbb{E}_3 \cosh^{\theta_2}(g(\boldsymbol h,\mb))\right]^{\frac{\theta_1}{\theta_2}}\displaystyle{\frac{\mathbb{E}_3 \left( \cosh^{\theta_2}(g(\boldsymbol h,\mb))\tanh^2(g(\boldsymbol h,\mb)) \right)}{\mathbb{E}_3 \cosh^{\theta_2}(g(\boldsymbol h,\mb))}}\right]}{\mathbb{E}_2 \left[\mathbb{E}_3\cosh^{\theta_2} (g(\boldsymbol h,\mb) ) \right]^{\frac{\theta_1}{\theta_2}}}\right\} \\
\label{seq:m2RSB}
\mb =& \mathbb{E}_1 \left\{ \frac{\mathbb{E}_2 \left[ \left[ \mathbb{E}_3 \cosh^{\theta_2}(g(\boldsymbol h,\mb)\right]^{\frac{\theta_1}{\theta_2}}\displaystyle{\frac{\mathbb{E}_3 \left( \cosh^{\theta_2}(g(\boldsymbol h,\mb)) \tanh(g(\boldsymbol h,\mb)) \right)}{\mathbb{E}_3 \cosh^{\theta_2}(g(\boldsymbol h,\mb))}}\right]}{\mathbb{E}_2 \left[\mathbb{E}_3 \cosh^{\theta_2} (g(\boldsymbol h,\mb))  \right]^{\frac{\theta_1}{\theta_2}}}\right\} \\
\label{seq:p12RSB}
\bar{p}_1&=\frac{\beta \qb_1}{[ 1-\beta(1-\qb_2) + \theta_1 (\qb_2-\qb_1) + \theta_2 (\qb_3-\qb_2) ]^2} 
\end{align}
\begin{align}
\label{seq:p22RSB}
\bar{p}_2 &= \bar{p}_1 + \frac{\beta (\qb_2-\qb_1)}{1-\beta[(1-\qb_2) + \theta_1 (\qb_2-\qb_1) + \theta_2 (\qb_3-\qb_2)] [1-\beta((1-\qb_2)+ \theta_2 (\qb_3-\qb_2))]} \\
\label{seq:p32RSB}
\bar{p}_3 &= \bar{p}_2 + \frac{\beta (\qb_3-\qb_2)}{[1-\beta (1-\qb_2)][1-\beta((1-\qb_2)+ \theta_2 (\qb_3-\qb_2))]}
\end{align}
\label{cor:SC_HOP2RSB}
\end{corollary}
Since the proof is similar to  that pursued to achieve Corollary (\ref{cor:SC_SK}) we omit it.


\subsection{Broken Replica Interpolation: K-RSB solution} \label{ssec:HOP_KRSB}
Mirroring the work done to describe the Sherrington-Kirkpatrick model with a signal, in this section we address the Hopfield model at the general step $K$ of RSB, providing just the main passages.
\begin{definition} \label{def:Htheta_KRSB}
	In the K-th step of replica-symmetry breaking, the distribution of the two-replica overlaps $q$ and $p$, in the thermodynamic limit, displays $K+1$ delta-peaks at the equilibrium values, referred to as $\bar{q}_1,...\bar{q}_{K+1}$ and as $\bar{p}_1,...\bar{p}_{K+1}$, repsectively, and the concentration is ruled by $\theta_i \in [0,1], \forall i =1, ... K$, namely
	\begin{equation}
\lim_{N\rightarrow +\infty} P'_N(q)= \sum_{a=0}^K (\theta_{a+1}-\theta_{a}) \delta(q-\qb_{a+1}) \label{limforq2KRSB}
	\end{equation}
	\begin{equation}
\lim_{N\rightarrow +\infty} P''_N(p)= \sum_{a=0}^K (\theta_{a+1}-\theta_{a}) \delta(p-\bar{p}_{a+1}). \label{limforp2KRSB}
	\end{equation}
	with $\theta_0=0$ and $\theta_{K+1}=1$.
\newline
	The magnetization still self-averages at $ \bar{m}$ as in (\ref{limform}).
\end{definition}

\begin{definition} 
\label{def:HOP_KRSB}
Given the interpolating parameters $\bm r = (x^{(1)}, ..., x^{(K+1)}, y^{(1)}, ..., y^{(K+1)}, z, w)$, $t$ and the i.i.d. auxiliary fields $\{h_i^{(1)},  ..., h_i^{(K+1)}\}_{i=1,...,N}$, with $h_i^{(1,...,K+1)} \sim \mathcal N [0,1]$ for $i=1,...,N$, and $\{J_{\mu}^{(1)}, ..., J_{\mu}^{(K+1)}\}_{\mu=1,...,P}$, with $J_{\mu}^{(1,...,K+1)} \sim \mathcal N [0,1]$ for $\mu=1,...,P$, we can write the K-RSB interpolating partition function $\mathcal Z_N(t, \boldsymbol r)$ recursively, starting by
\begin{align}
	\mathcal Z_{K+1}(t, \bm r) &=\sum_{\bm \sigma} \int D \bm \tau \exp \left \{ \beta \left[ \frac{t}{2N}\sum_{i,j=1}^{N,N}\xi_i \xi_j\sigma_i \sigma_j+\sqrt{\frac{t}{N}}\sum_{i, \mu=1}^{N,P} \xi_i^\mu \sigma_i \tau_\mu + w \sum_{i=1}^N \xi_i\sigma_i \right. \right. \nonumber \\ &+  \left. \left.  \sum_{a=1}^{K+1} \sqrt{x^{(a)}}\sum_{i=1}^N h_i^{(a)} \sigma_i + \sum_{a=1}^{K+1} \sqrt{y^{(a)}}\sum_{\mu=1}^P J_\mu^{(a)}\tau_\mu  +z \sum_{\mu=1}^P \frac{\tau_\mu^2}{2}  \right] \right\}
\end{align}
and then averaging out the fields one per time.
With $\mathbb{E}_a$ we denote the average over the variables $h_i^{(a)}$'s and $J_\mu^{(a)}$, for $a =1, ..., K+1$ and with $\mathbb{E}_0$ we denote the average over the variables $\xi^\mu$'s. 
\end{definition}

\begin{proposition}
At finite volume $N$ and finite $K$, the streaming of the K-RSB interpolating quenched pressure fulfills a standard transport equation, that reads as
\begin{align}
\dt \mathcal{A}_N + \sum_{b=0}^K \dot{x}^{(b+1)} \frac{\partial}{\partial x^{(b+1)}} \mathcal{A}_N + \sum_{b=0}^K \dot{y}^{(b+1)} \frac{\partial}{\partial y^{(b+1)}} \mathcal{A}_N + \dot{w}\mb \dw \mathcal{A}_N + \dot{z} \partial_z \mathcal{A}_N = V_N(t, \boldsymbol r) + S(t, \boldsymbol r)
\label{eq:ANfiniteHOP_KRSB}
\end{align}
where
\begin{align}
V_N (t, \bm r) &:= \frac{\beta}{2} \langle (m-\mb)^2 \rangle - \frac{\beta^2 \alpha}{2} \sum_{a=1}^{K+1} (\theta_{a} - \theta_{a-1}) \langle \Delta p_{12} \Delta q_{12} \rangle_{a} \\
S (t, \bm r) &:= -\frac{\beta}{2} \mb^2 + \frac{\beta^2\alpha}{2}\sum_{a=1}^{K+1} (\theta_{a} - \theta_{a-1}) \bar{p}_a \qb_a - \frac{\beta^2\alpha}{2}\bar{p}_{K+1}
\end{align}
\end{proposition}

\begin{proof}
The proof is similar to those provided to prove Proposition \ref{prop:9} for the 1RSB case and to prove Proposition \ref{prop13} for the 2RSB case. We start from
\begin{align}
\dt \mathcal{A}_N= \frac{1}{2} \langle m^2 \rangle + \frac{\alpha}{2} \left( \langle p_{11} \rangle - \sum_{a=0}^K (\theta_{a+1}-\theta_a)\langle p_{12} q_{12} \rangle_{a+1} \right)
\end{align}
and then we use
\begin{align}
\langle \Delta p_{12} \Delta q_{12} \rangle_a &= \langle p_{12} q_{12} \rangle_a - \bar{p}_a\qb_a + \qb_a\langle p_{12} \rangle_a + \bar{p}_a\langle q_{12}\rangle_a,  \forall a=1, ..., K+1 \label{quad_qapaHOP}\\
\langle (m - \mb)^2 \rangle &= \langle m^2 \rangle + \mb-2\mb\langle m \rangle \\
\frac{\alpha}{2} \langle p_{11} \rangle &=\frac{\alpha}{2} \langle p_{11} \rangle\left( 1-\qb_{K+1} + \sum_{b=0}^{K} (\qb_{b+1} -\qb_{b})\right) \notag \\
\frac{\alpha}{2} \bar{p}_K &= \frac{\alpha}{2} \sum_{b=0}^K (\bar{p}_{b+1} -\bar{p}_{b}).
\end{align}
In this way, placing
\begin{align}
\dot{y}^{(b)} &= - (\qb_{b} - \qb_{b-1}), \ b=1, ..., K+1 \\
\dot{x}^{(b)} &= - \alpha(\bar{p}_{b} - \bar{p}_{b-1}), \ b=1, ..., K+1 \\
\dot{z}&= (1-\qb_{K+1}) \\
\dot{w} &= - \mb,
\end{align}
we get the derivatives w.r.t. each $x^{(b)}$, $y^{(b)}$ and w.r.t. $w$.
\end{proof}

\begin{proposition}
The transport equation associated to the interpolating pressure function defined in (\ref{def:SK_KRSB}), in the thermodynamic limit and in the K-RSB scenario, reads as
\begin{align}
\dt \mathcal{A}_N + \sum_{b=0}^K \dot{x}^{(b+1)} \frac{\partial}{\partial x^{(b+1)}} \mathcal{A}_N &+ \sum_{b=0}^K \dot{y}^{(b+1)} \frac{\partial}{\partial y^{(b+1)}} \mathcal{A}_N + \dot{w}\dw \mathcal{A}_N + \dot{z} \partial_z \mathcal{A}_N = \notag \\
&=-\frac{\beta}{2} \mb^2 + \frac{\beta^2\alpha}{2}\sum_{a=0}^K (\theta_{a+1} - \theta_{a}) \bar{p}_a \qb_a - \frac{\beta^2\alpha}{2}\bar{p}_{K+1}
\label{eq:AN_KRSB_HOP}
\end{align}
whose solution is given by
\begin{align}
\mathcal{A}_N (t, \bm r) = \mathcal{A}_N(0, \bm r - \boldsymbol{ \dot{r}}t) +  t\left[\frac{\beta J_0}{2}\mb^2-\frac{\beta^2}{4} \left( \sum_{a=0}^K (\theta_{a+1} - \theta_a) \qb_{a+1}^2\right) + \frac{\beta^2}{2}(1-\qb_{K+1}) \right],
\end{align}
where
\begin{align}
\mathcal{A}_N(0, \bm r - \boldsymbol{ \dot{r}}t) = \frac{1}{\theta_1}\int D h^{(1)} \log \mathcal{N}_1 +& \frac{\alpha}{2} \sum_{z=1}^K \frac{1}{\theta_z} \log \left( 1+\beta \theta_z \frac{y_0^{(z+1)}}{1-\beta(z_0 + \sum_{b=1}^z \theta_b y_0^{(b+1)})} \right) - \notag \\
&- \frac{\alpha}{2}\log (1-\beta z_0) + \frac{\alpha}{2}\beta \frac{y_0^{(1)}}{1-\beta(z_0 + \sum_{b=1}^K \theta_b y_0^{(b+1)})}
\end{align}
with
\begin{align}
\mathcal{N}_a = \begin{cases}
\displaystyle{\int D h^{(a+1)} \left[\mathcal{N}_{a+1}\right]^{\theta_a/\theta_{a+1}} }\ \textnormal{for} \ a =1... K \\
2 \cosh \left( \beta (w +  \sum_{a=1}^{K+1} \sqrt{x_0^{(a)}} h^{(a)}) \right) \ \textnormal{for} \ a = K+1 \\
\end{cases}
\end{align}
\end{proposition}

\begin{theorem}
The K-RSB quenched pressure for the Hopfield model in the thermodynamic limit, reads as
\begin{align}
\mathcal{A}_{KRSB} = &\frac{1}{\theta_1}\int D h^{(1)} \log \mathcal{N}_1 - \frac{\alpha}{2}\log (1-\beta z_0) + \frac{\alpha}{2}\beta \frac{y_0^{(1)}}{1-\beta(z_0 + \sum_{b=1}^K \theta_b y_0^{(b+1)})} \notag \\
&-\frac{\beta}{2} \mb^2 + \frac{\beta^2\alpha}{2}\sum_{a=1}^{K+1} (\theta_{a} - \theta_{a-1}) \bar{p}_a \qb_a - \frac{\beta^2\alpha}{2}\bar{p}_{K+1}
\label{eq:Afin_KRSB_HOP}
\end{align}
with
\begin{align}
\mathcal{N}_a = \begin{cases}
\displaystyle{\int D h^{(a+1)}\left[ \mathcal{N}_{a+1}\right]^{\theta_a/\theta_{a+1}} }\ \textnormal{for} \ a =1, ... K \\
2^N \cosh^N \left( \beta (\mb +  \sum_{a=1}^{K+1} \sqrt{\alpha(\bar{p}_{a}-\bar{p}_{a-1}}  h^{(a)}) \right) \ \textnormal{for} \ a = K+1 \\
\end{cases}
\label{eq:def_N_HOP}
\end{align}
\end{theorem}
\begin{proof}
If we put $t=0$ and $\bm x= \bm y = w=z=0$ we obtain the K-RSB quenched pressure.
\end{proof}
\begin{corollary}
\label{cor:selfKRSB_HOP}
The self-consistence equations are 
\begin{align}
\qb_1 &= \mathbb{E}_1 \left\{ \frac{1}{\mathcal{N}_1}\mathbb{E}_2 \left[ \mathcal{N}_2^{\frac{\theta_1}{\theta_2}-1} \cdots \mathbb{E}_{K+1}\left[ \cosh^{\theta_k} (g(\boldsymbol h, \bar m)) \tanh (g(\boldsymbol h, \bar m) ) \right] \right]\right\}^2 \\
\qb_2 &=  \mathbb{E}_1 \left\{ \frac{1}{\mathcal{N}_1}\left[\mathbb{E}_2 \left[ \mathcal{N}_2^{\frac{\theta_1}{\theta_2}-1} \cdots \mathbb{E}_{K+1}\left[ \cosh^{\theta_k} (g(\boldsymbol h, \bar m)) \tanh (g(\boldsymbol h, \bar m))  \right] \right]\right]^2\right\} \\
\nonumber
... \\
\qb_{K+1} &=  \mathbb{E}_1 \left\{\frac{1}{\mathcal{N}_1}\mathbb{E}_2 \left[ \mathcal{N}_2^{\frac{\theta_1}{\theta_2}-1} \cdots \mathbb{E}_{K+1}\left[ \cosh^{\theta_k} (g(\boldsymbol h, \bar m)) \tanh^2 (g(\boldsymbol h, \bar m)  )\right] \right]\right\} \\
\bar{p}_1 &= \frac{\beta \qb_1}{\mathcal{Q}_1^2} \\
\bar{p}_h &= \bar{p}_{h-1} + \frac{\beta (\qb_h - \qb_{h-1})}{\mathcal{Q}_h \mathcal{Q}_{h-1}}, \ \ \ \forall h=2, ..., K+1 \\
\bar{m} &= \mathbb{E}_1 \left\{ \frac{1}{\mathcal{N}_1}\mathbb{E}_2 \left[ \mathcal{N}_2^{\frac{\theta_1}{\theta_2}-1} \cdots \mathbb{E}_{k+1}\left[ \cosh^{\theta_k} (g(\boldsymbol h, \bar m)) \tanh( g(\boldsymbol h, \bar m) ) \right] \right]\right\}
\end{align}
where $\boldsymbol h = (h^{(1)}, \cdots, h^{(K+1)})$, $g(\boldsymbol h, \bar m) = \beta \mb + \beta \sum_{a=1}^{K+1} \sqrt{\alpha(\bar{p}_{a}-\bar{p}_{a-1})}h^{(a)}$, $\mathcal{N}_1, ... \mathcal{N}_{K+1}$ are defined in (\ref{eq:def_N_HOP}), and
\begin{align}
\begin{cases}
\mathcal{Q}_{K+1} &= 1-\beta(1-\qb_{K+1}) \\
\mathcal{Q}_K &= \mathcal{Q}_{K+1} - \beta \theta_k (\qb_{K+1} - \qb_K)  \\
... \\
\mathcal{Q}_1 &= \mathcal{Q}_2 - \beta \theta_1 (\qb_2 - \qb_1)
\end{cases}.
\end{align}
\end{corollary}

\begin{proof}
The proof is similar to that provided for Corollary (\ref{cor:selfKRSB_SK}) and for Corollaries (\ref{cor:SC_HOP_1RSB}) and (\ref{cor:SC_HOP2RSB}).
In fact, for the self consistence equations w.r.t. $\mb$ and $\qb_j$, $j=1, ...K+1$ the proof is equal to Corollary (\ref{cor:selfKRSB_SK}).
\newline
On the other hand, we have
\begin{align}
\partial_{y^{(G)}} \mathcal{A}_N &= \frac{\alpha}{2} \beta \theta_{G-1} \sum_{j=1}^{G} B_j, \\
\partial_z \mathcal{A}_N &= \frac{\alpha \beta }{2} \sum_{j=1}^{K+2} B_j
\end{align}
where
\begin{align}
\displaystyle{\begin{cases}
B_1&= \frac{\beta y_0^{(1)}}{\mathcal{Q}_1} \\
B_2 &= \frac{1}{\theta_1 \mathcal{Q}_1} \\
B_{j+1} &= \left( \frac{1}{\theta_j} - \frac{1}{\theta_{j-1}} \right)\frac{1}{\mathcal{Q}_j}, \ \ \ j=2, ... , K \\
B_{K+2} &= -\frac{1}{\mathcal{Q}_{K+2}} \left( 1- \frac{1}{\theta_K}\right)
\end{cases}}
\end{align}

so,
\begin{align}
&\partial_{y^{(G+1)}} \mathcal{A}_N - \partial_{y^{(G)}} \mathcal{A}_N = \frac{\alpha \beta}{2 }(\theta_G - \theta_{G-1}) \bar{p}_G = \frac{\alpha \beta}{2 }(\theta_G - \theta_{G-1})\sum_{j=1}^{G} B_j + \frac{\alpha \beta}{2} B_{G+1}  , \ \ \ G= 1, ... K \\
&\partial_z \mathcal{A}_N - \partial_{y^{(K+1)}} \mathcal{A}_N = \frac{\alpha \beta}{2} (1-\theta_K) \bar{p}_{K+1} = \frac{\alpha \beta}{2} (1-\theta_K) \sum_{j=1}^{K+1} B_j + \frac{\alpha \beta}{2} B_{K+2}
\end{align}
and we reach the thesis.
\end{proof}

\section{Conclusions} \label{conclusions}

\subsection{A remark on the standard formulation: deepening the ansatz}\label{comment}
The Edward-Anderson spin-glass overlap $q_{\gamma \lambda}$ (and similarly its continuous counterpart $p_{\gamma \lambda}$ introduced to tackle the Hopfield model, see  Sec.~\ref{HopfieldSection}) measures how similar two different replicas are: the absolute value of $q_{\gamma \lambda}$ is 1 if all the spins/neurons belonging the replica $\gamma$ are parallel (or antiparallel) to those pertaining to the replica $\lambda$ and, in general, the larger its value and the larger the similarity between the configurations of the two replicas. When computing expectations of this quantity we can factorize the Boltzmann average as
\begin{equation}
\omega(\sigma_i^{(\gamma)} \sigma_j^{(\lambda)}) =\omega(\sigma_i^{(\gamma)}) ~ \omega(\sigma_j^{(\lambda)}),
\end{equation}
while for the quenched average appropriate care is required: clearly, if we assume such an average to be replica-independent, we end up with a RS painting of the system (and, for the Hopfield model, we recover the Amit-Gutfreund-Sompolinsky representation \cite{Amit}).
Conversely, if we retain the replica indices, it could happen that each replica {\em sits} in a different pure state such that,  in the asymptotic limit $N \to \infty$,  beyond ergodicity breaking also replica symmetry breaking spontaneously appears.

In order to deepen this concept, let us consider the telescopic definition of the quenched averages provided in the RSB scheme we pursued and restrict to the 1RSB picture (just for the sake of simplicity as its generalization to several, but finite, steps of RSB is straightforward), see eqs.~(\ref{eq:unouno})-(\ref{eq:duedue}) and eqs.~(\ref{eq:unouno_a})-(\ref{eq:duedue_a}): the auxiliary fields acting on a spin are conceived to simulate the effect of the remaining spins and the existence of two classes of fields (with the related averages) mirrors the existence of two temporal scales for thermalization, that is, a {\em fast scale} 
and a {\em slow scale}. 

Denoting with $\mathbb E_{\textrm{fast}}$ and $\mathbb E_{\textrm{slow}}$ the average over the fields corresponding to, respectively, the fast and the slow time scale, we compose the global average as Coolen and Van Mourik did in their dynamical derivation of the Parisi scheme \cite{Ton1,Ton2} 
and write
\begin{equation}
\langle \cdot \rangle = \mathbb E ~ \mathbb E_{\textrm{slow}} ~ \mathbb E_{\textrm{fast}} ~ \omega (\cdot).
\end{equation}
Let us now evaluate the average of the overlap $q_{\lambda \gamma}$ in the possible resulting scenarios. 
If the two fields coincide, the two replicas are subjected to the same field 
%
\begin{equation}
\label{eqn:qa}
\langle q_{\gamma \lambda} \rangle_a = \mathbb E ~ \mathbb E_{\textrm{slow}} \mathbb E_{\textrm{fast}}\frac{1}{N}\sum_i\omega (\sigma_i^{(\gamma)})\omega (\sigma_i^{(\lambda)}) = \mathbb E ~ \mathbb E_{\textrm{slow}} ~\mathbb E_{\textrm{fast}}\frac{1}{N}\sum_i\omega (\sigma_i)^2.
\end{equation}
%
If the two fields are distinct, the two replicas may share the slow timescale but evolve differently on the fast timescale (the opposite is thermodynamically forbidden \cite{Ton1,Ton2}), therefore averages can be partially factorized as
\begin{equation}
\label{eqn:qb}
\langle q_{\gamma\lambda} \rangle_b = \mathbb E ~ \mathbb E_{\textrm{slow}}\frac{1}{N}\sum_i\mathbb E_{\textrm{fast}}\omega (\sigma_i^{(\gamma)})\mathbb E_{\textrm{fast}}\omega (\sigma_i^{(\lambda)}) = \mathbb E ~  \mathbb E_{\textrm{slow}}\frac{1}{N}\sum_i\bigg(\mathbb E_{\textrm{fast}}\omega (\sigma_i)\bigg)^2.
\end{equation}
The case where replicas evolve independently correspond to an egodic regime and we can factorize everything obtaining
\begin{equation}
\label{eqn:qc}
\langle q_{\gamma\lambda} \rangle_c =\mathbb E \frac{1}{N}\sum_i  \mathbb E_{\textrm{slow}}\mathbb E_{\textrm{fast}}\omega (\sigma_i^{(\gamma)}) \mathbb E_{\textrm{slow}}\mathbb E_{\textrm{fast}}\omega (\sigma_i^{(\lambda)}) = \mathbb E \frac{1}{N}\sum_i\bigg( \mathbb E_{\textrm{slow}}\mathbb E_{\textrm{fast}}\omega (\sigma_i)\bigg)^2.
\end{equation}
Notice that (\ref{eqn:qa}) returns (\ref{eqn:q122_a}), and (\ref{eqn:qb}) returns (\ref{eqn:q121_a}), as long as we pose 
%
\begin{align}
\mathbb E_{\textrm{slow}}  [\cdot]=& \mathbb E_1  , \\
\mathbb E_{\textrm{fast}}[\cdot]=& \mathbb E_2[ \mathcal{W}_2  ~ \cdot ].
\end{align}
Indeed, in the RS thermodynamic limit, $q_{\gamma\lambda}=\frac{1}{N}\sum_i \sigma^{(\gamma)}_i\sigma^{(\lambda)}_i$ self-averages to a unique value $q_0$, while in the 1RSB it self-averages to two different values $q_1,\ q_2$: one accounting for the case where the replicas behave the same on both the fast and the slow timescales, and the another to account for their different behavior on the fast scale, while keeping the same on the slow one. This is in agreement with Coolen's perspective  \cite{Coolen} on the block-decomposition of the Parisi matrix: if we assume a fraction $\theta$ of replica couples belonging to the second group and a fraction $(1-\theta)$ of replica couples beloning to the first group, then, in the large $N$ limit, the probability distribution $P(q)$ would read as $P(q)=\theta P_1(q-\bar q_1)+(1-\theta)P_2(q-\bar q_2)$ (where $P_{1,2}(0)\geq 1$, $\int P_{1,2}(q-\bar q_i)dq=1$, $P_1(q-\bar q_1)P_2(q-\bar q_2)\simeq 0$ $ \forall q$).
By this perspective, in the transport PDE approach we pursued here, setting the potential equal to zero in the 1RSB framework (see (\ref{potenzialeRS-SK}) and (\ref{potenziale-RS-Hopfield})) is consistent with the assumption that $P(q)$ is decomposed in a sum of delta-functions, centered on the mean values  ($\lim_{N\rightarrow \infty}P_i(q)=\delta(q - \bar{q}_i)$). 
The generalization to an arbitrary, but finite, number of steps K of RSB is straightforward (see (\ref{V}) and (\ref{eqn:compactV})).

However, when trying to face RSB in associative neural networks as a perturbation of AGS theory (the RS scenario of the Hopfield model) a possible  concern lies in considering the Mattis magnetization $m_{\mu}:=\boldsymbol{\sigma}\cdot\boldsymbol{\xi}^{\mu}$ as a good quantifier to measure the retrieved signal (i.e., a generic quenched pattern $\boldsymbol \xi^{\mu}$).
This is because this observable is intrinsically defined  within a single replica, but, as in an RSB scenario different replicas can be locked on different pure states, it is no longer guaranteed that the scalar product between the quenched pattern $\boldsymbol \xi^{\mu}$ and a configuration $\boldsymbol \sigma^{(\gamma)}$ for the replica labelled as $\gamma$, i.e. $m_{\mu}^{(\gamma)}$, equals the scalar product between the same quanched pattern $\boldsymbol \xi^{\mu}$ and a configuration $\boldsymbol \sigma^{(\lambda)}$ for the replica $\lambda$, i.e. $m_{\mu}^{(\lambda)}$. In other words, it is not obvious that $m_{\mu}^{(\gamma)} \equiv m_{\mu}^{(\lambda)}$, hence raising the question on the soundness of a signal quantifier defined within a single replica.
Alternatively, as Parisi ultrametricity can be (partially) summarized by considering three replicas $\alpha, \gamma, \lambda$ and forcing their relative overlaps $q_{\alpha,\lambda}$ and $q_{\lambda,\gamma}$ to fulfill \cite{MPV}
\begin{equation}
P(q_{\alpha,\lambda},q_{\lambda,\gamma})=\frac{1}{2}P(q_{\alpha,\lambda})\delta(q_{\alpha,\lambda}-q_{\lambda,\gamma})+\frac{1}{2}P(q_{\alpha,\lambda})P(q_{\lambda,\gamma}),
\end{equation}
if we now assume that the replica  $\lambda$ {\em sits} in a configuration that is actually a pattern (as expected for a neural network under suitable conditions), say $\xi^{\mu}$, then the above equation turns into a constraint for the Mattis magnetizations that reads as
\begin{equation}\label{tana}
P(m_{\mu}^{(\alpha)},m_{\mu}^{(\gamma)})=\frac{1}{2}P(m_{\mu}^{(\alpha)})\delta(m_{\mu}^{(\alpha)}-m_{\mu}^{(\gamma)})+\frac{1}{2}P(m_{\mu}^{(\alpha)})P(m_{\mu}^{(\gamma)})
\end{equation}
that disagrees with a self-averaging ansatz for the Mattis magnetization. Yet, as far as we could check via a standard replica trick calculation, the self-averaging ansatz for the Mattis magnetization is the solely reasonable as the fields affecting the signal turn out to be RS.
\newline
Interestingly, this argument would also contribute to explain why, in Monte Carlo simulations, where no ansatz on RS is made, the Mattis magnetization never saturates to one, see e.g. \cite{Crisanti}, in fact, in the r.h.s. in eq. (\ref{tana}) above, the first part accounts for the RS signal of AGS theory, while the latter seems to point to a zero magnetization by a symmetry argument and the factorization.
In the present paper we did not deepen how to generalize AGS theory to overcome this problem, as the work was dedicated to the development of proper mathematical approaches to work out RSB calculations for neural networks, and we plan to discuss this in a dedicated forthcoming paper.

\subsection{Outlooks and future developments}
Neural networks are nowadays playing a pivotal role in the social and scientific progress (see e.g. \cite{medico,finanza,fisica,biologia}), especially due to some advances in machine learning research overall termed {\em deep learning} \cite{DL1}.
As a natural consequences of these applied achievements, we are witnessing an intensive quest for mathematical techniques able to tackle the emerging properties of these networks and this is just the context for this work: our aim here is to supply well-grounded theoretical tools to frame the behavior of these {\em intelligent machines} into a solid mathematical theory. In particular, we consider associative neural networks performing pattern recognition (i.e., the celebrated Hopfield model), and, even more specifically, we focus on the complex phenomenon of RSB (expected to affect this model when pushed close to its maximal capabilities).

In this work we did not examine the physical implications of RSB, mainly because we do believe that -- in its formulation for the Hopfield model and the related variations on theme -- its role remains marginal if not ill-posed. In our opinion, a proper RSB theory for associative neural networks could hardly be afforded as a {\em perturbation} of the AGS picture that perfectly describes the Hopfield model properties under the RS assumption. Therefore, while continuing the investigation of a reformulation of this problem to be presented in forthcoming papers, in the present one we accepted the standard approach and we focused purely on developing rigorous mathematical techniques, alternative to the widely-known replica trick.
\newline
In particular, we have shown that it is possible to {\em graft} the broken replica techniques developed by Guerra in \cite{Guerra} within the transport PDE approach developed by some of the present authors in \cite{AABF-NN2020} and we have presented the solution for the broken-replica quenched free-energy of the model, up to the K-th step of RSB.
These solutions were partially known in the literature (see \cite{Crisanti} for the 1RSB and \cite{Kuhn} for the 2-RSB) from non-rigorous tools; remarkably our approach perfectly reproduces these heuristic hints conferring them a mathematical rigour.

\appendix

\section{Proof of Lemma \ref{lemma:2}} \label{app1}

We prove explicitely only the first-order derivative with respect to $t$, that is eq.~(\ref{eq:2app}); for the others, the computation is analogous.
\begin{align}
&\dt A_N= \frac{1}{N} \mathbb{E}_0 \log \mathcal Z_0 = \frac{1}{N} \mathbb{E}_0 \mathbb{E}_1 \left( \frac{1}{\mathcal Z_1}\frac{1}{\theta} \mathcal Z_2^\theta \frac{1}{\mathcal Z_2} \dt \mathcal Z_2 \right) = \notag \\
&=\frac{1}{N} \mathbb{E}_0  \mathbb{E}_1 \mathbb{E}_2 \left\{ \mathcal{W}_2 \frac{1}{\mathcal Z_2}\sums \exp \left[ \beta \left (\frac{\sqrt{t}}{2}\frac{J}{\sqrt{N}} \sumij \si \sj z_{ij} + \suma \sqrt{x^{(a)}}\sumi h^{(a)}_i  \si  + \right. \right. \right. \nonumber \\
&\left. \left. \left. + t\frac{J_0}{2}Nm^2(\boldsymbol{\sigma}) +w \frac{J_0}{2} N m(\boldsymbol{\sigma}) \right) \right] \bigg (\frac{\beta J}{2\sqrt{t N}} \sumij \si \sj z_{ij} + \frac{\beta J_0}{2}N m^2 \bigg)  \right\} = \notag \\
&= \frac{\beta J_0}{2} \langle m^2 \rangle + \frac{\beta}{2N\sqrt{tN}} \mathbb{E}_0 \mathbb{E}_1 \mathbb{E}_2  \left\{ \mathcal{W}_2 \frac{1}{\mathcal Z_2}\sums \exp\left [ \beta \left (\frac{\sqrt{t}}{2}\frac{J}{\sqrt{N}} \sumij \si \sj z_{ij} + \right. \right.\right.\notag \\
&\left. \left. \left. +\suma \sqrt{x^{(a)}} \sumi h^{(a)}_i \si + t\frac{J_0}{2}Nm^2(\boldsymbol{\sigma})+w \frac{J_0}{2} N m(\boldsymbol{\sigma}) \right )\right] \sumij \si \sj z_{ij}  \right\} = \notag \\
&= \frac{\beta J_0}{2} \langle m^2 \rangle + \frac{\beta}{2N\sqrt{tN}} \mathbb{E}_0  \mathbb{E}_1  \mathbb{E}_2 \left\{\sumij  \partial_{z_{ij}} \left[ \mathcal{W}_2 \frac{1}{\mathcal Z_2}\sums \exp \left (\beta  (\frac{\sqrt{t}}{2}\frac{J}{\sqrt{N}} \sumij \si \sj z_{ij} + \right. \right. \right. \notag \\
&\left. \left. \left. +\suma \sqrt{x^{(a)}} \sumi h^{(a)}_i \si + t\frac{J_0}{2}Nm^2(\boldsymbol{\sigma})+w \frac{J_0}{2} N m(\boldsymbol{\sigma} ) ) \right) \right] \si \sj \right\},
\end{align}
where in the first passage we used the definition of $\mathcal Z_2$, in the second passage we highlighted the average magnetization and in the third passage we used Wick's theorem. \\
Now, we focus on the the derivative with respect to $z_{ij}$ which is computed apart and gives rise to three contributes denoted as $A, B, C$, that is
\footnotesize
\begin{align}
\notag
&\partial_{z_{ij}} \left \{ \mathcal{W}_2 \frac{1}{\mathcal Z_2}\sums \exp \left [\beta \left(\frac{\sqrt{t}}{2}\frac{J}{\sqrt{N}} \si \sj z_{ij} +\suma \sqrt{x^{(a)}} \sumi h^{(a)}_i \si + t\frac{J_0}{2}Nm^2(\boldsymbol{\sigma})+w \frac{J_0}{2} N m(\boldsymbol{\sigma}) \right) \right] \si \sj \right\} = A+B+C,
\end{align}
\normalsize
where
\footnotesize
\begin{align*}
A=& (\partial_{z_{ij}} \mathcal{W}_2 )\frac{1}{\mathcal Z_2} \sums \exp \bigg \{\beta (\frac{\sqrt{t}}{2}\frac{J}{\sqrt{N}} \sumij \si \sj z_{ij} +\suma \sqrt{x^{(a)}} \sumi h^{(a)}_i \si + t\frac{J_0}{2}Nm^2(\boldsymbol{\sigma})+w \frac{J_0}{2} N m(\boldsymbol{\sigma})) \bigg\} \si \sj = \notag\\
=& \frac{\beta \sqrt{t} J}{\sqrt{N}} \bigg\{ \theta \mathcal{W}_2 \frac{1}{	\mathcal Z_2} \sums \exp \bigg (\beta (\frac{\sqrt{t}}{2}\frac{J}{\sqrt{N}} \si \sj z_{ij} +\suma \sqrt{x^{(a)}} \sumi h^{(a)}_i \si+ t\frac{J_0}{2}Nm^2(\boldsymbol{\sigma})+w \frac{J_0}{2} N m(\boldsymbol{\sigma})) \bigg) \si \sj -\notag \\
& - \theta \mathcal{W}_2 \mathbb{E}_2 \left[ \mathcal{W}_2 \frac{1}{\mathcal Z_2} \sums \exp \bigg (\beta (\frac{\sqrt{t}}{2}\frac{J}{\sqrt{N}} \si \sj z_{ij} +  \suma \sqrt{x^{(a)}} \sumi h^{(a)}_i \si + t\frac{J_0}{2}Nm^2(\boldsymbol{\sigma})+w \frac{J_0}{2} N m(\boldsymbol{\sigma})) \bigg) \si \sj \right] \bigg \} \cdot \notag \\
&\cdot\frac{1}{\mathcal Z_2} \sums \exp \bigg (\beta(\frac{\sqrt{t}}{2}\frac{J}{\sqrt{N}} \sumij \si \sj z_{ij} +\suma \sqrt{x^{(a)}} \sumi h^{(a)}_i \si + t\frac{J_0}{2}Nm^2(\boldsymbol{\sigma})+w \frac{J_0}{2} N m(\boldsymbol{\sigma})) \bigg) \si \sj  = \notag \\
=&\frac{\beta \sqrt{t} J}{\sqrt{N}} \bigg\{  \theta \mathcal{W}_2 \left[\frac{1}{\mathcal Z_2} \sums \exp \bigg (\beta (\frac{\sqrt{t}}{2}\frac{J}{\sqrt{N}} \sumij \si \sj z_{ij} +\suma \sqrt{x^{(a)}} \sumi h^{(a)}_i \si +t\frac{J_0}{2}Nm^2(\boldsymbol{\sigma})+w \frac{J_0}{2} N m(\boldsymbol{\sigma})) \bigg) \si \sj \right]^2-\notag \\
& - \theta \mathcal{W}_2\frac{1}{\mathcal Z_2}\sums \exp \bigg (\beta (\frac{\sqrt{t}}{2}\frac{J}{\sqrt{N}} \sumij \si \sj z_{ij} +\suma \sqrt{x^{(a)}}\sumi h^{(a)}_i \si + t\frac{J_0}{2}Nm^2(\boldsymbol{\sigma})+w \frac{J_0}{2} N m(\boldsymbol{\sigma})) \bigg) \si \sj \cdot 
\end{align*}
\begin{align}
&\cdot\mathbb{E}_2 \left[\tilde{W_2}\frac{1}{\mathcal Z_2} \sums \exp \bigg (\beta (\frac{\sqrt{t}}{2}\frac{J}{\sqrt{N}} \sumij \si \sj z_{ij} +\suma \sqrt{x^{(a)}} \sumi h^{(a)}_i \si + t\frac{J_0}{2}Nm^2(\boldsymbol{\sigma})+w \frac{J_0}{2} N m(\boldsymbol{\sigma})) \bigg) \si \sj \right] \bigg\}
\end{align}
\normalsize
and, similarly,
\footnotesize
\begin{align}
B=& \frac{1}{\mathcal Z_2^2} \mathcal{W}_2\sums \exp \bigg (\beta (\frac{\sqrt{t}}{2}\frac{J}{\sqrt{N}} \sumij \si \sj z_{ij} +\suma \sqrt{x^{(a)}}\sumi h^{(a)} _i \si + t\frac{J_0}{2}Nm^2(\boldsymbol{\sigma})+w \frac{J_0}{2} N m(\boldsymbol{\sigma})) \bigg) \si \sj \partial_{z_{ij}} \mathcal Z_2 = \notag \\
=& - \frac{\beta \sqrt{t} J}{\sqrt{N}} \bigg[ \frac{1}{\mathcal Z_2} \sums \exp \bigg (\beta (\frac{\sqrt{t}}{2}\frac{J}{\sqrt{N}} \sumij \si \sj z_{ij} +\suma \sqrt{x^{(a)}} \sumi h^{(a)}_i \si + t\frac{J_0}{2}Nm^2(\boldsymbol{\sigma})+w \frac{J_0}{2} N m(\boldsymbol{\sigma})) \bigg) \si \sj \bigg]^2 
\end{align}
\begin{align}
C=& \mathcal{W}_2\frac{1}{\mathcal Z_2}\partial_{z_{ij}}\bigg[\sums \exp \bigg (\beta (\frac{\sqrt{t}}{2}\frac{J}{\sqrt{N}} \sumij \si \sj z_{ij} +\suma \sqrt{x^{(a)}} \sumi h^{(a)}_i \si +  t\frac{J_0}{2}Nm^2(\boldsymbol{\sigma})+w \frac{J_0}{2} N m(\boldsymbol{\sigma})) \bigg) \si \sj \bigg]= \notag \\
=& \frac{\beta \sqrt{t} J}{\sqrt{N}} \mathcal{W}_2 \frac{1}{\mathcal Z_2} \sums \exp \bigg (\beta (\frac{\sqrt{t}}{2}\frac{J}{\sqrt{N}} \sumij \si \sj z_{ij} +\suma \sqrt{x^{(a)}}\sumi h^{(a)}_i \si +  t\frac{J_0}{2}Nm^2(\boldsymbol{\sigma})+w \frac{J_0}{2} N m(\boldsymbol{\sigma})) \bigg) (\si \sj )^2
\end{align}
\normalsize
To sum up, we have
\begin{align}
\dt \mathcal A_N = \frac{\beta J_0}{2} \langle m^2 \rangle + \frac{\beta^2 J}{4} [1-(1-\theta) \langle q_{12}^2 \rangle_2 - \theta \langle q_{12}^2 \rangle_1]. 
\end{align}


\section{Proof of Corollary \ref{cor:SC_SK}}  \label{app:SC_SK}
We recall the initial relations (\ref{selfx1})-(\ref{selfw}) for $\mb, \qb_1, \qb_2$ obtained from (\ref{eq:x1app})-(\ref{eq:wapp}) 
\begin{align}
\frac{\partial}{\partial x^{(2)}} A_{\textrm{1RSB}} - \frac{\partial}{\partial x^{(1)}} A_{\textrm{1RSB}} &= \frac{\beta^2}{2}\theta \qb_1 \label{selfq1_app}\\
\frac{\beta^2}{2}-\frac{\partial}{\partial x^{(2)}} A_{\textrm{1RSB}} &= \frac{\beta^2}{2}(1-\theta)\qb_2 \label{selfq2_app}\\
\mb &= \frac{2}{\beta J_0}\dw A_{\textrm{1RSB}}. \label{selfm_app}
\end{align}
Now, we evaluate the derivatives w.r.t. $x^{(1,2)}$ and $w$ of the function (\ref{eq:fina}) starting from (\ref{eq:fina}) and we will then plug the resulting expressions into (\ref{selfq1_app})-(\ref{selfm_app}).
Let us pose $g(\bm h, \mb)=\beta \suma \sqrt{x_0^{(a)}}h^{(a)} + \beta w_0 \frac{J_0}{2}$ and compute the derivative of (\ref{eq:fina}) w.r.t. $x^{(1)}$:
\begin{align}
\frac{\partial}{\partial x^{(1)}} \mathcal{A}_{1RSB} = \mathbb{E}_1 \left\{ \frac{1}{\theta} \frac{\mathbb{E}_2 \left[\theta \cosh^\theta (g(\bm h, \mb)) \tanh(g(\bm h, \mb)) \beta\displaystyle{\frac{1}{2\sqrt{x_0^{(1)}}}}h^{(1)}\right]}{\mathbb{E}_2 \left[\cosh^\theta (g(\bm h, \mb))\right]}\right\}.
\end{align}
Now we use (\ref{eqn:gaussianrelation2}) w.r.t. $h^{(1)}$:
\footnotesize
\begin{align}
&\frac{\partial}{\partial x^{(1)}} \mathcal{A}_{1RSB} =\frac{\beta}{2\sqrt{x_0^{(1)}}} \mathbb{E}_1 \left\{ \partial_{h^{(1)}} \left[ \frac{\mathbb{E}_2 \left[\cosh^\theta (g(\bm h, \mb)) \tanh(g(\bm h, \mb)) \right]}{\mathbb{E}_2 \left[\cosh^\theta (g(\bm h, \mb))\right]}\right]\right\} = \notag \\
&= \frac{\beta}{2\sqrt{x_0^{(1)}}} \mathbb{E}_1 \left\{ \beta \sqrt{x^{(1)}}\left[\frac{\mathbb{E}_2\left[\theta\cosh^\theta (g(\bm h, \mb)) \tanh^2 (g(\bm h, \mb))\right]}{\mathbb{E}_2 \left[\cosh^\theta (g(\bm h, \mb))\right]}\right]+\right. \notag \\
&+ \beta \sqrt{x_0^{(1)}}\left[\frac{\mathbb{E}_2\left[\cosh^\theta (g(\bm h, \mb))(1-\tanh^2(g(\bm h, \mb)))\right]}{\mathbb{E}_2 \left[\cosh^\theta (g(\bm h, \mb))\right]}\right] - \left. \beta \sqrt{x_0^{(1)}} \theta \left[\frac{\mathbb{E}_2\left[\cosh^\theta (g(\bm h, \mb))\tanh^2(g(\bm h, \mb))\right]}{\mathbb{E}_2 \left[\cosh^\theta (g(\bm h, \mb))\right]}\right]^2\right\}.
\label{ARSx1}
\end{align}
\normalsize
Rearranging (\ref{ARSx1}) we have

\begin{align}
\frac{\partial}{\partial x^{(1)}} \mathcal{A}_{1RSB} &= \frac{\beta^2}{2} \left\{  1 - (1-\theta) \mathbb{E}_1\left[ \frac{\mathbb{E}_2\left[\cosh^\theta (g(\bm h, \mb))\tanh^2(g(\bm h, \mb))\right]}{\mathbb{E}_2 \left[\cosh^\theta (g(\bm h, \mb))\right]} \right] - \right. \notag \\
&\left.+\theta \mathbb{E}_1\left[\frac{\mathbb{E}_2\left[\cosh^\theta (g(\bm h, \mb))\tanh^2(g(\bm h, \mb))\right]}{\mathbb{E}_2 \left[\cosh^\theta (g(\bm h, \mb))\right]}\right]^2\right\} \label{eq:ARSBx1fin}
\end{align}
In the same way we compute the derivative of (\ref{1rsbtranseq}) w.r.t. $x^{(2)}$:
\begin{align}
\frac{\partial}{\partial x^{(2)}} \mathcal{A}_{1RSB} = \mathbb{E}_1 \left\{ \frac{1}{\theta} \frac{\mathbb{E}_2 \left[\theta \cosh^\theta (g(\bm h, \mb)) \tanh(g(\bm h, \mb)) \beta\displaystyle{\frac{1}{2\sqrt{x_0^{(2)}}}}h^{(2)}\right]}{\mathbb{E}_2 \left[\cosh^\theta (g(\bm h, \mb))\right]}\right\}
\end{align}

Now we use (\ref{eqn:gaussianrelation2}) w.r.t. $h^{(2)}$:
\footnotesize
\begin{align}
&\frac{\partial}{\partial x^{(2)}} \mathcal{A}_{1RSB} = \frac{\beta}{2\sqrt{x_0^{(2)}}} \mathbb{E}_1 \left\{  \frac{\mathbb{E}_2 \left[ \partial_{h^{(2)}}(\cosh^\theta (g(\bm h, \mb)) \tanh(g(\bm h, \mb)) )\right]}{\mathbb{E}_2 \left[\cosh^\theta (g(\bm h, \mb))\right]}\right\} = \notag \\
&= \frac{\beta}{2\sqrt{x_0^{(2)}}} \mathbb{E}_1 \left\{ \beta \sqrt{x_0^{(2)}}\frac{\mathbb{E}_2\left[\theta\cosh^\theta (g(\bm h, \mb)) \tanh^2 (g(\bm h, \mb))\right]}{\mathbb{E}_2 \left[\cosh^\theta (g(\bm h, \mb))\right]} +  \beta \sqrt{x_0^{(2)}}\frac{\mathbb{E}_2\left[\cosh^\theta (g(\bm h, \mb))(1-\tanh^2(g(\bm h, \mb)))\right]}{\mathbb{E}_2 \left[\cosh^\theta (g(\bm h, \mb))\right]} \right\}. \label{ARSx2}
\end{align}
\normalsize
Rearranging (\ref{ARSx2}) we have
\begin{align}
\frac{\partial}{\partial x^{(2)}} \mathcal{A}_{1RSB} &= \frac{\beta^2}{2} - \frac{\beta^2}{2} (1-\theta) \mathbb{E}_1 \left\{ \frac{\mathbb{E}_2\left[\cosh^\theta (g(\bm h, \mb))\tanh^2(g(\bm h, \mb))\right]}{\mathbb{E}_2 \left[\cosh^\theta (g(\bm h, \mb))\right]}\right\}
\label{eq:ARSBx2fin}
\end{align}

In the end, we compute the derivative of (\ref{1rsbtranseq}) w.r.t. $w$ :
\begin{align}
\dw \mathcal{A}_{1RSB} = \mathbb{E}_1 \left\{ \frac{1}{\theta}\frac{\mathbb{E}_2\left[\theta \cosh(g(\bm h, \mb))\tanh(g(\bm h, \mb))\displaystyle{\beta\frac{J_0}{2}}\right]}{\mathbb{E}_2 \left[\cosh^\theta (g(\bm h, \mb))\right]}\right\}
\label{ARSBwfin}
\end{align}

Now we use (\ref{charw}), (\ref{charx1}) and (\ref{charx2}) and we pose $t=1$, $x^{(1)}$, $x^{(2)}$, $w=0$; in this way we have $g(\bm h, \mb)= \beta J \sqrt{\qb_1} h^{(1)} + \beta J \sqrt{\qb_2-\qb_1} h^{(2)} + \beta \mb J_0$.

Finally, we use (\ref{ARSBwfin}), (\ref{eq:ARSBx1fin}) and (\ref{eq:ARSBx2fin}) to get 
\small
\begin{align}
&\frac{\partial}{\partial x^{(2)}}\mathcal{A}_{1RSB} - \frac{\partial}{\partial x^{(1)}}\mathcal{A}_{1RSB} = \frac{\beta^2}{2}\theta \mathbb{E}_1\left\{\frac{\mathbb{E}_2\left[\cosh^\theta (g(\bm h, \mb))\tanh^2(g(\bm h, \mb))\right]}{\mathbb{E}_2 \left[\cosh^\theta (g(\bm h, \mb))\right]}\right\}^2 = \frac{\beta^2}{2}\theta \qb_1 \\
&\frac{\partial}{\partial x^{(2)}}\mathcal{A}_{1RSB} = \frac{\beta^2}{2} - \frac{\beta^2}{2}(1-\theta) \mathbb{E}_1\left\{\frac{\mathbb{E}_2\left[\cosh^\theta (g(\bm h, \mb))\tanh^2(g(\bm h, \mb))\right]}{\mathbb{E}_2 \left[\cosh^\theta (g(\bm h, \mb))\right]}\right\}= \frac{\beta^2}{2} - \frac{\beta^2}{2}(1-\theta) \qb_2  \\
&\dw \mathcal{A}_{1RSB} = \beta\frac{J_0}{2} \mathbb{E}_1 \left\{ \frac{1}{\theta}\frac{\mathbb{E}_2\left[\theta \cosh(g(\bm h, \mb))\tanh(g(\bm h, \mb))\right]}{\mathbb{E}_2 \left[\cosh^\theta (g(\bm h, \mb))\right]}\right\} = \beta\frac{J_0}{2} \mb.
\end{align}
\normalsize

Rearranging these equations we obtain the above self-consistencies.


\section{Proof of Lemma \ref{lemma:4}} \label{app2}
We will prove only (\ref{eqn:partialtA}), being the proofs for the others obtained in a similar way. First of all, using (\ref{eqn:partialrA}) we see that
\begin{align}
\label{eqn:partialtAproof1}
\partial_t \mathcal{A}_N =\frac{1}{2}\langle m^2 \rangle+\frac{1}{2N\sqrt{Nt}}\mathbb{E}_0\mathbb{E}_1\mathbb{E}_2 \left[ \mathcal{W}_2\sum_{i,\mu}\xi_i^\mu\omega( \sigma_i\tau_\mu) \right]
\end{align}
Now, using Wick's theorem (\ref{eqn:gaussianrelation2}), we may rewrite the second member of (\ref{eqn:partialtAproof1}) as
\begin{align}
\label{eqn:partialtAproof2}
\frac{1}{2N\sqrt{Nt}}\sum_{i,\mu}\mathbb{E}_0\mathbb{E}_1 \mathbb{E}_2 \left[\partial_{\xi_i^\mu}\bigg(\mathcal{W}_2\omega( \sigma_i\tau_\mu) \bigg)\right] =D_1+D_2+D_3
\end{align}
Let's investigate those three terms:
\footnotesize
\begin{align}
\label{eqn:D1}
D_1 =&\frac{1}{2N\sqrt{Nt}}\sum_{i,\mu}\mathbb{E}_0 \mathbb{E}_1  \mathbb{E}_2 \bigg[\mathcal{W}_2\partial_{\xi_i^\mu}\omega( \sigma_i\tau_\mu) \bigg] = \frac{1}{2N^2}\sum_{i,\mu}\mathbb{E}_0 \mathbb{E}_1  \mathbb{E}_2 \left[\mathcal{W}_2\omega( \sigma_i^2\tau_\mu^2) \right]  -\frac{1}{2N^2}\sum_{i,\mu}\mathbb{E}_0 \mathbb{E}_1 \mathbb{E}_2 \left[\mathcal{W}_2\omega( \sigma_i\tau_\mu)^2 \right]=\nonumber \\
=&\frac{\alpha}{2} \langle p_{11} \rangle-\frac{\alpha}{2}\langle p_{12}q_{12} \rangle_2 
\end{align}
\begin{align}
D_2=&\frac{1}{2N\sqrt{Nt}}\sum_{i,\mu}\mathbb{E}_0 \mathbb{E}_1  \mathbb{E}_2 \left[\omega( \sigma_i\tau_\mu)\frac{\partial_{\xi_i^\mu}\mathcal Z_2^\theta}{\mathbb{E}_2 \left(\mathcal Z_2^\theta\right)} \right]= \frac{\theta}{2N^2}\sum_{i,\mu}\mathbb{E}_0 \left\{ \mathbb{E}_1 \left[ \mathbb{E}_2 \left(\mathcal{W}_2\omega( \sigma_i\tau_\mu)^2 \right) \right] \right\}=  \frac{\alpha}{2}\theta \langle p_{12}q_{12} \rangle_2 
\label{eqn:D2}
\end{align}

\begin{align}
D_3=&\frac{1}{2N\sqrt{Nt}}\sum_{i,\mu}\mathbb{E}_0 \mathbb{E}_1 \mathbb{E}_2 \left[\omega( \sigma_i\tau_\mu)\mathcal Z_2^\theta\partial_{\xi_i^\mu}\frac{1}{\mathbb{E}_2 \left(\mathcal Z_2^\theta\right)} \right]= -\frac{\theta}{2N^2}\sum_{i,\mu}\mathbb{E}_0 \mathbb{E}_1 \mathbb{E}_2 \left[\omega( \sigma_i\tau_\mu)\mathcal{W}_2\mathbb{E}_2 \left(\mathcal{W}_2\frac{\partial_{\xi_i^\mu}\mathcal Z_2}{\mathcal Z_2}\right)\right] = \nonumber \\
=&-\frac{\theta}{2N^2}\sum_{i,\mu}\mathbb{E}_0\mathbb{E}_1 \mathbb{E}_2 \left[\omega( \sigma_i\tau_\mu)\mathcal{W}_2\mathbb{E}_2 \left(\omega( \sigma_i\tau_\mu)\mathcal{W}_2\right) \right]=-\frac{\alpha}{2}\theta\langle p_{12}q_{12} \rangle_1
\label{eqn:D3}
\end{align}
\normalsize
Putting (\ref{eqn:D1}),  (\ref{eqn:D2}) and (\ref{eqn:D3}) inside (\ref{eqn:partialtAproof2}), and (\ref{eqn:partialtAproof2}) inside (\ref{eqn:partialtAproof1}) we find (\ref{eqn:partialtA}).

\section{One-body calculations: 1RSB}\label{1body1rsb}
In this appendix we report explicitly the calculations for the one-body problem for the 1RSB quenched pressure of the Hopfield model.
\footnotesize
\begin{align}
&\mathcal Z_2(0, \bm r_0)=\sum_{\bm \sigma} \exp\left({\sum_{a=1}^2 \sqrt{x_0^{(a)}}\sum_i h_i^{(a)}	\sigma_i + w_0 \sum_i \xi_i\sigma_i  }\right)\int D \bm \tau \exp\left({\sum_{a=1}^2 \sqrt{y_0^{(a)}}\sum_\mu J_\mu^{(a)}\tau_\mu +z_0\sum_\mu \frac{\tau_\mu^2}{2}}\right)= \nonumber \\
&=\prod_i\sum_{ \sigma_i=\pm 1} \exp\left({\sigma_i \big(\sqrt{x_0^{(1)}} h_i^{(1)} +\sqrt{x_0^{(2)}} h_i^{(2)} + w_0 \xi_i\big)}\right)\prod_\mu \int\frac{d\tau_\mu}{\sqrt{2\pi}} \exp\left({\tau_\mu \big (\sqrt{y_0^{(1)}}J_\mu^{(1)} +\sqrt{y_0^{(2)}}J_\mu^{(2)}\big)-\frac{\tau_\mu^2}{2}(1-z_0)}\right)= \nonumber \\
&=2^N\prod_i \cosh \left [ \sqrt{x_0^{(1)}} h_i^{(1)} +\sqrt{x_0^{(2)}} h_i^{(2)} + w_0 \xi_i \right] \frac{1}{(1-z_0)^\frac{P}{2}}\prod_\mu \exp \left [ \frac{\left (\sqrt{y_0^{(1)}}J_\mu^{(1)} +\sqrt{y_0^{(2)}}J_\mu^{(2)}\right)^2}{2(1-z_0)}\right]
\end{align}
\normalsize
where, to write the last equality we used (\ref{eqn:gaussianrelation2}). We may now proceed further to write
\begin{align}
\mathcal Z_1(0, \bm r_0)^\theta=&2^{\theta N}\prod_i \int Dh^{(2)} \cosh^\theta \left( \sqrt{x_0^{(1)}} h^{(1)} +\sqrt{x_0^{(2)}} h_i^{(2)} + w_0 \xi_i \right) \cdot \nonumber \\
&\cdot \frac{1}{(1-z_0)^{\theta \frac{ P}{2}}}\prod_\mu \int DJ^{(2)} \exp \left [ \theta\frac{\left (\sqrt{y_0^{(1)}}J_\mu^{(1)} +\sqrt{y_0^{(2)}}J^{(2)}\right)^2}{2(1-z_0)}\right ]= \nonumber \\
=&\frac{2^{\theta N}}{(1-z_0)^{\theta \frac{P}{2}}}\prod_i \int Dh^{(2)} \cosh^\theta \left ( \sqrt{x_0^{(1)}} h^{(1)} +\sqrt{x_0^{(2)}} h_i^{(2)} + w_0 \xi_i \right)\cdot \nonumber \\
&\cdot \prod_\mu \int \frac{dJ^{(2)}}{\sqrt{2\pi}} \exp \left [ \frac{\theta y_0^{(1)}(J_\mu^{(1)})^2 +2J^{(2)}\theta\sqrt{y_0^{(1)}y_0^{(2)}}J_\mu^{(1)}-(1-z_0-\theta y_0^{(2)})(J^{(2)})^2}{2(1-z_0)}\right]=\nonumber \\
=&2^{\theta N}\prod_i \int Dh^{(2)}\cosh^\theta \left ( \sqrt{x_0^{(1)}} h^{(1)} +\sqrt{x_0^{(2)}} h_i^{(2)} + w_0 \xi_i\right) \cdot \nonumber \\
&\cdot \frac{(1-z_0)^{ \frac{P}{2}(1-\theta)}}{(1-z_0-\theta y_0^{(2)})^\frac{P}{2}} \prod_\mu\exp \left [ (J_\mu^{(1)})^2\frac{ \theta y_0^{(1)}}{2(1-z_0)}\bigg(1+\frac{\theta y_0^{(2)}}{1-z_0-\theta y_0^{(2)}}\bigg)\right].
\end{align}
Then,
\begin{align}
\mathcal Z_1(0, \bm r_0)=&2^{N}\prod_i \left[\int Dh^{(2)}\cosh^\theta \left( \sqrt{x_0^{(1)}} h^{(1)} +\sqrt{x_0^{(2)}} h_i^{(2)} + w_0 \xi_i\right)\right]^\frac{1}{\theta}\cdot \nonumber \\
&\cdot \frac{(1-z_0)^{ \frac{P}{2\theta}(1-\theta)}}{(1-z_0-\theta y_0^{(2)})^\frac{P}{2\theta}} \prod_\mu\exp 
\left[ (J_\mu^{(1)})^2\frac{ y_0^{(1)}}{2(1-z_0-\theta y_0^{(2)})}\right]
\end{align}
By proceeding in the computation we find
\begin{align}
&\log \mathcal Z_0(0, \bm r_0)= N\log 2+\sum_i\frac{1}{\theta}\int Dh^{(1)} \log  \int Dh^{(2)}\cosh^\theta \left( \sqrt{x_0^{(1)}} h^{(1)} +\sqrt{x_0^{(2)}} h_i^{(2)} + w_0 \xi_i\right) + \nonumber \\
&+\frac{P}{2\theta}\log\bigg(1+\frac{\theta y_0^{(2)}}{1-z_0-\theta y_0^{(2)}}\bigg)
-\frac{P}{2}\log(1-z_0)+ \sum_\mu  \frac{ y_0^{(1)}}{2(1-z_0-\theta y_0^{(2)})}\int DJ^{(1)}( J^{(1)})^2= \nonumber \\
&= N\log 2+\sum_i\frac{1}{\theta}\int Dh^{(1)} \log  \int Dh^{(2)}\cosh^\theta \left( \sqrt{x_0^{(1)}} h^{(1)} +\sqrt{x_0^{(2)}} h_i^{(2)} + w_0 \xi_i \right) + \nonumber \\
&+\frac{P}{2\theta}\log\bigg(1+\frac{\theta y_0^{(2)}}{1-z_0-\theta y_0^{(2)}}\bigg)
-\frac{P}{2}\log(1-z_0)
+\frac{P}{2}\frac{ y_0^{(1)}}{(1-z_0-\theta y_0^{(2)})}
\end{align}
Finally, we can easily write $\mathcal A_0(0, \bm r_0)$ by noting that there is no dependence on $\xi_i^\mu$'s, so that there is no need to make the average, and by noting that, due to the parity of $\cosh(\cdot)$ and $Dh^{(a)}$, the argument of $\cosh(\cdot)$ is indipendent from the sign of $\xi_i$. We can therefore put $N$ instead of the sum over $i$ and write
\begin{align}
\label{eqn:A0finApp}
\mathcal A_0(0, \bm r_0)=&\log 2+\frac{1}{\theta}\int Dh^{(1)} \log  \int Dh^{(2)}\cosh^\theta \left( h^{(1)} \sqrt{x_0^{(1)}}+h^{(2)} \sqrt{x_0^{(2)}} + w_0\right) + \nonumber \\
&+\frac{\alpha}{2\theta}\log\bigg(1+\frac{\theta y_0^{(2)}}{1-z_0-\theta y_0^{(2)}}\bigg)
-\frac{\alpha}{2}\log(1-z_0)
+\frac{\alpha}{2}\frac{ y_0^{(1)}}{(1-z_0-\theta y_0^{(2)})}
\end{align}

\section{One-body calculations: 2RSB}\label{1body2rsb}
\begin{align}
&\mathcal{A}_N (0, \bm r - \bm \dot{r} t)= \frac{1}{N\theta_1} \mathbb{E}_0  \mathbb{E}_1 \left[ \log \mathbb{E}_2 \left( \mathbb{E}_3 \mathcal Z_3^{\theta_2}\right)^{\frac{\theta_1}{\theta_2}} \right] = \notag \\
&=\frac{1}{N\theta_1} \mathbb{E}_0 \mathbb{E}_1 \left\{ \log \mathbb{E}_2 \left\{ \mathbb{E}_3 \left[ \sums \int D\tau \exp \left( \beta ( w_0 \sumi \xi_i \si + \sumanew \sqrt{x_0^{(a)}} \sumi h_i^{(a)} \si + \right. \right. \right. \right.  \notag \\
& \left. \left. \left. \left. + \sumanew \sqrt{y_0^{(a)}} \sum_{mu=1}^P J_\mu^{(a)} \tau_\mu + z_0 \sum_\mu \frac{\tau_\mu^2}{2} ) \right) \right]^{\theta_2}\right\}^{\frac{\theta_1}{\theta_2}} \right\} = B_1+B_2
\end{align}
In the last passage we write the initial condition as the sum of two terms, one dependent on $\bm \sigma$ and the other on $\bm \tau$. 

\begin{align}
B_1 &= \frac{1}{N\theta_1} \mathbb{E}_1 \left\{ \log \mathbb{E}_2 \left[ \mathbb{E}_3 \left( \sums \exp \left( \beta ( w_0 \sumi \xi_i \si + \sumanew \sqrt{x_0^{(a)}} \sumi h_i^{(a)} \si )\right)\right)^{\theta_2} \right]^{\frac{\theta_1}{\theta_2}}\right\} = \notag \\ 
&= \frac{1}{N\theta_1}\mathbb{E}_1 \left\{ \log \mathbb{E}_2 \left[ \mathbb{E}_3 \left( \prod_i 2\cosh \left( \beta ( w_0 + \sumanew \sqrt{x_0^{(a)}} h_i^{(a)})\right) \right)^{\theta_2} \right]^{\frac{\theta_1}{\theta_2}} \right\} = \notag \\
&= \log 2 + \frac{1}{\theta_1}\mathbb{E}_1 \left\{ \log \mathbb{E}_2 \left[ \mathbb{E}_3 \left( \cosh \left( \beta ( w_0 + \sumanew \sqrt{x_0^{(a)}} h^{(a)})\right) \right)^{\theta_2} \right]^{\frac{\theta_1}{\theta_1}} \right\}
\end{align}

\begin{align}
B_2 =& \frac{1}{N\theta_1} \mathbb{E}_1 \left\{ \log \mathbb{E}_2 \left[ \mathbb{E}_3 \left( \int D \tau \exp \left( \beta(\sumanew \sqrt{y_0^{(a)}} \sum_{mu=1}^P J_\mu^{(a)} \tau_\mu + z_0 \sum_\mu \frac{\tau_\mu^2}{2} )\right)\right)^{\theta_2} \right]^{\frac{\theta_1}{\theta_2}}\right\}
\label{eq:1body2RSB} 
\end{align}

We compute the integral separately and then we reinsert it into (\ref{eq:1body2RSB})
\footnotesize
\begin{align}
\int D \tau \exp \left[ \beta \left(\sumanew \sqrt{y_0^{(a)}} \sum_{mu=1}^P J_\mu^{(a)} \tau_\mu + z_0 \sum_\mu \frac{\tau_\mu^2}{2} \right)\right] = \prod_{\mu=1}^P \sqrt{\frac{1}{\beta(1-z_0)}}\exp \left [ \frac{\beta \left( \sumanew \sqrt{y_0^{(a)}} J_\mu^{(a)} \right)^2}{2(1-z_0)} \right]
\end{align}
\normalsize
\begin{align}
B_2 =& -\frac{\alpha}{2\theta_1}\log(\beta(1-z_0)) + \frac{1}{N\theta_1} \mathbb{E}_1 \left\{ \mathbb{E}_2 \left[ \mathbb{E}_3 \left( \prod_{\mu=1}^P \exp \left( \frac{\beta \left( \sumanew \sqrt{y_0^{(a)}} J_\mu^{(a)} \right)^2}{2(1-z_0)} \right)\right)^{\theta_2} \right]^{\frac{\theta_1}{\theta_2}}\right\}
\label{eq:1body2RSB_2}
\end{align}

We compute the average over $J_\mu^{(3)}$ separately and then we reinsert it into (\ref{eq:1body2RSB_2}):
\footnotesize
\begin{align}
\int DJ_\mu^{(3)}&\exp\left(  \frac{\beta \theta_2 \left( \sumanew \sqrt{y_0^{(a)}} J_\mu^{(a)} \right)^2}{2(1-z_0)} \right) = \sqrt{\frac{1-z_0}{1-z_0-\beta \theta_2 y_0^{(3)}}} \exp\left( \frac{\beta \theta_2}{2(1-z_0-\beta \theta_2 y_0^{(3)})} (\sqrt{y_0^{(1)}} J_\mu^{(1)} + \sqrt{y_0^{(2)}} J_\mu^{(2)})^2\right)
\end{align}
\normalsize
\begin{align}
B_2 =& -\frac{\alpha}{2\theta_1}\log[\beta(1-z_0)] + \frac{\alpha}{2\theta_1}\log \left( \frac{1-z_0}{1-z_0-\beta \theta_2 y_0^{(3)}}\right) + \notag \\
&+\frac{1}{N\theta_1}\mathbb{E}_1 \left\{ \log \mathbb{E}_2 \left[ \prod_\mu \exp \left(\frac{\beta \theta_2 \left(\sqrt{y_0^{(1)}} J_\mu^{(1)} + \sqrt{y_0^{(2)}} J_\mu^{(2)}\right)^2}{2(1-z_0 - \beta \theta_2 y_0^{(3)})} \right)\right]^{\frac{\theta_1}{\theta_2}}\right\}
\label{eq:1body2RSB_3}
\end{align}

Now we compute the average over $J_\mu^{(2)}$ separately and then we reinsert in (\ref{eq:1body2RSB_3}):

\begin{align}
\int DJ_\mu^{(2)} &\exp \left(\frac{\beta \theta_1 \left(\sqrt{y_0^{(1)}} J_\mu^{(1)} + \sqrt{y_0^{(2)}} J_\mu^{(2)}\right)^2}{2(1-z_0 - \beta \theta_2 y_0^{(3)})} \right)=\notag \\
&=\sqrt{\frac{1-z_0-\beta \theta_2 y_0^{(3)}}{1-z_0-\beta \theta_2 y_0^{(3)} - \beta \theta_1 y_0^{(2)}}} \exp \left( \frac{\beta \theta_1 y_0^{(1)} (J_\mu^{(1)})^2 }{2(1-z_0-\beta \theta_2 y_0^{(3)} - \beta \theta_1 y_0^{(2)})}\right)
\end{align}

\begin{align}
B_2 =& -\frac{\alpha}{2\theta_1}\log[\beta(1-z_0)] + \frac{\alpha}{2\theta_1}\log \left( \frac{1-z_0}{1-z_0-\beta \theta_2 y_0^{(3)}}\right) + \frac{\alpha}{2}\log \left( \frac{1-z_0-\beta \theta_2 y_0^{(3)}}{1-z_0-\beta \theta_2 y_0^{(3)} - \beta \theta_1 y_0^{(2)}}\right)+\notag \\
&+\frac{1}{N\theta_1}\mathbb{E}_1 \left\{ \log  \prod_\mu \exp \left(\frac{\beta \theta_1 y_0^{(1)} (J_\mu^{(1)})^2}{2(1-z_0 - \beta \theta_2 y_0^{(3)} - \beta \theta_1y_0^{(2)}} \right)\right\} = \notag \\
=&-\frac{\alpha}{2\theta_1}\log[\beta(1-z_0)] + \frac{\alpha}{2\theta_1}\log \left( \frac{1-z_0}{1-z_0-\beta \theta_2 y_0^{(3)}}\right) + \frac{\alpha}{2}\log \left( \frac{1-z_0-\beta \theta_2 y_0^{(3)}}{1-z_0-\beta \theta_2 y_0^{(3)} - \beta \theta_1 y_0^{(2)})}\right)+\notag \\
&+ \frac{\alpha}{\theta_1} \int D J_\mu^{(1)} \left( \frac{\beta \theta_1 y_0^{(1)} (J_\mu^{(1)})^2}{2(1-z_0 - \beta\theta_2 y_0^{(3)} - \beta \theta_1 y_0^{(2)})} \right) = \notag \\ 
=&-\frac{\alpha}{2\theta_1}\log[\beta(1-z_0)] + \frac{\alpha}{2\theta_1}\log \left( \frac{1-z_0}{1-z_0-\beta \theta_2 y_0^{(3)}}\right) + \frac{\alpha}{2}\log \left( \frac{1-z_0-\beta \theta_2 y_0^{(3)}}{1-z_0-\beta \theta_2 y_0^{(3)} - \beta \theta_1 y_0^{(2)}}\right) + \notag \\ 
&+\frac{\alpha \beta y_0^{(1)}}{2(1-z_0 - \beta \theta_2 y_0^{(3)}-\beta \theta_1 y_0^{(2)})}
\label{eq:1body2RSB_4}
\end{align}

Thus, we can rearrange the two terms together and we have 
	\begin{align}
	\label{eqn:A0finApp2RSB}
	&\mathcal A_0(0, \bm r_0)= \log 2 + \frac{1}{\theta_1}\mathbb{E}_1 \left\{ \log \mathbb{E}_2 \left[ \mathbb{E}_3 \left( \cosh \left( \beta ( w_0 + \sumanew \sqrt{x_0^{(a)}} h^{(a)})\right) \right)^{\theta_2} \right]^{\frac{\theta_1}{\theta_1}} \right\}-\notag \\ 
	&-\frac{\alpha}{2\theta_1}\log(\beta(1-z_0)) + \frac{\alpha}{2\theta_1}\log \left( \frac{1-z_0}{1-z_0-\beta \theta_2 y_0^{(3)}}\right) + \frac{\alpha}{2}\log \left( \frac{1-z_0-\beta \theta_2 y_0^{(3)}}{1-z_0-\beta \theta_2 y_0^{(3)} - \beta \theta_1 y_0^{(2)}}\right) + \notag \\ 
&+\frac{\alpha \beta y_0^{(1)}}{2(1-z_0 - \beta \theta_2 y_0^{(3)}-\beta \theta_1 y_0^{(2)})}
	\end{align}

\section*{Acknowledgments}
The authors are grateful to Francesco Alemanno for several fruitful discussions. \\
EA acknowledges Sapienza University of Rome for financial support (Progetto Ateneo RG11715C7CC31E3D).
AB acknowledges Unisalento and INFN for financial support.


\begin{thebibliography}{99}

\bibitem{Agliari-Barattolo} E. Agliari, A. Barra, C. Longo, D. Tantari, {\em Neural Networks retrieving binary patterns in a sea of real ones}, J. Stat. Phys. \textbf{168}, 1085, (2017).

\bibitem{ABT} E. Agliari, A. Barra, B. Tirozzi, {\em  Free energies of Boltzmann Machines: self-averaging, annealed and replica symmetric approximations in the thermodynamic limit}, J. Stat. Mech. 033301 (2019).

\bibitem{Marullo} E. Agliari, A. Fachechi, C. Marullo, {\em The ``relativistic'' Hopfield network with correlated patterns}, submitted (2020).

\bibitem{AMT-JSP2019}  E. Agliari, D. Migliozzi, and D. Tantari, {\em Non-convex multi-species Hopfield models}, J. Stat. Phys. \textbf{172}(5):1247, (2018).

\bibitem{Alemannation1} E. Agliari, F. Alemanno, A. Barra, A. Fachechi, {\em Dreaming neural networks: rigorous results},
J. Stat. Mech. 083503 (2019).

\bibitem{biologia} C. Angermueller, et al., {\em Deep learning for computational biology}, Molec. Sys. Biol. \textbf{12}, 7, (2016).

\bibitem{AABF-NN2020}
E. Agliari, F. Alemanno, A. Barra, A. Fachechi, {\em Generalized Guerra's interpolating techniques for dense associative memories}, Neur. Nets. in press (2020).

\bibitem{Agliari-Dantoni} E. Agliari, et al., {\em Parallel retrieval of correlated patterns: From Hopfield networks to Boltzmann machines}, Neural Networks \textbf{38}, 52, (2013).

\bibitem{Martino1} F. Alemanno, et al., {\em Neural networks with redundant representations: detecting the undetectable}, \textbf{124}, 028301, (2020).

\bibitem{Martino2} F. Alemanno, et al., {\em Interpolating between boolean and extremely high noisy patterns through minimal dense associative networks},  in press DOI: $10.1088/1751-8121/ab6943$, (2020).

\bibitem{Amit} D.J. Amit, {\em Modeling brain functions}, Cambridge Univ. Press (1989).

\bibitem{jean1} J. Barbier, et al., {\em Mutual information for symmetric rank-one matrix estimation: A proof of the replica formula}, Neural Inf. Proc. Sys. (NIPS), Barcelona, (2016).

\bibitem{jean2} J. Barbier, N. Macris, {\em The adaptive interpolation method: a simple scheme to prove replica formulas in Bayesian inference}, Prob. Th. and Rel. Fiel. \textbf{174}(3-4), 1133, (2019).

\bibitem{BGDiBiasio} A. Barra, A. Di Biasio, F. Guerra, {\em Replica symmetry breaking in mean field spin glasses trough Hamilton-Jacobi technique},
JSTAT P09006, (2010).

\bibitem{Barra-JSP2010} A. Barra, G. Genovese, F. Guerra, {\em The replica symmetric approximation of the analogical neural network}, J. Stat. Phys. \textbf{140}(4):784, (2010).

\bibitem{bipartiti} A. Barra, G. Genovese, F. Guerra, {\em Equilibrium statistical mechanics  of bipartite spin systems}, J. Phys. A \textbf{44}, 245002, (2011).

\bibitem{Albert1} A. Barra, M. Beccaria, A. Fachechi,{\em A new mechanical approach to handle generalized Hopfield neural networks}, Neural Networks (2018).

\bibitem{ZiqquratBarra} A. Barra, P. Contucci. E. Mingione, D. Tantari, {\em Multi-Species mean-field spin-glasses: Rigorous results}, Ann. H. Poincar\`e   \textbf{16}(3), 691, (2015).

\bibitem{BarraEquivalenceRBMeAHN} A. Barra, A. Bernacchia, E. Santucci, P. Contucci, {\em On the equivalence among Hopfield neural networks and restricted Boltzman machines}, Neural Networks \textbf{34}, 1-9, (2012).

\bibitem{Gauss-1} A. Barra, G. Genovese, F. Guerra, D. Tantari, {\em About a solvable mean field model of a Gaussian spin glass}, J. Phys. A \textbf{47}(15), 155002, (2014).

\bibitem{Barra-RBMsPriors1} A. Barra, G. Genovese, P. Sollich, D. Tantari, {\em Phase transitions of Restricted Boltzmann Machines with generic priors}, Phys. Rev. E \textbf{96}, 042156, (2017).

\bibitem{Barra-RBMsPriors2} A. Barra, G. Genovese, P. Sollich, D. Tantari, {\em Phase Diagram of Restricted Boltzmann Machines $\&$ Generalized Hopfield Models}, Phys. Rev. E \textbf{97}, 022310, (2018).

\bibitem{Gauss-2} G. Ben Arous, A. Dembo, A. Guionnet, {\em Aging of spherical spin glasses}, Prob. Theor. Related Fields \textbf{120}, 1, (2001).

\bibitem{Bovier1} A. Bovier, V. Gayrard, {\em Hopfield models as generalized random mean field models}, Mathematical aspects of spin glasses and neural networks, 3-89, Birkhauser, Boston (1998).

\bibitem{Bovier2} A. Bovier, V. Gayrard, P. Picco, {\em Gibbs states of the Hopfield model in the regime of perfect memory}, Prob. Theor. $\&$ Rel. Fields  \textbf{100}(3):329, (1994).

\bibitem{Bovier3} A. Bovier, V. Gayrard, P. Picco, {\em Gibbs states of the Hopfield model with extensively many patterns}, J. Stat. Phys. \textbf{79}(1-2):395, (1995).

\bibitem{Univ1} P. Carmona, Y. Hu, {\em Universality in Sherrington-Kirkpatrick's spin glass model}, Ann. Henri Poincar\`e \textbf{42}, 2, (2006).

\bibitem{Ton1} A.C.C. Coolen,  J. Van Mourik, {\em Cluster derivation of the Parisi scheme for disordered systems}, AIP Conference Proceedings \textbf{553}, 1, APS press, (2001).

\bibitem{Coolen} A.C.C. Coolen, R. Kuhn, P. Sollich, {\em Theory of neural information processing systems}, Oxford Press (2005).

\bibitem{Crisanti} A. Crisanti, D.J. Amit, H. Gutfreund, {\em Saturation Level of the Hopfield Model for Neural Network}, Europhys. Lett. 2(4), 337-341 (1986).

\bibitem{Viktor} V. Dotsenko, {\em An introduction to the theory of spin glasses and neural networks}, World Scientific, (1995).

\bibitem{Dotsenko2} V. Dotsenko, B. Tirozzi, {\em Replica symmetry breaking in neural networks with modified pseudo-inverse interactions}, J. Phys. A \textbf{24}:5163-5180, (1991).

\bibitem{DL1} Y. Le Cun, Y. Bengio, G. Hinton, {\em Deep learning}, Nature \textbf{521}:436-444, (2015).

\bibitem{angel-learning} A. Engel, C. Van den Broeck, {\em Statistical mechanics of learning}, Cambridge University Press (2001).

\bibitem{Albert2} A. Fachechi, E. Agliari, A. Barra, {\em Dreaming neural networks: forgetting spurious memories and reinforcing pure ones}, Neural Net. \textbf{112}, 24, (2019).

\bibitem{Genovese} G. Genovese, {\em Universality in bipartite mean field spin glasses}, J. Math. Phys. \textbf{53}(12):123304, (2012).

\bibitem{GuerraSum} F. Guerra, {\em Sum rules for the free energy in the mean field spin glass model}, Fiel. Inst. Comm. \textbf{30}, 11, (2001).

\bibitem{Guerra} F. Guerra, {\em Broken replica symmetry bounds in the mean field spin glass model}, Comm. Math. Phys. \textbf{233}(1), 1, (2003).

\bibitem{GuerraTon}  F. Guerra, F.L. Toninelli, {\em The thermodynamic limit in mean field spin glass models}, Comm. Math. Phys. \textbf{230}(1), 71-79, (2002).

\bibitem{finanza} J.B. Heaton, N.G. Polson, J. Hendrik Witte, {\em Deep learning for finance: deep portfolios}, Appl. Stoc. Mod. Busin. Ind. \textbf{33}(1), 3, (2017).

\bibitem{medico} G. Litjens, et al., {\em Deep learning as a tool for increased accuracy and efficiency of histopathological diagnosis}, Sci. Rep. \textbf{6}, 26286, (2016).

\bibitem{Mezard} M. Mezard, {\em Mean-field message-passing equations in the Hopfield model and its generalizations}, Phys. Rev. E \textbf{95}(2), 022117 (2017).

\bibitem{MPV} M. M\'ezard, G. Parisi, M.A. Virasoro {\em Spin Glass Theory and Beyond}, World Scientific, Singapore (1987).

\bibitem{Murrat1} J.C. Mourrat, {\em Parisi's formula is a Hamilton-Jacobi equation in Wasserstein space}, arXiv preprint arXiv:1906.08471, (2019).

\bibitem{MurratPanchenko} J.C. Mourrat, D. Panchenko, {\em Extending the Parisi formula along a Hamilton-Jacobi equation}, Electron. J. Probab. \textbf{25}(23), 1, (2020).

\bibitem{DmitryBook} D. Panchenko, {\em The Sherrington-Kirkpatrick model}, Springer Science $\&$ Business Media (2013).

\bibitem{ZiqquratPanchenko} D. Panchenko, {\em The free energy in a multi-species Sherrington–Kirkpatrick model}, Ann. of Prob. \textbf{43}(6), 3494, (2015).

\bibitem{Tirozzi} L. Pastur, M. Shcherbina, B. Tirozzi, {\em On the replica symmetric equations for the Hopfield model}, J. Math. Phys. \textbf{40}(8): 3930, (1999).

\bibitem{Pastur} L. Pastur, M. Shcherbina, B. Tirozzi, {\em The replica-symmetric solution without replica trick for the Hopfield model}, J. Stat. Phys. \textbf{74}(5-6):1161, (1994).

\bibitem{Hinton1} R. Salakhutdinov, G. Hinton, {\em Deep Boltzmann machines}, Artificial Intelligence and Statistics (2009).

\bibitem{sompo-learning} H.S. Seung, H. Sompolinsky, N. Tishby, {\em Statistical mechanics of learning from examples}, Phys. Rev. A \textbf{45}(8):6056, (1992).

\bibitem{Kuhn} H. Steffan, R. Kuhn, {\em Replica symmetry breaking in attractor neural network models}, Z. Phys. B \textbf{95}, 249, (1994).

\bibitem{Tala1} M. Talagrand, {\em Rigorous results for the Hopfield model with many patterns}, Prob. Theor. $\&$ Rel. Fiel. \textbf{110}(2):177, (1998).

\bibitem{Tala2} M. Talagrand, {\em Exponential inequalities and convergence of moments in the replica-symmetric regime of the Hopfield model}, Ann. Prob. 1393-1469, (2000).

\bibitem{TalaParisi} M. Talagrand, {\em The parisi formula}, Annals of Mathematics, 221-263, (2006).

\bibitem{Monasson} J. Tubiana, R. Monasson, {\em Emergence of compositional representations in restricted Boltzmann machines}, Phys. Rev. Lett. \textbf{118}(13), 138301, (2017).

\bibitem{fisica} P.T. Komiske, E.M. Metodiev, M.D. Schwartz, {\em Deep learning in color: towards automated quark-gluon jet discrimination}, J. High En. Phys. \textbf{2017}(1), 110, (2017).

\bibitem{Ton2} J. Van Mourik, A.C.C. Coolen, {\em Cluster derivation of Parisi's RSB solution for disordered systems}, J. Phys. A \textbf{34}(10),  L111, (2001).

\bibitem{Lenka} L. Zdeborova, F. Krzakala, {\em Statistical physics of inference: Thresholds and algorithms}, Adv. in Phys. \textbf{65}(5):453-552, (2016).

\end{thebibliography}
\end{document}